\documentclass[a4paper,reqno,10pt]{amsart}

\usepackage{epsf,graphicx,epsfig}
\usepackage{amscd}
\usepackage[all,cmtip]{xy}
\usepackage{amsmath,latexsym,amssymb,amsthm}
\usepackage[nospace,noadjust]{cite}
\usepackage{textcomp}
\usepackage{dsfont}
\usepackage{setspace,cite}
\usepackage{lscape,fancyhdr,fancybox}
\usepackage{stmaryrd}
\usepackage{tikz}
\usepackage{cancel}
\usetikzlibrary{shapes,arrows,decorations.markings}
\usepackage{graphicx}

\usepackage{color}
\usepackage{url}
\usepackage{enumerate}
\usepackage[mathscr]{euscript}
  \theoremstyle{plain}

\swapnumbers
    \newtheorem{thm}{Theorem}[section]
    \newtheorem{prop}[thm]{Proposition}
     
   \newtheorem{lemma}[thm]{Lemma}

    \newtheorem{subsec}[thm]{}
\theoremstyle{definition}
    \newtheorem{defn}[thm]{Definition}
        \newtheorem{remark}[thm]{Remark}
    \newtheorem{exam}[thm]{Example}

\theoremstyle{remark}

\title{}
\author{}
\date{}
\usepackage{amssymb}

\usepackage{hyperref}

\hypersetup{
	colorlinks,
	citecolor=blue,
	filecolor=black,
	linkcolor=blue,
	urlcolor=black
}

\begin{document}

\title[]{From $n$-Leibniz algebras and linear $n$-racks to the solutions of the (higher analogue of) Yang-Baxter equation}

\author{Apurba Das}
\address{Department of Mathematics,
Indian Institute of Technology, Kharagpur 721302, West Bengal, India.}
\email{apurbadas348@gmail.com, apurbadas348@maths.iitkgp.ac.in}

\author{Suman Majhi}
\address{Department of Mathematics, Indian Institute of Technology, Kharagpur 721302, West Bengal, India.}
\email{majhisuman693@gmail.com, suman693@kgpian.iitkgp.ac.in}

\begin{abstract}
In this paper, we first demonstrate that a finite-dimensional $n$-Leibniz algebra naturally gives rise to an $n$-rack structure on the underlying vector space. Given any $n$-Leibniz algebra, we also construct two Yang-Baxter operators on suitable vector spaces and connect them by a homomorphism. Next, we introduce linear $n$-racks as the coalgebraic version of $n$-racks and show that a cocommutative linear $n$-rack yields a linear rack structure and hence a Yang-Baxter operator. An $n$-Leibniz algebra canonically gives rise to a cocommutative linear $n$-rack and thus produces a Yang-Baxter operator.

In the last part, following the well-known close connections among Leibniz algebras, (linear) racks and Yang-Baxter operators, we consider a higher-ary generalization of Yang-Baxter operators (called $n$-Yang-Baxter operators). In particular, we show that $n$-Leibniz algebras and cocommutative linear $n$-racks naturally provide $n$-Yang-Baxter operators. Finally, we consider a set-theoretical variant of $n$-Yang-Baxter operators and propose some problems. 
\end{abstract}

\maketitle

%\curraddr{}
%\email{}

%\subjclass[2010]{}
%\keywords{}

%\begin{center}


%\end{center}

%Averaging operators, Racks, Set-theoretic Yang-Baxter equation.

\medskip

% {2020 MSC classification:} 16D20, 16W99, 16E40, 16S80.

% {Keywords:} Averaging algebras of nonzero weight, Triassociative algebras, $L_\infty$-algebras, Cohomology, Homotopy algebras.

%\medskip

%\noindent {\sf Date of resubmission:} July 26, 2021.

\thispagestyle{empty}

\tableofcontents

%\vspace{0.2cm}

\medskip

\section{Introduction}\label{sec1}
\subsection{Yang-Baxter equation}
The Yang-Baxter equation first appeared in a work of Yang on statistical mechanics \cite{yang} and reappeared in Baxter's work on the eight-vertex model \cite{baxter}. It is one of the fundamental equations in mathematical physics that has been used in numerous studies on integrable systems, quantum groups, knot theory, quantum field theory, braided categories and some other areas of mathematics and mathematical physics. Recall that a linear map $R: V \otimes V\rightarrow  V\otimes V$ is a Yang-Baxter operator on a vector space $V$ if it is invertible and satisfies the Yang-Baxter equation:
\begin{align}\label{ybe-first}
    (R \otimes \mathrm{Id}) (\mathrm{Id} \otimes R) (R \otimes \mathrm{Id}) = (\mathrm{Id} \otimes R) (R \otimes \mathrm{Id}) (\mathrm{Id} \otimes R).
\end{align}
If $X$ is a basis for the vector space $V$, then a bijective map $r : X \times X \rightarrow X \times X$ satisfying
\begin{align}\label{set-ybe}
    (r \times \mathrm{Id}) (\mathrm{Id} \times r) (r \times \mathrm{Id}) =  (\mathrm{Id} \times r)  (r \times \mathrm{Id}) (\mathrm{Id} \times r)
\end{align}
induces a Yang-Baxter operator on $V$. In this case, one says that $r$ is a set-theoretical solution of the Yang-Baxter equation on $X$ (or simply a set-theoretical solution on $X$). In \cite{drinfeld}, Drinfeld posed the question of finding Yang-Baxter operators arising from set-theoretic solutions. Since then, the Yang-Baxter equation (\ref{ybe-first}) and its set-theoretical variant (\ref{set-ybe}) have been extensively studied in connection with several other algebraic structures. In particular, the works of Etingof-Schedler-Soloviev \cite{etingof} and Lu-Yau-Zhu \cite{lu} discussed algebraic and geometrical interpretations of non-degenerate involutive set-theoretical solutions, and introduced several structures associated with them.

\medskip

In \cite{rump0, rump}, Rump introduced cycle sets and braces as generalizations of radical rings to obtain non-degenerate involutive set-theoretical solutions. As the noncommutative analogue of braces, Guarnieri and Vendramin \cite{guar} considered skew braces that yield nondegenerate (not necessarily involutive) set-theoretical solutions. On the other hand, a rack is an algebraic structure given by a nonempty set equipped with a binary operation satisfying axioms analogous to the second and third Reidemeister moves in knot theory \cite{etingof, jack}. A group with the conjugation operation gives a first example of a rack. Knot theorists have deeply studied racks to construct invariants for knots and links. Any rack naturally provides a set-theoretical solution. Racks are intimately related to Leibniz algebras \cite{bloh,loday} (noncommutative analogue of Lie algebras) as the tangent space at a neutral element of a Lie rack carries a Leibniz algebra structure \cite{kinyon}. Kinyon \cite{kinyon} showed that one may also construct a rack from a finite-dimensional Leibniz algebra. It has been shown by Lebed \cite{lebed2} that a central Leibniz algebra yields a Yang-Baxter operator. Linear racks are the coalgebraic version of racks and are shown to be useful in constructing Yang-Baxter operators \cite{kram,carter,lebed,lebed2}. In \cite{lebed}, Lebed provides a functorial construction from the category of Leibniz algebras to the category of linear racks.

\subsection{$n$-ary structures and the Yang-Baxter equation} The notion of $n$-Leibniz algebras considered by Casas, Loday and Pirashvili \cite{casas1} as the $n$-ary generalization of Leibniz algebras. They first appeared in the work of Nambu \cite{nambu} while generalizing classical mechanics to $n$-ary set-up. It has been shown that the set of fundamental elements of an $n$-Leibniz algebra inherits a Leibniz algebra structure. Later, Biyogmam \cite{biyog} considered the notion of $n$-racks and showed that the tangent space at a neutral element of a Lie $n$-rack inherits an $n$-Leibniz algebra structure. Any $n$-rack induces a rack structure on the $n-1$ cartesian product of the underlying set. Recently, Xu and Sheng \cite{xu-sheng} showed that a finite-dimensional $3$-Leibniz algebra always gives rise to a $3$-rack structure on the underlying vector space. In our first result, we show that a finite-dimensional $n$-Leibniz algebra naturally gives rise to an $n$-rack structure in a functorial way that generalizes the results of Kinyon \cite{kinyon} for $n=2$ and of Xu-Sheng \cite{xu-sheng} for $n=3$. Next, we focus on central $n$-Leibniz algebras and show that they produce Yang-Baxter operators. Subsequently, we construct two Yang-Baxter operators from a given $n$-Leibniz algebra (need not be central) and connect them by a suitable homomorphism. Our results generalize the Yang-Baxter operators obtained in \cite{xu-sheng} and \cite{abramov}.

\medskip

Next, we introduce linear $n$-racks as the coalgebraic version of $n$-racks. Given a linear rack, one can construct a linear $ n$-rack on the same underlying coalgebra. On the other side, a cocommutative linear $n$-rack naturally induces a linear rack structure on the $n-1$ tensor power of the underlying coalgebra (which is the coalgebraic version of the construction of a rack from an $n$-rack). Combining this result with the construction of a Yang-Baxter operator from a linear rack, we obtain a Yang-Baxter operator from a given cocommutative linear $n$-rack. Next, we show that an $n$-Leibniz algebra $\mathcal{L}$ naturally gives rise to a cocommutative linear $n$-rack structure on the direct sum ${\bf k} \oplus \mathcal{L}$. Hence, applying our previous result, one obtains a Yang-Baxter operator from an $n$-Leibniz algebra. This generalizes the well-known Yang-Baxter operators arising from Leibniz algebras \cite{lebed} and $3$-Leibniz algebras \cite{xu-sheng} to the context of $n$-Leibniz algebras.

%\subsection{Leibniz algebras, racks, linear racks and the Yang-Baxter equation}

%\subsection{$n$-Leibniz algebras, $n$-racks, linear $n$-racks and the Yang-Baxter equation}

\subsection{A higher analogue to the Yang-Baxter equation} 
In the final part of this paper, we introduce a higher-analogue of the Yang-Baxter equation (which we call the $n$-Yang-Baxter equation) that is closely related to $n$-Leibniz algebras and (linear) $n$-racks, in the same way the Yang-Baxter equation is related to Leibniz algebras and (linear) racks. An invertible solution of the $n$-Yang-Baxter equation is called an $n$-Yang-Baxter operator. As mentioned above, any $n$-Leibniz algebras and cocommutative linear $n$-racks naturally give rise to $n$-Yang-Baxter operators. Next, we show that a Yang-Baxter operator on a vector space $V$ naturally induces an $n$-Yang-Baxter operator on the same vector space. On the other hand, an $n$-Yang-Baxter operator on a vector space $V$ gives rise to a Yang-Baxter operator on the tensor product $V^{\otimes (n-1)}$. Finally, we also consider a set-theoretical variant of the $n$-Yang-Baxter equation, whose bijective solutions are called set-theoretical $n$-solutions. We observe that any $n$-rack provides a set-theoretical $n$-solution. Finally, we end with some interesting questions.

\subsection{Organization of the paper} The paper is organized as follows. In Section \ref{sec2}, we mainly show that a finite-dimensional $n$-Leibniz algebra gives rise to an $n$-rack structure in a functorial way. In Section \ref{sec3}, we first construct a Yang-Baxter operator from a central $n$-Leibniz algebra. As a consequence, we obtain two Yang-Baxter operators from an arbitrary $n$-Leibniz algebra and connect them by a suitable homomorphism. In Section \ref{sec4}, we introduce the notion of a linear $n$-rack and show that a cocommutative linear $n$-rack induces a linear rack structure and thus provides a Yang-Baxter operator. Further, we show that an $n$-Leibniz algebra naturally gives rise to a linear $n$-rack structure. Finally, we consider a higher-analogue of the Yang-Baxter equation and its set-theoretical variant in Section \ref{sec5}. We find examples of $n$-Yang-Baxter operators arising from $n$-Leibniz algebras and cocommutative linear $n$-racks.

\medskip

All vector spaces, linear maps, multilinear maps and tensor products are over a field {\bf k} of characteristics $0$.

%Let $(\mathcal{L}, \{ ~, ~ \})$ be a finite-dimensional Leibniz algebra. Define a binary operation $\triangleleft : \mathcal{L} \times \mathcal{L} \rightarrow \mathcal{L}$ by 
%\begin{align}
%    x \triangleleft y = \sum_{k=0}^\infty \frac{(\mathrm{ad}_y)^k (x)}{k!}, \text{ for } x, y \in \mathcal{L}.
%\end{align}
%Then it has been shown by Kinyon \cite{kinyon} that $(\mathcal{L}, \triangleleft)$ is a rack. 

\section{From {\em n}-Leibniz algebras to {\em n}-racks}\label{sec2}
In this section, we first recall some necessary background on $n$-Leibniz algebras \cite{casas1} and $n$-racks \cite{biyog}. Among others, we show that an $n$-Leibniz algebra naturally gives rise to an $n$-rack structure on the underlying vector space.

\begin{defn}\label{defn-n-leib}
    (i) An {\bf $n$-Leibniz algebra} is a vector space $\mathcal{L}$ equipped with an $n$-linear operation $[-, \ldots, -] : \underbrace{\mathcal{L} \times \cdots \times \mathcal{L} }_{n \text{ copies}} \rightarrow \mathcal{L}$ that satisfies the {\em fundamental identity}:
    \begin{align}\label{fun-identity}
        [[ x_1, \ldots, x_n],  y_1, \ldots, y_{n-1}] = \sum_{i=1}^n ~ [x_1, \ldots, x_{i-1}, [x_i, y_1, \ldots, y_{n-1}], x_{i+1}, \ldots, x_n],
    \end{align}
    for all $x_1, \ldots, x_n, y_1, \ldots, y_{n-1} \in \mathcal{L}$. An $n$-Leibniz algebra as above is denoted by the pair $(\mathcal{L}, [-, \ldots, -])$.

    \medskip

    (ii) Let $(\mathcal{L}, [-, \ldots, -])$ and $(\mathcal{L}', [-, \ldots, -]')$ be two $n$-Leibniz algebras. A linear map $\varphi: \mathcal{L} \rightarrow \mathcal{L}'$ is said to be a {\bf homomorphism} between them if it is a bracket preserving map, i.e.,
    \begin{align*}
        \varphi ([x_1, \ldots, x_n]) = [\varphi (x_1), \ldots, \varphi (x_n )]', \text{ for all } x_1, \ldots, x_n \in \mathcal{L}. 
    \end{align*}
    \end{defn}

Note that the fundamental identity (\ref{fun-identity}) in the above definition is equivalent to saying that the linear maps $[-, y_1, \ldots, y_{n-1}]: \mathcal{L} \rightarrow \mathcal{L}$ by fixing the right $(n-1)$ coordinates are derivations for the bracket on $\mathcal{L}$. For this reason, the $n$-Leibniz algebras considered above are often called right $n$-Leibniz algebras. Inspired by this, one may define left $n$-Leibniz algebras. Then it is easy to see that $(\mathcal{L}, [-, \ldots, -])$ is a right $n$-Leibniz algebra if and only if $(\mathcal{L}, [-, \ldots, -]^\mathrm{op})$ is a left $n$-Leibniz algebra, where $[x_1, \ldots, x_n]^\mathrm{op} := [x_n, x_{n-1}, \ldots, x_1]$. In this paper, by an $n$-Leibniz algebra, we shall always mean a right $n$-Leibniz algebra as of Definition \ref{defn-n-leib}. However, all the results of the present paper can be generalized to left $n$-Leibniz algebras by corresponding modifications.

\begin{exam}\label{leib-to-n} \cite[Proposition 3.2]{casas1} Let $(\mathfrak{g}, \{ -, -\})$ be a Leibniz algebra. Then the vector space $\mathfrak{g}$ can be given an $n$-Leibniz algebra structure with the bracket operation
    \begin{align*}
        [x_1, \ldots, x_n] := \{ x_1, \{ x_2, \ldots \{ x_{n-1}, x_n \} \cdot \cdot \} \}, \text{ for } x_1, \ldots, x_n \in \mathfrak{g}.
    \end{align*}
\end{exam}

\begin{exam}
    Let $(\mathfrak{g}, \{ -, - \})$ be a Leibniz algebra and $[-, \ldots, -] : \mathfrak{g} \times \cdots \times \mathfrak{g} \rightarrow \mathfrak{g}$ be any $n$-linear operation such that for any $y \in \mathfrak{g}$, the right translation map $\{ -, y \} : \mathfrak{g} \rightarrow \mathfrak{g}$ is a derivation for the bracket $[-, \ldots, - ]$. Define an $(n+1)$-linear map $[ \! \! [ -, \ldots, - ] \! \! ] : \mathfrak{g} \times \cdots \times \mathfrak{g} \rightarrow \mathfrak{g}$ by
    \begin{align*}
        [ \! \! [ x_1, \ldots, x_{n+1} ] \! \! ] := \{ x_1, [x_2, \ldots, x_{n+1}]\}, \text{ for } x_1, \ldots, x_{n+1} \in \mathfrak{g}.
    \end{align*}
    Then $(\mathfrak{g}, [ \! \! [ -, \ldots, - ] \! \! ])$ is an $(n+1)$-Leibniz algebra.
\end{exam}

See \cite{casas1, casas3} for some other examples of $n$-Leibniz algebras. The following result \cite[Proposition 3.4]{casas1} shows that an $n$-Leibniz algebra naturally induces a Leibniz algebra structure.

\begin{prop} \label{funda-leibniz}
    Let $(\mathcal{L}, [-, \ldots, -])$ be an $n$-Leibniz algebra. Then (the space of fundamental elements) $\mathcal{L}^{\otimes (n-1)}$ carries a Leibniz algebra structure with the bracket operation 
    \begin{align}\label{funda-leib-bracket}
        \{ x_1 \otimes \cdots \otimes x_{n-1} ~ \! , ~ \! y_1 \otimes \cdots \otimes y_{n-1} \} := \sum_{i=1}^{n-1} x_1 \otimes \cdots \otimes [ x_i, y_1, \ldots, y_{n-1}] \otimes \cdots \otimes x_{n-1},
    \end{align}
    for $x_1 \otimes \cdots \otimes x_{n-1}$, $y_1 \otimes \cdots \otimes y_{n-1} \in \mathcal{L}^{\otimes (n-1)}$.
\end{prop}

\medskip

\begin{defn} \cite{biyog}
    A (right) {\bf $n$-shelf} is a pair $(X, \langle - , \ldots, - \rangle)$ consisting of a nonempty set $X$ equipped with an $n$-ary operation $\langle - , \ldots, - \rangle : \underbrace{X \times \cdots \times X}_{n \text{ copies}} \rightarrow X$ that satisfies the following distributive law:
    \begin{align}\label{n-shelf}
         \langle \langle x_1, \ldots, x_n \rangle, y_1, \ldots, y_{n-1} \rangle = \langle \langle x_1,  y_1, \ldots, y_{n-1} \rangle , \ldots, \langle x_n,  y_1, \ldots, y_{n-1} \rangle \rangle,
    \end{align}
    for all $x_1, \ldots, x_n, y_1, \ldots, y_{n-1} \in X$. An $n$-shelf $(X, \langle - , \ldots, - \rangle)$ is said to be a (right) {\bf $n$-rack} if for all $y_1, \ldots, y_{n-1} \in X$, the right translation map $\langle -,  y_1, \ldots, y_{n-1} \rangle : X \rightarrow X$ is bijective.
\end{defn}

Let $(X, \langle -, \ldots, - \rangle)$ and $(X', \langle - , \ldots, - \rangle')$ be two $n$-racks. A set-map $\varphi : X \rightarrow X'$ is said to be a homomorphism of $n$-racks if $\varphi ( \langle x_1, \ldots, x_n \rangle ) = \langle \varphi (x_1), \ldots, \varphi (x_n) \rangle'$, for all $x_1, \ldots, x_n \in X.$

\begin{exam}
Let $(G, ~\! \! \cdot ~ \! \!)$ be a group. Define an $n$-ary operation $\langle - , \ldots, - \rangle : G \times \cdots \times G \rightarrow G$ by $\langle x_1, \ldots, x_n \rangle := x_n \cdots x_2 x_1 x_2^{-1} \cdots x_n^{-1}$, for $x_1, \ldots, x_n \in G$. Then $(G, \langle - , \ldots, - \rangle)$ is an $n$-rack, called the {\em conjugation} $n$-rack.
\end{exam}

\begin{exam}\label{exam-rack-to-nrack}
    Let $(X, \triangleleft)$ be a rack. Define an $n$-ary operation $\langle - , \ldots, - \rangle : X \times \cdots \times X \rightarrow X$ by
    \begin{align*}
        \langle x_1, \ldots, x_n \rangle := ( \cdots (( x_1 \triangleleft x_2) \triangleleft x_3 ) \cdots ) \triangleleft x_n, \text{ for } x_1, \ldots, x_n \in X.
    \end{align*}
    Then $(X, \langle - , \ldots, - \rangle )$ is an $n$-rack.
\end{exam}

%\begin{exam}
%    Let $(X, \triangleleft)$ be a rack. Define an $n$-ary operation $\langle - , \ldots, - \rangle : X \times \cdots \times X \rightarrow X$ by
%    \begin{align*}
%        \langle x_1, \ldots, x_n \rangle := x_1 \triangleleft (x_2 \triangleleft \cdots (x_{n-1} \triangleleft x_n) \cdots), \text{ for } x_1, \ldots, x_n \in X.
%    \end{align*}
%    Then $(X, \langle - , \ldots, - \rangle )$ is an $n$-rack.
%\end{exam}

\begin{exam}\label{rack-n-rack}
    Let $(X, \triangleleft)$ be a rack and $\langle -, \ldots, - \rangle : X \times \cdots \times X \rightarrow X$ be any $n$-ary operation on $X$ such that
    \begin{align*}
        \langle x_1, \ldots, x_n \rangle \triangleleft y = \langle x_1 \triangleleft y, \ldots, x_n \triangleleft y  \rangle, \text{ for all } x_1, \ldots, x_n, y \in X.
    \end{align*}
    Define an $(n+1)$-ary operation $ \langle \!  \langle  -, \ldots, - \rangle \! \rangle : X \times \cdots \times X \rightarrow X$ by $ \langle \!  \langle  x_1, \ldots, x_{n+1} \rangle \! \rangle := x_1 \triangleleft \langle x_2, \ldots, x_{n+1} \rangle$. Then $(X,  \langle \!  \langle  -, \ldots, - \rangle \! \rangle)$ is an $(n+1)$-rack.
\end{exam}

\begin{remark} 
    Let $(X, \langle -, \ldots, - \rangle)$ be an $m$-rack and $(X, \langle - , \ldots, - \rangle')$ be a $n$-rack both defined on a set $X$. Suppose for each $x_1, x_2, \ldots \in X$, the right translation map $\langle - , x_1, \ldots, x_{m-1} \rangle : X \rightarrow X$ is an automorphism for the $n$-rack $(X, \langle -, \ldots, - \rangle')$ and the right translation $\langle - , x_1, \ldots, x_{n-1} \rangle' : X \rightarrow X$ is an automorphism for the $m$-rack $(X, \langle -, \ldots, - \rangle)$. Then we say that the $m$-rack structure $\langle -, \ldots, - \rangle$ and the $n$-rack structure $\langle - ,\ldots , - \rangle'$ are compatible. We define a $(m+n-1)$-ary operation $\ll -, \ldots, - \gg : X \times \cdots \times X \rightarrow X$ by
    \begin{align*}
        \ll x_1, \ldots, x_{m+n-1} \gg := \langle \langle x_1, \ldots, x_m \rangle , x_{m+1}, \ldots, x_{m+n-1} \rangle', \text{ for } x_1, \ldots, x_{m+n-1} \in X.
    \end{align*}
    Then it is easy to see that $(X, \ll -, \ldots, - \gg )$ is a $(m+n-1)$-rack.
\end{remark}

The following result shows that an $n$-rack naturally gives rise to a rack structure \cite{biyog}.

\begin{prop}\label{prop-nrack-rack}
    Let $(X, \langle - , \ldots, - \rangle )$ be an $n$-rack. Then the set $X^{\times (n-1)}$ carries a rack structure with the binary operation
    \begin{align*}
        (x_1, \ldots, x_{n-1}) \triangleleft_{\langle - , \ldots, - \rangle} (y_1, \ldots, y_{n-1}) := \big( \langle x_1, y_1, \ldots, y_{n-1} \rangle, \ldots,  \langle x_{n-1}, y_1, \ldots, y_{n-1} \rangle   \big),
    \end{align*}
    for $(x_1, \ldots, x_{n-1}), (y_1, \ldots, y_{n-1}) \in X^{\times (n-1)}.$ Moreover, the map 
    \begin{align*}
        X^{\times (n-1)} \rightarrow \mathrm{Aut} (X, \langle -, \ldots, - \rangle) ~ \text{ given by } ~  (y_1, \ldots, y_{n-1}) \mapsto \langle -, y_1, \ldots, y_{n-1} \rangle
    \end{align*}
is a rack homomorphism, where $\mathrm{Aut} (X, \langle -, \ldots, - \rangle)$ is endowed with the conjugation rack structure.
\end{prop}

\begin{remark}
Let $(X, \langle - , \ldots , - \rangle)$ be an $n$-rack. Then $r: X^{\times (n-1)} \times X^{\times (n-1)} \rightarrow  X^{\times (n-1)} \times X^{\times (n-1)}$ defined by
\begin{align*}
r (( x_1, \ldots, x_{n-1}), (y_1, \ldots, y_{n-1})) = \big(   (y_1, \ldots, y_{n-1}), ( \langle x_1, y_1, \ldots y_{n-1} \rangle, \ldots , \langle x_{n-1}, y_1, \ldots y_{n-1} \rangle)  \big),
    \end{align*}
for $(x_1, \ldots, x_{n-1}), (y_1, \ldots, y_{n-1}) \in X^{\times (n-1)}$, is a set-theoretical solution on the set $X^{\times (n-1)}$.
\end{remark}

\begin{remark}
    Let $n, k \geq 2$ and $(X, \langle -, \ldots, - \rangle)$ be an $((n-1)(k-1)+1)$-rack. Then it can be checked that the space $X^{\times (n-1)}$ equipped with the operation
    \begin{align*}
        &\langle \! \langle  (x_{1,1}, \ldots, x_{1,n-1}),  (x_{2,1}, \ldots, x_{2,n-1}), \ldots, (x_{k,1}, \ldots, x_{k,n-1})    \rangle \! \rangle \\
        & \qquad : = \big(  \langle x_{1,1}, x_{2,1}, \ldots, x_{2,n-1}, \ldots, x_{k,1}, \ldots, x_{k,n-1} \rangle , \ldots, \langle x_{1,n-1}, x_{2,1}, \ldots, x_{2,n-1}, \ldots, x_{k,1}, \ldots, x_{k,n-1} \rangle  \big)
    \end{align*}
    is a $k$-rack. In particular, when $k =2$, we recover the rack obtained in Proposition \ref{prop-nrack-rack}.
\end{remark}

\medskip

Let $(\mathcal{L}, [-, \ldots, -])$ be an $n$-Leibniz algebra. A linear map $\mathfrak{D} : \mathcal{L} \rightarrow \mathcal{L}$ is said to be a {\bf derivation} on $\mathcal{L}$ if
\begin{align}\label{deri-iden}
    \mathfrak{D} ([x_1, \ldots, x_n]) = \sum_{i=1}^n [x_1, \ldots, \mathfrak{D}(x_i), \ldots, x_n], \text{ for all } x_1, \ldots, x_n \in \mathcal{L}.
\end{align}
It follows from the fundamental identity (\ref{fun-identity}) that the right translation maps (also called adjoint maps) $\mathrm{ad}_{y_1, \ldots, y_{n-1}} : = [-, y_1, \ldots, y_{n-1}] : \mathcal{L} \rightarrow \mathcal{L}$ are derivations on the $n$-Leibniz algebra $(\mathcal{L}, [-, \ldots, -])$. All such derivations are called {\em inner derivations}.

\begin{prop}\label{prop-expo}
    Let $(\mathcal{L}, [-, \ldots, -])$ be a finite-dimensional $n$-Leibniz algebra. Then for any derivation $\mathfrak{D} : \mathcal{L} \rightarrow \mathcal{L}$, the map $\mathrm{exp} (\mathfrak{D}) = \sum_{k=0}^\infty \frac{\mathfrak{D}^k }{k ! }$ is an $n$-Leibniz algebra isomorphism on $\mathcal{L}$.
\end{prop}

\begin{proof}
Since $\mathcal{L}$ is a finite-dimensional vector space, any linear map on $\mathcal{L}$ is continuous and hence bounded (with respect to any suitable norm $||\cdot ||$). Hence, for any $x \in \mathcal{L}$, we have
\begin{align*}
    \sum_{k=0}^\infty \frac{ || \mathfrak{D}^k (x) ||}{k!} \leq \sum_{k=0}^\infty \frac{|| \mathfrak{D} ||^k || x ||}{k!} = e^{||\mathfrak{D}||} || x || < \infty \quad (\because || \mathfrak{D} || < \infty).
\end{align*}
As $\mathcal{L}$ is a Banach space, and every absolutely convergent series in a Banach space is convergent, we get that the map $\mathrm{exp} (\mathfrak{D})$ is well-defined. Next, for any $x_1, \ldots, x_n \in \mathcal{L}$ and $k \geq 0$, we have from (\ref{deri-iden}) that
    \begin{align*}
        \frac{1}{k!} \mathfrak{D}^k ([x_1, \ldots, x_n]) = \sum_{\substack{i_1 + \cdots + i_n = k \\ i_1, \ldots, i_n \geq 0}} \big[ \frac{\mathfrak{D}^{i_1} (x_1)}{i_1 !}, \ldots, \frac{\mathfrak{D}^{i_n} (x_n)}{i_n !} \big].
    \end{align*}
    Hence
    \begin{align*}
      \mathrm{exp} (\mathfrak{D})  [  x_1, \ldots, x_n ] = \sum_{k=0}^\infty \frac{1}{k!} \mathfrak{D}^k ( [  x_1, \ldots, x_n ]) 
=~& \sum_{k=0}^\infty \sum_{\substack{i_1 + \cdots + i_n = k \\ i_1, \ldots, i_n \geq 0}} \big[ \frac{\mathfrak{D}^{i_1} (x_1)}{i_1 !}, \ldots, \frac{\mathfrak{D}^{i_n} (x_n)}{i_n !} \big] \\
=~& \big[ \sum_{i_1 = 0}^\infty \frac{ \mathfrak{D}^{i_1} (x_1)}{i_1 !} , \ldots, \sum_{i_n = 0}^\infty \frac{ \mathfrak{D}^{i_n} (x_n)}{i_n !} \big] \\
=~& [ \mathrm{exp} (\mathfrak{D}) (x_1), \ldots, \mathrm{exp} (\mathfrak{D}) (x_n)],
    \end{align*}
    which shows that $\mathrm{exp} (\mathfrak{D}) : \mathcal{L} \rightarrow \mathcal{L}$ is an $n$-Leibniz algebra homomorphism. Finally, it is an isomorphism with the inverse $\mathrm{exp} (- \mathfrak{D})$.
\end{proof}

\begin{lemma}
    Let $(\mathcal{L}, [-, \ldots, -])$ be a finite-dimensional $n$-Leibniz algebra and $\varphi : \mathcal{L} \rightarrow \mathcal{L}$ be an $n$-Leibniz algebra homomorphism. Then for any $x_1, \ldots, x_{n-1} \in \mathcal{L}$, we have
    \begin{align}\label{alpha-commute}
        \varphi \circ \mathrm{exp} (\mathrm{ad}_{x_1, \ldots, x_{n-1}}) = \mathrm{exp} (\mathrm{ad}_{\varphi (x_1), \ldots, \varphi (x_{n-1})}) \circ \varphi.
    \end{align}
\end{lemma}

\begin{proof}
    For any $x_n \in \mathcal{L}$, we have
    \begin{align*}
         \big( \varphi \circ \mathrm{exp} (\mathrm{ad}_{x_1, \ldots, x_{n-1}}) \big) (x_n) =~& \varphi \circ \lim_{N \rightarrow \infty} \sum_{k=0}^N \frac{(\mathrm{ad}_{x_1, \ldots, x_{n-1}})^k}{k ! } (x_n) \\
         =~& \lim_{N \rightarrow \infty} \sum_{k=0}^N \frac{ \varphi ( (\mathrm{ad}_{x_1, \ldots, x_{n-1}})^k (x_n))}{k ! } \quad (\because \varphi \text{ is continuous}) \\
         =~& \lim_{N \rightarrow \infty} \sum_{k=0}^N \frac{ (\mathrm{ad}_{\varphi (x_1), \ldots, \varphi ( x_{n-1})})^k ~ \! \varphi (x_n) }{k!} = \big(  \mathrm{exp} (\mathrm{ad}_{\varphi (x_1), \ldots, \varphi (x_{n-1})}) \circ \varphi    \big) (x_n).
    \end{align*}
    This completes the proof.
\end{proof}

When we have a homomorphism $\varphi : \mathcal{L} \rightarrow \mathcal{L}'$ between two finite-dimensional $n$-Leibniz algebras, then one can similarly show that 
$ \varphi \circ \mathrm{exp} (\mathrm{ad}^{\mathcal{L}}_{x_1, \ldots, x_{n-1}}) = \mathrm{exp} (\mathrm{ad}^{\mathcal{L}'}_{\varphi (x_1), \ldots, \varphi (x_{n-1})}) \circ \varphi$, for all $x_1, \ldots, x_{n-1} \in \mathcal{L}$. Here $\mathrm{ad}^{\mathcal{L}}_{x_1, \ldots, x_{n-1}}$ and $\mathrm{ad}^{\mathcal{L}'}_{\varphi (x_1), \ldots, \varphi (x_{n-1})}$ are respectively the adjoint maps on the $n$-Leibniz algebras $\mathcal{L}$ and $\mathcal{L}'$.

\begin{thm}\label{thm-nleibniz-nrack}
    Let $(\mathcal{L}, [-, \ldots, -])$ be a finite-dimensional $n$-Leibniz algebra. Define an $n$-ary operation $\langle - , \ldots, - \rangle : \mathcal{L} \times \cdots \times \mathcal{L} \rightarrow \mathcal{L}$ by
    \begin{align*}
        \langle x_1, \ldots, x_n \rangle := \mathrm{exp} (\mathrm{ad}_{x_2, \ldots, x_n}) (x_1), \text{ for } x_1, \ldots, x_n \in \mathcal{L}.
    \end{align*}
    Then $(\mathcal{L}, \langle -, \ldots, - \rangle)$ is an $n$-rack. Moreover, if $\varphi: \mathcal{L} \rightarrow \mathcal{L}'$ is a homomorphism of finite-dimensional $n$-Leibniz algebras, then $\varphi$ is also a homomorphism of the corresponding $n$-racks.
\end{thm}

\begin{proof}
    For any $x_1, \ldots, x_n, y_1, \ldots, y_{n-1} \in \mathcal{L}$, we observe that
    \begin{align*}
        \langle \langle x_1, \ldots, x_n \rangle, y_1, \ldots, y_{n-1} \rangle 
        &= \langle \mathrm{exp} (\mathrm{ad}_{x_2, \ldots, x_n} ) (x_1), y_1, \ldots, y_{n-1} \rangle \\
        &= \mathrm{exp} (\mathrm{ad}_{y_1, \ldots, y_{n-1}} ) ~ \! \mathrm{exp} (\mathrm{ad}_{x_2, \ldots, x_n} )(x_1) \\
        &\stackrel{(\ref{alpha-commute})}{=} \mathrm{exp} \big(  \mathrm{ad}_{ \mathrm{exp} (\mathrm{ad}_{y_1, \ldots, y_{n-1}} ) (x_2), \ldots, \mathrm{exp} (\mathrm{ad}_{y_1, \ldots, y_{n-1}} ) (x_n)   }   \big) ~ \! \mathrm{exp} (\mathrm{ad}_{y_1, \ldots, y_{n-1}} ) (x_1) \\
        &= \mathrm{exp} \big( \mathrm{ad}_{  \langle x_2, y_1, \ldots, y_{n-1} \rangle   , \ldots, \langle x_n, y_1, \ldots, y_{n-1} \rangle  } \big) \langle x_1, y_1, \ldots, y_{n-1} \rangle  \\
        &= \langle \langle x_1, y_1, \ldots, y_{n-1} \rangle , \langle x_2, y_1, \ldots, y_{n-1} \rangle , \ldots, \langle x_n, y_1, \ldots, y_{n-1} \rangle \rangle
    \end{align*}
    which verifies the identity (\ref{n-shelf}). Finally, the map $\langle - , y_1, \ldots, y_{n-1} \rangle : \mathcal{L} \rightarrow \mathcal{L}$ is simply the exponential map $\mathrm{exp} (\mathrm{ad}_{y_1, \ldots, y_{n-1}})$ and hence it is bijective (see Proposition \ref{prop-expo}). This proves the first part.

    For the second part, we observe that
    \begin{align*}
        \varphi ( \langle x_1, \ldots, x_n \rangle ) =~& \varphi \big(   \mathrm{exp} (\mathrm{ad}^\mathcal{L}_{x_2, \ldots, x_n}) (x_1)  \big) \\
        =~&  \mathrm{exp} (\mathrm{ad}^{\mathcal{L}'}_{\varphi(x_2), \ldots, \varphi( x_n)}) ( \varphi(x_1)) = \langle    \varphi (x_1), \ldots, \varphi (x_n) \rangle',
    \end{align*}
    for all $x_1, \ldots, x_{n} \in \mathcal{L}$. This shows that $\varphi$ is a homomorphism of the corresponding $n$-racks.
\end{proof}

\begin{remark}
    When $n = 2$, the above result coincides with the construction of a rack from a finite-dimensional Leibniz algebra given by Kinyon \cite{kinyon}.  Moreover, when $n=3$, we recover the construction of a $3$-rack from a finite-dimensional $3$-Leibniz algebra \cite{xu-sheng}.
\end{remark}

Let $(\mathcal{L}, [-, \ldots, -])$ be a finite-dimensional $n$-Leibniz algebra. Then by Theorem \ref{thm-nleibniz-nrack}, one can construct an $n$-rack $(\mathcal{L}, \langle - , \ldots, - \rangle )$ and hence a rack $(\mathcal{L}^{\times (n-1)}, \triangleleft_{ \langle - , \ldots, - \rangle })$ by Proposition \ref{prop-nrack-rack}. On the other hand, from the given $n$-Leibniz algebra $(\mathcal{L}, [-, \ldots, -])$, by Proposition \ref{funda-leibniz}, one may construct the Leibniz algebra $(\mathcal{L}^{\otimes (n-1)}, \{ ~, ~ \})$ on the space of fundamental elements and hence a rack structure $( \mathcal{L}^{\otimes (n-1)}, \triangleleft)$ by Kinyon's construction. In the following result, we find a rack homomorphism $\varphi$ between them (see the diagram below):
\begin{align}
    \xymatrix{
     & n\text{-rack} \ar[rr]^{\mathrm{Proposition ~} \ref{prop-nrack-rack}} &  & \text{rack} \ar@{-->}[dd]^\varphi\\
   \substack{ \text{f.d. } n\text{-Leibniz ~algebra} \\ (\mathcal{L}, [-, \ldots, -])} \ar@/^0.8pc/[ru]^{\mathrm{Theorem ~} \ref{thm-nleibniz-nrack}} \ar@/_0.8pc/[rd]_{\mathrm{Proposition ~} \ref{funda-leibniz}} &  &  & \\
      & \substack{ \text{Leibniz algebra } \\ (\mathcal{L}^{\otimes (n-1)}, \{ -, - \} ) } \ar[rr]_{ \substack{\mathrm{Kinyon's } \\ \mathrm{construction} }} & & \text{rack.}
    }
\end{align}

\begin{thm}
    Let $(\mathcal{L}, [-, \ldots, -])$ be a finite-dimensional $n$-Leibniz algebra. We define a set-map $\varphi : \mathcal{L}^{\times (n-1)} \rightarrow \mathcal{L}^{\otimes (n-1)}$ by
    \begin{align*}
        \varphi (x_1, \ldots, x_{n-1}) = x_1 \otimes \cdots \otimes x_{n-1}, \text{ for } (x_1, \ldots, x_{n-1}) \in \mathcal{L}^{\times (n-1)}.
    \end{align*}
    Then $\varphi$ is a rack homomorphism from $( \mathcal{L}^{\times (n-1)} , \triangleleft_{ \langle -, \ldots, - \rangle}) $ to $( \mathcal{L}^{\otimes (n-1)} , \triangleleft) $.
\end{thm}

\begin{proof}
    For any $(x_1, \ldots, x_{n-1}), (y_1, \ldots, y_{n-1}) \in \mathcal{L}^{\times (n-1)}$, we have
    \begin{align*}
        &\varphi \big( (x_1, \ldots, x_{n-1}) \triangleleft_{ \langle -, \ldots, - \rangle} (y_1, \ldots, y_{n-1})   \big) \\
        &= \varphi (  \langle x_1, y_1, \ldots, y_{n-1} \rangle, \ldots,  \langle x_{n-1}, y_1, \ldots, y_{n-1} \rangle  ) \\
       & = \varphi \big(  \mathrm{exp} (\mathrm{ad}_{y_1, \ldots, y_{n-1}}) (x_1), \ldots , \mathrm{exp} (\mathrm{ad}_{y_1, \ldots, y_{n-1}}) (x_{n-1})  \big) \\
        &= \big(  \sum_{i_1 = 0}^\infty \frac{ (\mathrm{ad}_{y_1, \ldots, y_{n-1}})^{i_1} (x_1) }{i_1 !}   \big) \otimes \cdots \otimes \big(   \sum_{i_{n-1} = 0}^\infty \frac{ (\mathrm{ad}_{y_1, \ldots, y_{n-1}})^{i_{n-1}} (x_{n-1}) }{i_{n-1} !}   \big)\\
       & = \sum_{k=0}^\infty ~\sum_{\substack{ i_1 + \cdots + i_{n-1} = k \\ i_1, \ldots, i_{n-1} \geq 0}} \frac{ (\mathrm{ad}_{y_1, \ldots, y_{n-1}})^{i_1} (x_1) }{i_1 !} \otimes \cdots \otimes \frac{ (\mathrm{ad}_{y_1, \ldots, y_{n-1}})^{i_{n-1}} (x_{n-1}) }{i_{n-1} !}.
    \end{align*}
    On the other hand, in the Leibniz algebra $(\mathcal{L}^{\otimes (n-1)}, \{ ~, ~ \})$, we get that 
    \begin{align*}
        \mathrm{ad}_{y_1 \otimes \cdots \otimes y_{n-1}} = \sum_{i=1}^{n-1} ~ \mathrm{Id} \otimes \cdots \otimes \underbrace{\mathrm{ad}_{y_1, \ldots, y_{n-1}} }_{i\text{-th place}} \otimes \cdots \otimes \mathrm{Id} \qquad (\text{follows from } (\ref{fun-identity}))
    \end{align*}
    which yields that
    \begin{align}\label{ad-k}
        (\mathrm{ad}_{y_1 \otimes \cdots \otimes y_{n-1}})^k = \sum_{\substack{ i_1 + \cdots + i_{n-1} = k \\ i_1, \ldots, i_{n-1} \geq 0}} \frac{k!}{i_1 ! \cdots i_{n-1}!} ~ (\mathrm{ad}_{y_1, \ldots, y_{n-1}})^{i_1} \otimes \cdots \otimes  (\mathrm{ad}_{y_1, \ldots, y_{n-1}})^{i_{n-1}}.
    \end{align}
    Hence 
    \begin{align*}
        &\varphi (x_1, \ldots, x_{n-1}) \triangleleft \varphi (y_1, \ldots, y_{n-1}) \\
       & = (x_1 \otimes \cdots \otimes x_{n-1}) \triangleleft (y_1 \otimes \cdots \otimes y_{n-1})\\
        &= \mathrm{exp} (\mathrm{ad}_{y_1 \otimes \cdots \otimes y_{n-1}}) (x_1 \otimes \cdots \otimes x_{n-1})\\
        &= \sum_{k=0}^\infty \frac{  (\mathrm{ad}_{y_1 \otimes \cdots \otimes y_{n-1}})^k}{k!} (x_1 \otimes \cdots \otimes x_{n-1})\\
        &\stackrel{\textnormal{(\ref{ad-k})}}{=} \sum_{k=0}^\infty \bigg(  \sum_{\substack{ i_1 + \cdots + i_{n-1} = k \\ i_1, \ldots, i_{n-1} \geq 0}} \frac{1}{i_1 ! \cdots i_{n-1}!} ~ (\mathrm{ad}_{y_1, \ldots, y_{n-1}})^{i_1} \otimes \cdots \otimes  (\mathrm{ad}_{y_1, \ldots, y_{n-1}})^{i_{n-1}} \bigg) (x_1 \otimes \cdots \otimes x_{n-1})\\
        &= \sum_{k=0}^\infty ~ \sum_{\substack{ i_1 + \cdots + i_{n-1} = k \\ i_1, \ldots, i_{n-1} \geq 0}} \frac{ (\mathrm{ad}_{y_1, \ldots, y_{n-1}})^{i_1} (x_1) }{i_1 !} \otimes \cdots \otimes \frac{ (\mathrm{ad}_{y_1, \ldots, y_{n-1}})^{i_{n-1}} (x_{n-1}) }{i_{n-1} !}.
    \end{align*}
    This shows that $ \varphi \big( (x_1, \ldots, x_{n-1}) \triangleleft_{ \langle -, \ldots, - \rangle } (y_1, \ldots, y_{n-1})   \big) = \varphi (x_1, \ldots, x_{n-1}) \triangleleft \varphi (y_1, \ldots, y_{n-1})$ which proves the result.
\end{proof}

\section{Central {\em n}-Leibniz algebras and the Yang-Baxter equation}\label{sec3}

In this section, we consider central $n$-Leibniz algebras and show that they give rise to Yang-Baxter operators. Consequently, we define two Yang-Baxter operators from a given $n$-Leibniz algebra (which need not be central) and connect them by a suitable homomorphism.

\medskip

First, recall that a {\bf central Leibniz algebra} is a Leibniz algebra $(\mathfrak{g}, \{ -, - \})$ equipped with a central element, i.e., an element ${\bf 1} \in \mathfrak{g}$ that satisfies $\{ {\bf 1} , x \} = \{ x, {\bf 1} \} = 0$, for all $x \in \mathfrak{g}$. Let $(\mathfrak{g}, \{ -, - \}, 1)$ be a central Leibniz algebra. Then define a linear map $R : \mathfrak{g} \otimes \mathfrak{g} \rightarrow  \mathfrak{g} \otimes \mathfrak{g} $ by 
\begin{align}\label{lebed-exp}
    R (x \otimes y ) = y \otimes x + {\bf 1} \otimes \{ x, y \}, \text{ for } x \otimes y \in \mathfrak{g} \otimes \mathfrak{g}.
\end{align}
In \cite{lebed2}, Lebed showed that $R$ is a Yang-Baxter operator on the vector space $\mathfrak{g}$. More generally, let $\mathfrak{g}$ be a vector space equipped with a bilinear map $\{ - , - \} : \mathfrak{g} \times \mathfrak{g} \rightarrow \mathfrak{g}$. We consider the linear map 
\begin{align*}
\widetilde{R} : ({\bf k} \oplus \mathfrak{g}) \otimes ({\bf k} \oplus \mathfrak{g}) \rightarrow ({\bf k} \oplus \mathfrak{g}) \otimes ({\bf k} \oplus \mathfrak{g}) ~ \text{ by } ~\widetilde{R} ((\lambda, x) \otimes (\mu, y)) = (\mu, y) \otimes (\lambda, x) ~ \! +~ \!(1, 0) \otimes (0, \{x, y \}),
\end{align*}
for $(\lambda. x), (\mu, y )\in {\bf k} \oplus \mathfrak{g}.$ Then $(\mathfrak{g}, \{ -, - \})$ is a Leibniz algebra if and only if $\widetilde{R}$ is a Yang-Baxter operator on the vector space ${\bf k} \oplus \mathfrak{g}$ \cite{baez-crans}.

%In the following, we first consider central $n$-Leibniz algebras and generalize the above result to obtain pre-$n$-Yang-Baxter operators.

\medskip

Next, let $(\mathcal{L}, [-, \ldots, -])$ be an $n$-Leibniz algebra. Generalizing the previous concept, an element ${\bf 1} \in \mathcal{L}$ is said to be a {\em central element} if 
\begin{align*}
[ x_1, \ldots, x_{i-1}, {\bf 1}, x_{i+1}, \ldots, x_n ] = 0, \text{ for all } 1 \leq i \leq n \text{ and } x_1, \ldots, x_{i-1}, x_{i+1}, \ldots, x_n \in \mathcal{L}.
\end{align*}
A {\bf central $n$-Leibniz algebra} $(\mathcal{L}, [-, \ldots, -], {\bf 1})$ is an $n$-Leibniz algebra $(\mathcal{L}, [-, \ldots, -])$ equipped with a central element ${\bf 1} \in \mathcal{L}$.

\begin{remark}
Let $(\mathfrak{g}, \{ -, - \}, {\bf 1})$ be a central Leibniz algebra. Then $(\mathcal{L}, [-, \ldots, -], {\bf 1})$ is a central $n$-Leibniz algebra, where the $n$-Leibniz bracket on $\mathcal{L}$ is given in Example \ref{leib-to-n}.
 % \item[(ii)] Let $(\mathcal{L}, [-, \ldots, -], {\bf 1})$ be a central $n$-Leibniz algebra. Then $(\mathcal{L}^{\otimes (n-1)},  \{ - , - \}, {\bf 1}^{\otimes (n-1)})$ is a central Leibniz algebra, where the Leibniz bracket on $\mathcal{L}^{\otimes (n-1)}$is given in .....\textcolor{red}{remove}. 
\end{remark}

\begin{thm}
    Let $(\mathcal{L}, [-, \ldots, -], {\bf 1})$ be a central $n$-Leibniz algebra. Then the triple $(\mathcal{L}^{\otimes (n-1)}, \{ ~, ~ \} , {\bf 1}^{\otimes (n-1)})$ is a central Leibniz algebra, where the Leibniz bracket on the space $ \mathcal{L}^{\otimes (n-1)}$ is defined in (\ref{funda-leib-bracket}). Consequently, the map $R : \mathcal{L}^{\otimes (n-1)} \otimes \mathcal{L}^{\otimes (n-1)} \rightarrow \mathcal{L}^{\otimes (n-1)} \otimes \mathcal{L}^{\otimes (n-1)}$ defined by
    \begin{align*}
        R ((x_1 \otimes \cdots \otimes x_{n-1}) \otimes (y_1 \otimes \cdots \otimes y_{n-1}) ) =~& (y_1 \otimes \cdots \otimes y_{n-1}) \otimes (x_1 \otimes \cdots \otimes x_{n-1}) \\
        +& \sum_{i=1}^{n-1} ({\bf 1}^{\otimes (n-1)}) \otimes (x_1 \otimes \cdots \otimes [x_i, y_1, \ldots, y_{n-1}] \otimes \cdots \otimes x_{n-1})
    \end{align*}
    is a Yang-Baxter operator on the vector space $\mathcal{L}^{\otimes (n-1)}$.
\end{thm}

\begin{proof}
    For any $x_1 \otimes \cdots \otimes x_{n-1} \in \mathcal{L}^{\otimes (n-1)}$, we observe that
    \begin{align*}
    \{ {\bf 1}^{\otimes (n-1)} , x_1 \otimes \cdots \otimes x_{n-1}  \} =~& \sum_{i=1}^{n-1} {\bf 1} \otimes \cdots \otimes \underbrace{[{\bf 1}, x_1, \ldots, x_{n-1}]}_{i\text{-}\mathrm{th~ place}} \otimes \cdots \otimes {\bf 1} = 0,\\
        \{ x_1 \otimes \cdots \otimes x_{n-1} , {\bf 1}^{\otimes (n-1)} \} =~& \sum_{i=1}^{n-1} x_1 \otimes \cdots \otimes [x_i, {\bf 1}, \ldots, {\bf 1}] \otimes \cdots \otimes x_{n-1} = 0.
    \end{align*}
    This shows that ${\bf 1}^{\otimes (n-1)}$ is a central element in the Leibniz algebra $(\mathcal{L}^{\otimes (n-1)}, \{ ~, ~ \})$. Hence, the first part follows. Finally, the last part follows from the result of Lebed (see the expression (\ref{lebed-exp})).
\end{proof}

In the following, we shall construct two Yang-Baxter operators associated to an (need not be central) $n$-Leibniz algebra. Both constructions are simple and based on embeddings of any $n$-Leibniz algebras to central $n$-Leibniz algebras.

\begin{prop}
    Let $(\mathcal{L}, [-, \ldots, -])$ be an $n$-Leibniz algebra. Then $( \overline{\mathcal{L}^{\otimes (n-1)}}= {\bf k} \oplus \mathcal{L}^{\otimes (n-1)}, \{ \! \! \{ ~, ~ \} \! \! \}, (1, 0) )$ is a central Leibniz algebra, where
    \begin{align*}
       \{ \! \! \{  (\lambda, x_1 \otimes \cdots \otimes x_{n-1}), (\mu, y_1 \otimes \cdots \otimes y_{n-1}) \} \! \! \} = \sum_{i=1}^{n-1} \big( 0,  x_1 \otimes \cdots \otimes [x_i, y_1, \ldots, y_{n-1}] \otimes \cdots \otimes x_{n-1} \big).
    \end{align*}
    Consequently, the map $R_1 :  \overline{\mathcal{L}^{\otimes (n-1)}} \otimes  \overline{\mathcal{L}^{\otimes (n-1)}} \rightarrow  \overline{\mathcal{L}^{\otimes (n-1)}} \otimes  \overline{\mathcal{L}^{\otimes (n-1)}}$ given by
    \begin{align*}
        R_1 \big(  (\lambda, x_1 \otimes \cdots \otimes x_{n-1}) ~\otimes ~ &  (\mu, y_1 \otimes \cdots \otimes y_{n-1})   \big) = (\mu, y_1 \otimes \cdots \otimes y_{n-1}) \otimes (\lambda, x_1 \otimes \cdots \otimes x_{n-1}) \\
       & + \sum_{i=1}^{n-1} \underbrace{(1, 0)}_{\in ~  \overline{\mathcal{L}^{\otimes (n-1)}}  } \otimes~ (0, x_1 \otimes \cdots \otimes [x_i, y_1, \ldots, y_{n-1}] \otimes \cdots \otimes x_{n-1})
    \end{align*}
    is a Yang-Baxter operator on the vector space $ \overline{\mathcal{L}^{\otimes (n-1)}}$.
\end{prop}

\begin{prop}\label{nl-cnl}
    Let $(\mathcal{L}, [-, \ldots, -])$ be an $n$-Leibniz algebra. Then $(\overline{\mathcal{L}} = {\bf k} \oplus \mathcal{L}, \llbracket -, \ldots, - \rrbracket ,(1, 0))$ is a central $n$-Leibniz algebra, where
    \begin{align*}
        \llbracket (\lambda_1, x_1), \ldots, (\lambda_n, x_n ) \rrbracket = (0, [x_1, \ldots, x_n]), \text{ for all } (\lambda_i, x_i) \in \overline{\mathcal{L}}.
    \end{align*}
    Consequently, the map $R_2 : \overline{\mathcal{L}}^{\otimes (n-1)} \otimes \overline{\mathcal{L}}^{\otimes (n-1)} \rightarrow \overline{\mathcal{L}}^{\otimes (n-1)} \otimes \overline{\mathcal{L}}^{\otimes (n-1)}$ given by
    \begin{align*}
        &R_2 \big( ( (\lambda_1, x_1) \otimes \cdots \otimes (\lambda_{n-1}, x_{n-1})) \otimes ( (\mu_1, y_1) \otimes \cdots \otimes (\mu_{n-1}, y_{n-1}))  \big) \\
        & \quad = ( (\mu_1, y_1) \otimes \cdots \otimes (\mu_{n-1}, y_{n-1})) \otimes ( (\lambda_1, x_1) \otimes \cdots \otimes (\lambda_{n-1}, x_{n-1})) \\
        & \qquad + \sum_{i=1}^{n-1} ((1, 0)^{\otimes (n-1)} ) \otimes \big(   (\lambda_1, x_1) \otimes \cdots \otimes (0, [x_i, y_1, \ldots, y_{n-1}]) \otimes \cdots \otimes (\lambda_{n-1}, x_{n-1})  \big)
    \end{align*}
    is a Yang-Baxter operator on the vector space $\overline{\mathcal{L}}^{\otimes (n-1)} $.
\end{prop}

Since both the above Yang-Baxter operators are obtained from a given $n$-Leibniz algebra, it is natural to find a relation between the two operators. For this, we first take two central Leibniz algebras $(\mathfrak{g}, \{ -, - \}, {\bf 1}_{\mathfrak{g}})$ and $(\mathfrak{g}', \{ -, - \}', {\bf 1}_{\mathfrak{g}'})$. Suppose $\varphi : \mathfrak{g} \rightarrow \mathfrak{g}'$ is a homomorphism of Leibniz algebras satisfying $\varphi ({\bf 1}_\mathfrak{g}) = {\bf 1}_{\mathfrak{g}'}$. In this case, we say that $\varphi$ is a homomorphism of central Leibniz algebras. Then $\varphi $ satisfies
\begin{align*}
   R_{\mathfrak{g}'} \circ ( \varphi \otimes \varphi ) = ( \varphi \otimes \varphi ) \circ R_\mathfrak{g},
\end{align*}
where $R_\mathfrak{g}$ (resp. $R_{\mathfrak{g}'}$) is the Yang-Baxter operator on the vector space $\mathfrak{g}$ (resp. $\mathfrak{g}'$) obtained by Lebed's construction (\ref{lebed-exp}). Next, let $(\mathcal{L}, [-, \ldots, - ])$ be any given $n$-Leibniz algebra. We define an injective linear map 
 \begin{align}\label{eta-map}
       \eta: \overline{ \mathcal{L}^{\otimes (n-1)}} \rightarrow \overline{\mathcal{L}}^{\otimes (n-1)} ~~ \text{ by } ~~ \eta ( (\lambda, x_1 \otimes \cdots \otimes x_{n-1})) = \lambda \underbrace{(1, 0) \otimes \cdots \otimes (1, 0)}_{(n-1) \mathrm{~ times}} ~+~ (0, x_1) \otimes \cdots \otimes (0, x_{n-1}),
    \end{align}
    for $ (\lambda, x_1 \otimes \cdots \otimes x_{n-1}) \in \overline{ \mathcal{L}^{\otimes (n-1)}} $. Then we have the following. 

\begin{prop}
    The map $\eta : \overline{ \mathcal{L}^{\otimes (n-1)}} \rightarrow \overline{\mathcal{L}}^{\otimes (n-1)}$ defined above is a homomorphism of central Leibniz algebras from $ ( \overline{\mathcal{L}^{\otimes (n-1)}}, \{ \! \! \{ ~, ~ \} \! \! \}, (1, 0) )$ to $(\overline{\mathcal{L}}^{\otimes (n-1)} ,  \{ \! \llbracket ~, ~ \rrbracket \! \}, (1,0)^{\otimes (n-1)} )$, where $(\overline{\mathcal{L}}^{\otimes (n-1)} ,  \{ \! \llbracket ~, ~ \rrbracket \! \}, (1,0)^{\otimes (n-1)} )$ is the central Leibniz algebra obtained from the central $n$-Leibniz algebra $(\overline{\mathcal{L}}, \llbracket -, \ldots, - \rrbracket ,(1, 0))$ given in Proposition \ref{nl-cnl}. Additionally, the map $\eta$ satisfies $R_2 \circ (\eta \otimes \eta) = (\eta \otimes \eta) \circ R_1$.
\end{prop}

\begin{proof}
    For any $(\lambda, x_1 \otimes \cdots \otimes x_{n-1}) , (\mu, y_1 \otimes \cdots \otimes y_{n-1}) \in  \overline{ \mathcal{L}^{\otimes (n-1)}}$, we observe that
    \begin{align*}
        &\eta \big(  \{ \! \! \{ (\lambda, x_1 \otimes \cdots \otimes x_{n-1}) ,  (\mu, y_1 \otimes \cdots \otimes y_{n-1})  \} \! \! \}  \big) \\
        &= \sum_{i=1}^{n-1} \eta (0, x_1 \otimes \cdots \otimes [x_i, y_1, \ldots, y_{n-1}] \otimes \cdots \otimes x_{n-1}) \\
        &= \sum_{i=1}^{n-1} (0, x_1) \otimes \cdots \otimes (0, [x_i, y_1, \ldots, y_{n-1}]) \otimes \cdots \otimes (0, x_{n-1}) \\
        &=  \{ \! \llbracket \lambda \underbrace{(1,0) \otimes \cdots \otimes (1, 0)}_{n-1 \mathrm{~ times}} ~\! + ~\! (0, x_1) \otimes \cdots \otimes (0, x_{n-1}) ~ \! , ~ \! \mu \underbrace{(1,0) \otimes \cdots \otimes (1, 0)}_{n-1 \mathrm{~ times }} ~ \! + ~\! (0, y_1) \otimes \cdots \otimes (0, y_{n-1}) \rrbracket \! \} \\
        &=  \{ \! \llbracket \eta ( \lambda, x_1 \otimes \cdots \otimes x_{n-1} ) , \eta (\mu, y_1 \otimes \cdots \otimes y_{n-1}) \rrbracket \! \}.
    \end{align*}
    Thus, $\eta$ is a homomorphism of Leibniz algebras. Moreover, we have $\eta (1, 0) = (1, 0)^{\otimes (n-1)}$, which shows that $\eta$ preserves the central elements. This proves the first part. Finally, the second part follows as $\eta$ is a homomorphism of central Leibniz algebras.
\end{proof}
    
\section{Linear {\em n}-racks and the Yang-Baxter equation}\label{sec4}
In this section, we first recall linear racks and then introduce linear $n$-racks as the $n$-ary generalization of linear racks. We show that every linear rack gives rise to a linear $n$-rack. On the other hand, when the underlying coalgebra $C$ is cocommutative, every linear $n$-rack structure on $C$ induces a linear rack structure on the tensor product $C^{\otimes (n-1)}$. Hence, by Lebed's construction \cite{lebed}, one obtains a Yang-Baxter operator on $C^{\otimes (n-1)}$. Further, we show that an $n$-Leibniz algebra $(\mathcal{L}, [-, \ldots, -])$ yields a cocommutative linear $n$-rack structure on $\overline{\mathcal{L}} = {\bf k} \oplus \mathcal{L}$, and therefore, gives rise to a Yang-Baxter operator on $\overline{\mathcal{L}}^{\otimes (n-1)}$.

\begin{defn} \cite{carter,lebed,xu-sheng}
    A {\bf linear shelf} is a pair $(C, \triangleleft)$ consisting of a coassociative counital coalgebra $(C, \Delta, \varepsilon)$ equipped with a coalgebra homomorphism $\triangleleft : C \otimes C \rightarrow C$ that satisfies the following general self-distributive property:
    \begin{align*}
        (u \triangleleft v) \triangleleft w = (u \triangleleft w^{(1)}) \triangleleft (v \triangleleft w^{(2)}), \text{ for all } u, v, w \in C.
    \end{align*}
    A linear shelf $(C, \triangleleft)$ is said to be a {\bf linear rack} if there exists an additional coalgebra homomorphism $\widetilde{\triangleleft} : C \otimes C \rightarrow C$ that makes $(C, \widetilde{\triangleleft})$ into another linear shelf satisfying additionally
    \begin{align}\label{addition-lin}
        (u \triangleleft v^{(2)}) ~\widetilde{\triangleleft}  ~v^{(1)} = \varepsilon (v) u = (u ~\widetilde{\triangleleft} ~ v^{(2)}) \triangleleft v^{(1)}, \text{ for all } u , v \in C.
    \end{align}
\end{defn}

Let $(C, \triangleleft)$ and $(C', \triangleleft')$ be two linear racks. A {\bf homomorphism} of linear racks from $(C, \triangleleft)$ to $(C', \triangleleft')$ is a coalgebra homomorphism $f : C \rightarrow C'$ satisfying $f (u \triangleleft v) = f(u) \triangleleft' f(v)$, for $u, v \in C.$

\begin{exam}\label{exam-lrack}
    For any nonempty set $X$, consider the vector space ${\bf k}[X]$ spanned by the elements of $X$. That is, 
    \begin{align*}
        {\bf k} [X] = \big\{ \sum_{x \in X} \lambda_x ~ \!  x ~ \! | ~ \! \lambda_x \in {\bf k} \text{ and } \lambda_x = 0 \text{ for all but finitely many } x \big\}.
    \end{align*}
    Then ${\bf k} [X]$ is a coassociative counital coalgebra with the coproduct $\Delta : {\bf k} [X] \rightarrow {\bf k} [X] \otimes {\bf k} [X]$ and the counit $\varepsilon : {\bf k} [X] \rightarrow {\bf k}$ which are respectively given by
    \begin{align*}
        \Delta (x) = x \otimes x ~~~~ \text{ and } ~~~~ \varepsilon (x) = 1, \text{ for } x \in X.
    \end{align*}
    If $(X , \triangleleft )$ is a rack then by extending the rack operation $\triangleleft : X \times X \rightarrow X$ linearly to the coalgebra ${\bf k} [X]$, one obtain a linear rack $({\bf k} [X], \triangleleft).$
\end{exam}

\begin{exam}\label{last-examm}\cite{lebed} Let $\mathcal{L}$ be any vector space (not necessarily having any additional structure). Then the vector space $\overline{\mathcal{L}} = {\bf k} \oplus \mathcal{L}$ can be given a cocommutative coassociative counital coalgebra structure with the coproduct $\Delta: \overline{\mathcal{L}} \rightarrow \overline{\mathcal{L}} \otimes \overline{\mathcal{L}}$ and the count $\varepsilon: \overline{\mathcal{L}} \rightarrow {\bf k}$ which are respectively given by 
\begin{align*}
    \Delta ( \lambda , x) :=~& (\lambda , x) \otimes (1 , 0) + (1 , 0) \otimes (0, x) = \lambda (1, 0) \otimes (1, 0) + (0, x) \otimes (1, 0) + (1, 0) \otimes (0, x) \\
    & ~~~~ \text{ and } ~~~~ \varepsilon (\lambda, x) := \lambda, \text{ for } (\lambda, x) \in \overline{\mathcal{L}}.
\end{align*}
Further, if $(\mathcal{L}, \{ - , - \})$ is a Leibniz algebra, then $(\overline{ \mathcal{L}}, \triangleleft)$ is a linear rack, where $\triangleleft : \overline{ \mathcal{L}} \otimes \overline{ \mathcal{L}} \rightarrow \overline{ \mathcal{L}}$ is defined by
\begin{align*}
    (\lambda, x) \triangleleft (\mu, y) = (\lambda \mu ~\! , ~\! \mu x + \{ x, y \}), \text{ for } (\lambda, x), (\mu, y) \in \overline{ \mathcal{L}}.
\end{align*}
\end{exam}

\medskip

In \cite{lebed}, Lebed showed that a linear rack whose underlying coalgebra is cocommutative gives rise to a Yang-Baxter operator. More precisely, one has the following result.

\begin{prop}\label{lebed-prop}   Let $(C, \triangleleft)$ be a linear rack where $C$ is cocommutative. We define a linear map
\begin{align*}
  R^\triangleleft : C \otimes C \rightarrow C \otimes C ~~~ \text{ by } ~  R^\triangleleft (u \otimes v) = v^{(1)} \otimes (u \triangleleft v^{(2)}), \text{ for } u \otimes v \in C \otimes C.
\end{align*}
    Then $R^\triangleleft$ is a Yang-Baxter operator (on the vector space $C$) with the inverse 
    \begin{align*}
        (R^\triangleleft)^{-1} (u \otimes v) = (v ~ \! \widetilde{\triangleleft} ~ \!  u^{(2)}) \otimes u^{(1)}, \text{ for }  u \otimes v \in C \otimes C.
    \end{align*}
\end{prop}

\begin{defn}
    A {\bf linear $n$-shelf} is a pair $(C, \langle -, \ldots, - \rangle)$ consisting of a coassociative counital coalgebra $(C, \Delta, \varepsilon)$ equipped with a coalgebra homomorphism $\langle -, \ldots, - \rangle : C^{\otimes n} \rightarrow C$ that satisfies the following distributive property:
    \begin{align}\label{self-dis}
        \langle \langle u_1, \ldots, u_n \rangle, v_1, \ldots, v_{n-1} \rangle = \langle \langle u_1, v_1^{(1)}, \ldots, v^{(1)}_{n-1} \rangle, \ldots , \langle u_n, v_1^{(n)}, \ldots, v^{(n)}_{n-1} \rangle \rangle,
    \end{align}
for all $u_1, \ldots, u_n, v_1, \ldots, v_{n-1} \in C$. A linear $n$-shelf $(C, \langle -, \ldots, - \rangle)$ is said to be a {\bf linear $n$-rack} if there exists an additional coalgebra homomorphism $\ll -, \ldots, - \gg : C^{\otimes n}  \rightarrow C$ that makes $(C, \ll -, \ldots, - \gg)$ into another linear $n$-shelf such that for all $u, v_1, \ldots, v_{n-1} \in C$,
    \begin{align}\label{inv-prop}
        \ll \langle u, v_1^{(2)}, \ldots, v_{n-1}^{(2)} \rangle, v_{n-1}^{(1)}, \ldots, v_1^{(1)} \gg = \langle \ll u, v_1^{(2)}, \ldots, v_{n-1}^{(2)} \gg , v_{n-1}^{(1)}, \ldots, v_1^{(1)} \rangle = \varepsilon (v_1) \cdots \varepsilon (v_{n-1}) u.
    \end{align}
\end{defn}

Let $(C, \langle -, \ldots, - \rangle)$ and $(C', \langle -, \ldots, - \rangle')$ be two linear $n$-racks. A {\bf homomorphism} of linear $n$-racks from $C$ to $C'$ is a coalgebra homomorphism $f : C \rightarrow C'$ satisfying
\begin{align*}
    f ( \langle u_1, \ldots, u_n \rangle ) = \langle f(u_1), \ldots, f(u_n) \rangle', \text{ for all } u_1, \ldots, u_n \in C.
\end{align*}

\begin{exam}\label{exam-20}
    Let $(X, \langle -, \ldots, - \rangle)$ be an $n$-rack. By extending the $n$-rack operation $\langle -, \ldots, - \rangle$ linearly to the coalgebra $({\bf k} [X], \Delta, \varepsilon)$ given in Example \ref{exam-lrack}, one obtains a linear $n$-rack $({\bf k}[X], \langle -, \ldots, - \rangle)$. Conversely, let $(C, \langle -, \ldots, - \rangle)$ be a linear $n$-rack. Then the set of all group-like elements of the coalgebra $C$ (recall that a nonzero element $c \in C$ is said to be a group-like element if $\Delta (c) = c \otimes c$) naturally inherits an $n$-rack structure.
\end{exam}

\begin{thm}\label{thm-lrack-lnrack}
    Let $(C, \triangleleft)$ be a linear rack. For any $n \geq 2$, 
    define a linear map $\langle -, \ldots, - \rangle_\triangleleft : C^{\otimes n} \rightarrow C$ by
    \begin{align}\label{some-n}
        \langle u_1, \ldots, u_n \rangle_\triangleleft := ( \cdots  ((u_1 \triangleleft u_2) \triangleleft u_3) \cdots) \triangleleft u_n, \text{ for } u_1, \ldots, u_n \in C.
    \end{align}
    Then $(C, \langle -, \ldots, - \rangle_\triangleleft)$ is a linear $n$-rack.
    %, where the linear map $\langle -, \ldots, - \rangle_\triangleleft : C \otimes \cdots \otimes C \rightarrow C$ is given by 
   % \begin{align*}
    %    \langle u_1, \ldots, u_n \rangle := u_1 \triangleleft (u_2 \triangleleft \cdots (u_{n-1} \triangleleft u_n) \cdots), \text{ for } u_1, \ldots, u_n \in C.
   % \end{align*}
    Moreover, if $f : (C, \triangleleft) \rightarrow (C', \triangleleft')$ is a homomorphism of linear racks then $f$ is a homomorphism of corresponding linear $n$-racks from $(C, \langle -, \ldots, - \rangle_\triangleleft )$ to $(C', \langle -, \ldots, - \rangle_{\triangleleft'} )$.
\end{thm}

\begin{proof} Note that $(C, \Delta, \varepsilon)$ is a coassociative counital coalgebra implies that $(C^{\otimes (n-1)}, \Delta_{C^{\otimes (n-1)}}, \varepsilon_{C^{\otimes (n-1)}})$ is so, where
    \begin{align*}
       & \Delta_{C^{\otimes (n-1)}} (u_1 \otimes \cdots \otimes u_{n-1}) = (u_1^{(1)} \otimes \cdots \otimes u^{(1)}_{n-1}) \otimes (u_1^{(2)} \otimes \cdots \otimes u^{(2)}_{n-1}) \\
       & \quad \qquad \text{ and }  ~~~  \varepsilon_{C^{\otimes (n-1)}}  (u_1 \otimes \cdots \otimes u_{n-1}) = \varepsilon (u_1) \cdots \varepsilon (u_{n-1}),
    \end{align*}
    for $u_1 \otimes \cdots \otimes u_{n-1} \in C^{\otimes (n-1)}$. Next, for any $u_1, \ldots, u_n \in C$, we have
    \begin{align*}
        \Delta (\langle u_1, \ldots, u_n \rangle_\triangleleft) 
        &= \Delta  \big( ( \cdots  ((u_1 \triangleleft u_2) \triangleleft u_3) \cdots) \triangleleft u_n    \big) \\
        &= \big( ( \cdots  ((u_1 \triangleleft u_2) \triangleleft u_3) \cdots) \triangleleft u_n \big)^{(1)} \otimes \big(  ( \cdots  ((u_1 \triangleleft u_2) \triangleleft u_3) \cdots) \triangleleft u_n  \big)^{(2)} \\
        &= \big( ( \cdots  ((u_1^{(1)} \triangleleft u_2^{(1)}) \triangleleft u_3^{(1)}) \cdots) \triangleleft u_n^{(1)} \big) \otimes \big(  ( \cdots  ((u_1^{(2)} \triangleleft u_2^{(2)}) \triangleleft u_3^{(2)}) \cdots) \triangleleft u_n^{(2)}  \big) \\
        &= \langle u_1^{(1)}, \ldots, u_n^{(1)} \rangle_\triangleleft \otimes \langle u_1^{(2)}, \ldots, u_n^{(2)} \rangle_\triangleleft \\
        &= (\langle - , \ldots, - \rangle_\triangleleft \otimes \langle - , \ldots, - \rangle_\triangleleft) \big( (u_1^{(1)} \otimes \cdots \otimes u_n^{(1)})  \otimes  (u_1^{(2)} \otimes \cdots \otimes u_n^{(2)})   \big) \\
        &= \big( (\langle - , \ldots, - \rangle_\triangleleft \otimes \langle - , \ldots, - \rangle_\triangleleft) \circ \Delta_{C^{\otimes n}} \big) (u_1 \otimes \cdots \otimes  u_n)
    \end{align*}
    and 
    \begin{align*}
        (\varepsilon \circ \langle - , \ldots, - \rangle_\triangleleft ) (u_1 \otimes \cdots \otimes u_n) =~& \varepsilon (\langle u_1, \ldots, u_n \rangle_\triangleleft) 
= \varepsilon \big(  ( \cdots  ((u_1 \triangleleft u_2) \triangleleft u_3) \cdots) \triangleleft u_n   \big) \\
=~& \varepsilon (u_1) \cdots \varepsilon (u_n) \qquad  (\because \varepsilon (u \triangleleft v) = \varepsilon_{C \otimes C} (u \otimes v) = \varepsilon (u) \varepsilon (v)).
    \end{align*}
    This shows that $\langle -, \ldots, - \rangle_\triangleleft : C^{\otimes n} \rightarrow C$ is a coalgebra homomorphism. Next, for any $u_1, \ldots, u_n , v \in C$, we observe that
    \begin{align*}
        \langle u_1, u_2,  \ldots, u_n \rangle_\triangleleft \triangleleft v =~& \big(  ( \cdots  ((u_1 \triangleleft u_2) \triangleleft u_3) \cdots) \triangleleft u_n  \big) \triangleleft v \\
       =~&   ( \cdots  ( ( (u_1 \triangleleft v^{(1)}) \triangleleft (u_2 \triangleleft v^{(2)}) ) \triangleleft (u_3 \triangleleft v^{(3)}) ) \cdots ) \triangleleft (u_n \triangleleft v^{(n)}) \\
       =~& \langle u_1 \triangleleft v^{(1)}, u_2 \triangleleft v^{(2)}, \ldots, u_n \triangleleft v^{(n)} \rangle_\triangleleft
    \end{align*}
    which in turn implies that
    \begin{align*}
        \langle \langle u_1, \ldots, u_n \rangle_\triangleleft, v_1, \ldots, v_{n-1} \rangle_\triangleleft 
        &=  ( \cdots ( ( \langle u_1, \ldots, u_n \rangle_\triangleleft \triangleleft v_1 ) \triangleleft v_2 ) \cdots ) \triangleleft v_{n-1} \\
        &= \big\langle     ( \cdots  ( (u_1 \triangleleft v_1^{(1)}) \triangleleft v_2^{(1)}) \cdots) \triangleleft v_{n-1}^{(1)} , \ldots,    ( \cdots  ( (u_n \triangleleft v_1^{(n)}) \triangleleft v_2^{(n)}) \cdots) \triangleleft v_{n-1}^{(n)} \big\rangle \\
        &= \big\langle \langle u_1, v_1^{(1)}, \ldots, v_{n-1}^{(1)} \rangle, \ldots, \langle u_n, v_1^{(n)}, \ldots, v_{n-1}^{(n)} \rangle \big\rangle.
    \end{align*}
    This proves the self-distributivity condition (\ref{self-dis}). Hence $(C, \langle -, \ldots, - \rangle_\triangleleft)$ is a linear $n$-shelf.

  Finally, since $(C, \triangleleft)$ is a linear rack, there exists a coalgebra homomorphism $\widetilde{ \triangleleft} : C \otimes C \rightarrow C$ that makes $(C, \widetilde{\triangleleft})$ into a linear shelf satisfying additionally (\ref{addition-lin}). We define a linear map $\langle -, \ldots, - \rangle_{\widetilde{\triangleleft}} : C^{\otimes n} \rightarrow C$ similar to (\ref{some-n}) simply by replacing $\triangleleft$ by $\widetilde{\triangleleft}$. Then $(C,\langle -, \ldots, - \rangle_{\widetilde{\triangleleft}} )$ is a linear $n$-shelf. Moreover, we have
   \begin{align*}
        \langle \langle u, v_1^{(2)}, \ldots, v_{n-1}^{(2)} \rangle_\triangleleft, v_{n-1}^{(1)}, \ldots, v_1^{(1)} \rangle_{\widetilde{\triangleleft}} 
         =~&    \big( \cdots \big(   \big( (  \cdots   ( (u \triangleleft v_1^{(2)}) \triangleleft v_2^{(2)} ) \cdots ) \triangleleft v_{n-1}^{(2)} \big) ~ \! \widetilde{\triangleleft} ~ \! v_{n-1}^{(1)} \big) \cdots \big) ~ \! \widetilde{\triangleleft} ~ \! v_1^{(1)} \\
         =~& \varepsilon (v_1) \cdots \varepsilon (v_{n-1}) u
   \end{align*}
   and similarly, $ \langle \langle u, v_1^{(2)}, \ldots, v_{n-1}^{(2)} \rangle_{\widetilde{\triangleleft}}, v_{n-1}^{(1)}, \ldots, v_1^{(1)} \rangle_{{\triangleleft}} = \varepsilon (v_1) \cdots \varepsilon (v_{n-1}) u$. Hence, the first part follows.

   For the second part, we observe that
   \begin{align*}
       f ( \langle u_1, \ldots, u_n \rangle_\triangleleft) =~& f \big(    ( \cdots  ((u_1 \triangleleft u_2) \triangleleft u_3) \cdots) \triangleleft u_n   \big) \\
       =~&    ( \cdots  (( f(u_1) \triangleleft f(u_2 )) \triangleleft f(u_3) ) \cdots) \triangleleft f(u_n) =  \langle f(u_1), \ldots, f(u_n) \rangle_{\triangleleft'},
   \end{align*}
   for any $u_1, \ldots, u_n \in C$. This shows that $f$ is a homomorphism of corresponding linear $n$-racks.
\end{proof}

The above result can be regarded as the linear version of the construction given in Example \ref{rack-n-rack}. In the following result, we prove the linear version of Proposition \ref{prop-nrack-rack} under the cocommutativity assumption.

\begin{thm}\label{thm-linn-lin}
    Let $(C, \langle -, \ldots , - \rangle)$ be a linear $n$-rack where $C$ is cocommutative. Define a linear map $\triangleleft_{ \langle -, \ldots , - \rangle} : C^{\otimes (n-1)} \otimes C^{\otimes (n-1)} \rightarrow C^{\otimes (n-1)}$ by
    \begin{align}\label{first-n}
        (u_1 \otimes \cdots \otimes u_{n-1}) \triangleleft_{ \langle -, \ldots , - \rangle} (v_1 \otimes \cdots \otimes v_{n-1}) = \langle u_1 , v_1^{(1)}, \ldots, v_{n-1}^{(1)}   \rangle \otimes \cdots \otimes \langle u_{n-1} , v_1^{(n-1)}, \ldots, v_{n-1}^{(n-1)} \rangle,
    \end{align}
    for $u_1 \otimes \cdots \otimes u_{n-1} ~ \! ,  v_1 \otimes \cdots \otimes v_{n-1} \in C^{\otimes (n-1)}$. Then $(C^{\otimes (n-1)}, \triangleleft_{ \langle -, \ldots , - \rangle} ) $ is a linear rack. Moreover, if $f : (C, \langle -, \ldots , - \rangle) \rightarrow (C', \langle -, \ldots , - \rangle')$ is a homomorphism of cocommutative linear $n$-racks then the map $f^{\otimes (n-1)} : C^{\otimes (n-1)} \rightarrow C'^{\otimes (n-1)}$ is a homomorphism of linear racks from $(C^{\otimes (n-1)}, \triangleleft_{\langle -, \ldots , - \rangle })$ to $(C'^{\otimes (n-1)}, \triangleleft_{\langle -, \ldots , - \rangle' })$.
\end{thm}

\begin{proof} We divide the proof into the following steps.

  \noindent  {\sf Step 1.} {\em $\triangleleft_{\langle -, \ldots, - \rangle } $ is a coalgebra homomorphism.} For any $u_1 \otimes \cdots \otimes u_{n-1} ~\! , ~ \! v_1 \otimes \cdots \otimes v_{n-1} \in C^{\otimes (n-1)}$,
     \begin{align*}
       & \big( (\triangleleft_{\langle -, \ldots , - \rangle } \otimes \triangleleft_{\langle -, \ldots , - \rangle }) \circ \Delta_{C^{\otimes (n-1)} \otimes C^{\otimes (n-1)}} \big) ((u_1 \otimes \cdots \otimes u_{n-1}) \otimes (v_1 \otimes \cdots \otimes v_{n-1})) \\
        &= \big( (u_1^{(1)} \otimes \cdots \otimes u_{n-1}^{(1)}) \triangleleft_{\langle -, \ldots , - \rangle } (v_1^{(1)} \otimes \cdots \otimes v_{n-1}^{(1)}) \big) \otimes \big( (u_1^{(2)} \otimes \cdots \otimes u_{n-1}^{(2)}) \triangleleft_{\langle -, \ldots , - \rangle } (v_1^{(2)} \otimes \cdots \otimes v_{n-1}^{(2)}) \big) \\
        &= \big(  \langle u_1^{(1)} , v_1^{(1)(1)} , \ldots, v_{n-1}^{(1)(1)} \rangle \otimes \cdots \otimes \langle u_{n-1}^{(1)} , v_1^{(1)(n-1)} , \ldots, v_{n-1}^{(1)(n-1)} \rangle  \big) \\
        & \qquad \qquad \qquad \qquad \otimes \big(  \langle u_1^{(2)} , v_1^{(2)(1)} , \ldots, v_{n-1}^{(2)(1)} \rangle \otimes \cdots \otimes  \langle u_{n-1}^{(2)} , v_1^{(2)(n-1)} , \ldots, v_{n-1}^{(2)(n-1)} \rangle  \big) \\
        &= \big(   \langle u_1, v_1^{(1)}, \ldots, v_{n-1}^{(1)} \rangle^{(1)} \otimes \cdots \otimes \langle  u_{n-1}, v_1^{(n-1)}, \ldots, v_{n-1}^{(n-1)}  \rangle^{(1)} \big) \\
        & \qquad \qquad \qquad \qquad \otimes \big(   \langle u_1, v_1^{(1)}, \ldots, v_{n-1}^{(1)} \rangle^{(2)} \otimes \cdots \otimes \langle  u_{n-1}, v_1^{(n-1)}, \ldots, v_{n-1}^{(n-1)}  \rangle^{(2)}    \big) \\
        &= \Delta_{C^\otimes (n-1)} \big(   \langle u_1, v_1^{(1)}, \ldots, v_{n-1}^{(1)} \rangle \otimes \cdots \otimes  \langle u_{n-1}, v_1^{(n-1)}, \ldots, v_{n-1}^{(n-1)} \rangle \big) \\
        &= (\Delta_{C^\otimes (n-1)} \circ \triangleleft_{\langle -, \ldots , - \rangle }) ((u_1 \otimes \cdots \otimes u_{n-1}) \otimes (v_1 \otimes \cdots \otimes v_{n-1}))
     \end{align*}
     and
     \begin{align*}
         &(\varepsilon_{C^{\otimes (n-1)}} \circ \triangleleft_{\langle -, \ldots , - \rangle }) \big(   (u_1 \otimes \cdots \otimes u_{n-1}) \otimes (v_1 \otimes \cdots \otimes v_{n-1})   \big) \\
         &= \varepsilon_{C^{\otimes (n-1)}} \big(   \langle u_1, v_1^{(1)}, \ldots, v_{n-1}^{(1)} \rangle \otimes  \langle u_2, v_1^{(2)}, \ldots, v_{n-1}^{(2)} \rangle \otimes \cdots \otimes  \langle u_{n-1}, v_1^{(n-1)}, \ldots, v_{n-1}^{(n-1)} \rangle   \big) \\
         &= \varepsilon (u_1) \varepsilon (v_1^{(1)}) \cdots \varepsilon (v_{n-1}^{(1)}) 
         \varepsilon (u_2) \varepsilon (v_1^{(2)}) \cdots \varepsilon (v_{n-1}^{(2)}) \cdots \varepsilon (u_{n-1}) \varepsilon (v_1^{(n-1)}) \cdots \varepsilon (v_{n-1}^{(n-1)}) \\
         &= \varepsilon (u_1) \varepsilon (u_2) \cdots \varepsilon (u_{n-1}) \big(  \varepsilon (v_1^{(1)}) \varepsilon (v_1^{(2)}) \cdots \varepsilon (v_1^{(n-1)})    \big) \cdots \big(  \varepsilon (v_{n-1}^{(1)}) \varepsilon (v_{n-1}^{(2)}) \cdots \varepsilon (v_{n-1}^{(n-1)})   \big) \\
         &= \varepsilon (u_1) \varepsilon (u_2) \cdots \varepsilon (u_{n-1}) \varepsilon (v_1) \cdots \varepsilon (v_{n-1}) \\
         &= (\varepsilon_{C^{\otimes (n-1)} \otimes C^{\otimes (n-1)}}) \big(   (u_1 \otimes \cdots \otimes u_{n-1}) \otimes (v_1 \otimes \cdots \otimes v_{n-1})   \big).
     \end{align*}

     \medskip

   \noindent   {\sf Step 2.} { $\triangleleft_{\langle -, \ldots, - \rangle }$ {\em is a linear shelf.}} For any $u_1  \otimes \cdots \otimes u_{n-1}, ~ \! v_1  \otimes \cdots \otimes v_{n-1} ~ \! , ~ \! w_1 \otimes \cdots \otimes w_{n-1} \in C^{\otimes (n-1)}$, we have
   \begin{align*}
        &\big(  (u_1  \otimes \cdots \otimes u_{n-1}) \triangleleft_{\langle -, \ldots , - \rangle } ( v_1  \otimes \cdots \otimes v_{n-1}) \big) \triangleleft_{\langle -, \ldots , - \rangle } (w_1 \otimes \cdots \otimes w_{n-1}) \\
        &= \big( \langle u_1, v_1^{(1)}, \ldots, v_{n-1}^{(1)} \rangle \otimes \cdots \otimes \langle  u_{n-1}, v_1^{(n-1)}, \ldots, v_{n-1}^{(n-1)} \rangle   \big) \triangleleft_{\langle -, \ldots , - \rangle } (w_1 \otimes \cdots \otimes w_{n-1}) \\
        &= \big\langle \langle u_1, v_1^{(1)}, \ldots, v_{n-1}^{(1)} \rangle , w_1^{(1)}, \ldots, w_{n-1}^{(1)} \big\rangle \otimes \cdots \otimes  \big\langle \langle u_{n-1}, v_1^{(n-1)}, \ldots, v_{n-1}^{(n-1)} \rangle , w_1^{(n-1)}, \ldots, w_{n-1}^{(n-1)} \big\rangle \\
        &= \big\langle \langle u_1, w_1^{(1)(1)}, \ldots, w_{n-1}^{(1) (1)} \rangle, \langle 
  v_1^{(1)}, w_1^{(1) (2)}, \ldots, w_{n-1}^{(1) (2)}  \rangle, \ldots, \langle v_{n-1}^{(1)}, w_1^{(1) (n)} , \ldots, w_{n-1}^{(1)(n)} \rangle  \big\rangle  \otimes \cdots \otimes \\
  & \quad \big\langle \langle u_{n-1}, w_1^{(n-1)(1)}, \ldots, w_{n-1}^{(n-1) (1)} \rangle, \langle 
  v_1^{(n-1)}, w_1^{(n-1) (2)}, \ldots, w_{n-1}^{(n-1) (2)}  \rangle, \ldots, \langle v_{n-1}^{(n-1)}, w_1^{(n-1) (n)} , \ldots, w_{n-1}^{(n-1)(n)} \rangle  \big\rangle \\
  &= \big\langle \langle u_1, w_1^{(1)(1)}, \ldots, w_{n-1}^{(1)(1)} \rangle, \langle v_1, w_1^{(2)(1)}, \ldots, w_{n-1}^{(2)(1)}   \rangle^{(1)}, \ldots, \langle v_{n-1}, w_1^{(2)(n-1)}, \ldots, w_{n-1}^{(2)(n-1)}   \rangle^{(1)} \big\rangle  \otimes \cdots \otimes \\
  & \quad \big\langle \langle u_{n-1}, w_1^{(1)(n-1)}, \ldots, w_{n-1}^{(1)(n-1)} \rangle, \langle v_1, w_1^{(2)(1)}, \ldots, w_{n-1}^{(2)(1)}   \rangle^{(n-1)}, \ldots, \langle v_{n-1}, w_1^{(2)(n-1)}, \ldots, w_{n-1}^{(2)(n-1)}   \rangle^{(n-1)} \big\rangle \\
  &= \big( \langle u_1, w_1^{(1) (1)}, \ldots, w_{n-1}^{(1) (1)}    \rangle \otimes \cdots \otimes  \langle u_{n-1}, w_1^{ (1)(n-1)} , \ldots, w_{n-1}^{ (1)(n-1)} \otimes    \rangle    \big) \\
  & \qquad \qquad \quad \qquad \qquad \qquad \triangleleft_{\langle -, \ldots , - \rangle }   \big(  \langle v_1, w_1^{(2) (1)}, \ldots, w_{n-1}^{(2) (1)} \rangle \otimes \cdots \otimes  \langle v_{n-1}, w_1^{ (2)(n-1)} , \ldots, w_{n-1}^{ (2)(n-1)} \otimes  \rangle  \big) \\
  &= \big(  {\scriptstyle  (u_1 \otimes \cdots \otimes u_{n-1})} \triangleleft_{\langle -, \ldots , - \rangle } {\scriptstyle  (w_1^{(1)} \otimes \cdots \otimes w_{n-1}^{(1)})}   \big) \triangleleft_{\langle -, \ldots , - \rangle } \big( {\scriptstyle  (v_1 \otimes \cdots \otimes v_{n-1})} \triangleleft_{\langle -, \ldots , - \rangle } {\scriptstyle (w_1^{(2)} \otimes \cdots \otimes w_{n-1}^{(2)})}    \big).
   \end{align*}

   \medskip

  \noindent   {\sf Step 3.} {$ \triangleleft_{\langle -, \ldots, - \rangle } $ {\em is a linear rack.} } Since $(C, \langle -, \ldots, - \rangle)$ is a linear $n$-rack, there exists a coalgebra homomorphism $\ll - , \ldots , - \gg : C^{\otimes n} \rightarrow C$ that makes $(C, \ll - , \ldots , - \gg)$ into a linear $n$-shelf satisfying additionally (\ref{inv-prop}). We define a linear map $\triangleleft_{\ll - , \ldots , - \gg} : C^{\otimes (n-1)} \otimes C^{\otimes (n-1)} \rightarrow C^{\otimes (n-1)}$ by
  \begin{align*}
   {\scriptstyle (u_1 \otimes \cdots \otimes u_{n-1})} \triangleleft_{\ll - , \ldots , - \gg }  {\scriptstyle (v_1 \otimes \cdots \otimes v_{n-1}) }
   :=  \ll u_1, v_{n-1}^{(1)}, \ldots, v_1^{(1)} \gg \otimes \cdots \otimes \ll u_{n-1}, v_{n-1}^{(n-1)}, \ldots, v_1^{(n-1)} \gg.
  \end{align*}
  Then   $(C^{\otimes n-1}, \triangleleft_{ \ll - , \ldots, \gg})$ is a linear shelf. Additionally, 
     \begin{align*}
         &\big( (u_1 \otimes \cdots u_{n-1}) \triangleleft_{\langle -, \ldots, - \rangle} (v_1^{(2)} \otimes \cdots \otimes v_{n-1}^{(2)}) \big) \triangleleft_{ \ll - , \cdots ,- \gg } (v_1^{(1)} \otimes \cdots \otimes v_{n-1}^{(1)}) \\
         &= \big(   \langle u_1, v_1^{(2)(1)}, \ldots, v_{n-1}^{(2) (1)} \rangle \otimes \cdots \otimes \langle   u_{n-1} , v_1^{(2) (n-1)}, \ldots, v_{n-1}^{(2) (n-1)}  \rangle    \big) \triangleleft_{ \ll - , \cdots ,- \gg } (v_1^{(1)} \otimes \cdots \otimes v_{n-1}^{(1)}) \\
        & = \ll {\scriptstyle  \langle u_1, v_1^{(2)(1)}, \ldots, v_{n-1}^{(2) (1)} \rangle , v_{n-1}^{(1) (1)}, \ldots, v_{1}^{(1) (1)} } \gg \otimes \cdots \otimes \ll {\scriptstyle \langle u_{n-1}, v_1^{(2)(n-1)}, \ldots, v_{n-1}^{(2) (n-1)} \rangle , v_{n-1}^{(1) (n-1)}, \ldots, v_{1}^{(1) (n-1)} } \gg \\
         &= \big(   \varepsilon (v_1^{(1)}) \cdots \varepsilon ( v_{n-1}^{(1)}) u_1  \big) \otimes \cdots \otimes \big( \varepsilon (v_1^{(n-1)}) \cdots \varepsilon (v_1^{(n-1)}) u_{n-1}   \big) \\
         &= \big( \varepsilon (v_1^{(1)}) \cdots \varepsilon (v_1^{(n-1)})    \big) \cdots \big(   \varepsilon (v_{n-1}^{(1)}) \cdots \varepsilon (v_{n-1}^{(n-1)})   \big) (u_1 \otimes \cdots u_{n-1}) \\
        & = \varepsilon_{C^{\otimes (n-1)}} (v_1 \otimes \cdots \otimes v_{n-1}) (u_1 \otimes \cdots u_{n-1}).
     \end{align*}
     Similarly, we have
     \begin{align*}
        & \big( { (u_1 \otimes \cdots u_{n-1}) }\triangleleft_{\ll -, \ldots, - \gg} (v_1^{(2)} \otimes \cdots \otimes v_{n-1}^{(2)}) \big) \triangleleft_{ \langle - , \cdots ,- \rangle } (v_1^{(1)} \otimes \cdots \otimes v_{n-1}^{(1)}) \\
        & \qquad \quad = \varepsilon_{C^{\otimes (n-1)}} (v_1 \otimes \cdots \otimes v_{n-1}) ( { u_1 \otimes \cdots u_{n-1} }).
     \end{align*}
     This proves that $( C^{\otimes (n-1)} , \triangleleft_{\langle -, \ldots, - \rangle})$ is a linear rack.

     For the second part, we observe that $f : C \rightarrow C'$ is a coalgebra homomorphism implies that $f^{\otimes (n-1)} : C^{\otimes (n-1)} \rightarrow C'^{\otimes (n-1)}$ is also a coalgebra homomorphism. Also, we have
     \begin{align*}
        &f^{\otimes (n-1)} \big(    (u_1 \otimes \cdots \otimes u_{n-1}) \triangleleft_{\langle -, \ldots, - \rangle}  (v_1 \otimes \cdots \otimes v_{n-1})\big) \\
        %&= f (\langle u_1 , v_1^{(1)}, \ldots, v_{n-1}^{(1)}   \rangle ) \otimes \cdots \otimes f (\langle u_{n-1} , v_1^{(n-1)}, \ldots, v_{n-1}^{(n-1)} \rangle ) \\
        &= \langle  f(u_1), f(v_1^{(1)}), \ldots, f (v_{n-1}^{(1)})    \rangle' \otimes \cdots \otimes \langle  f (u_{n-1}), f (v_1^{(n-1)}), \ldots, f (v_{n-1}^{(n-1)}) \rangle' \\
        &= \langle  f(u_1), f(v_1)^{(1)}, \ldots, f (v_{n-1})^{(1)}    \rangle' \otimes \cdots \otimes \langle  f (u_{n-1}), f (v_1)^{(n-1)}, \ldots, f (v_{n-1})^{(n-1)} \rangle' \\
        &= (f^{\otimes (n-1)} (u_1 \otimes \cdots \otimes u_{n-1})) \triangleleft_{\langle -, \ldots, - \rangle'} (f^{\otimes (n-1)} (v_1 \otimes \cdots \otimes v_{n-1}))
     \end{align*}
     which shows that $f^{\otimes (n-1)} : C^{\otimes (n-1)} \rightarrow C'^{\otimes (n-1)}$ is a homomorphism of linear racks.
\end{proof}

Combining the above theorem with Proposition \ref{lebed-prop}, we obtain the following result.

\begin{thm}\label{thm-lin-n-soln}
    Let $(C, \langle - , \ldots, - \rangle)$ be a linear $n$-rack where $C$ is cocommutative. Then the linear map $R^{\langle -, \ldots, - \rangle} : C^{\otimes (n-1)} \otimes C^{\otimes (n-1)} \rightarrow C^{\otimes (n-1)} \otimes C^{\otimes (n-1)}$ defined by
    \begin{align*}
        &R^{\langle -, \ldots, - \rangle} \big( (u_1 \otimes \cdots \otimes u_{n-1}) \otimes (v_1 \otimes \cdots \otimes v_{n-1})  \big) \\
        & \qquad = (v_1^{(1)} \otimes \cdots \otimes v^{(1)}_{n-1}) \otimes ( \langle u_1, v_1^{(2)}, \ldots, v_{n-1}^{(2)} \rangle \otimes \cdots \otimes \langle u_{n-1}, v_1^{(n)}, \ldots, v_{n-1}^{(n)} \rangle)
    \end{align*}
    is a Yang-Baxter operator on the vector space $C^{\otimes (n-1)}$ with the inverse map
    \begin{align*}
        &(R^{\langle -, \ldots, - \rangle})^{-1}  \big( (u_1 \otimes \cdots \otimes u_{n-1}) \otimes (v_1 \otimes \cdots \otimes v_{n-1})  \big) \\
        & \qquad = \big(  \ll v_1, u_{n-1}^{(2)}, \ldots, u_{1}^{(2)} \gg \otimes  \cdots \otimes  \ll v_{n-1}, u_{n-1}^{(n)}, \ldots, u_{1}^{(n)} \gg   \big) \otimes (u_1^{(1)} \otimes \cdots \otimes u_{n-1}^{(1)}).
    \end{align*}
\end{thm}

\medskip

\medskip

Let $(X, \triangleleft)$ be a rack. Then by Example \ref{exam-rack-to-nrack}, one may construct an $n$-rack $(X, \langle -, \ldots, - \rangle)$ and hence a linear $n$-rack $({\bf k}[X], \langle -, \ldots, - \rangle)$ by Example \ref{exam-20}. On the other hand, from the given rack $(X, \triangleleft)$, by Example \ref{exam-lrack}, one have the linear rack $({\bf k}[X], \triangleleft)$ and hence a linear $n$-rack $({\bf k}[X], \langle -, \ldots, - \rangle_\triangleleft)$ by Theorem \ref{thm-lrack-lnrack}. In the next result, we show that the above two linear $n$-racks coincide, making the following diagram commutative:
\begin{align}
    \xymatrix{
     & \substack{n\text{-}\mathrm{rack} \\ (X, \langle -, \ldots, - \rangle )  } \ar[rr]^{\mathrm{Exam~} \ref{exam-20}} &  & \substack{ \text{linear } n\text{-rack} \\ ({\bf k}[X], \langle -, \ldots, - \rangle ) } \ar@{=}[dd]\\
     \substack{\text{rack} \\ (X, \triangleleft) } \ar@/^0.8pc/[ru]^{\mathrm{Exam~} \ref{exam-rack-to-nrack}} \ar@/_0.8pc/[rd]_{\mathrm{Exam~} \ref{exam-lrack}} &  &  & \\
      & \substack{ \text{linear rack} \\ ({\bf k}[X], \triangleleft ) }\ar[rr]_{\mathrm{Theorem~} \ref{thm-lrack-lnrack} }   & & \substack{ \text{linear } n\text{-rack} \\ ({\bf k}[X], \langle -, \ldots, - \rangle_\triangleleft ). }
    }
\end{align}

\begin{prop}
    Let $(X, \triangleleft)$ be a rack. Then the linear $n$-racks $({\bf k}[X], \langle -, \ldots, - \rangle)$ and  $({\bf k}[X], \langle -, \ldots, - \rangle_\triangleleft)$ are the same. 
\end{prop}

\begin{proof}
    It is easy to see that the coassociative counital coalgebra structures on both ${\bf k}[X]$ are the same. Next, for any $\lambda_1 x_1, \ldots, \lambda_n x_n \in {\bf k}[X]$ (with $\lambda_i \in {\bf k}$ and $x_i \in X$), we have
    \begin{align*}
         \langle \lambda_1 x_1, \ldots, \lambda_n x_n \rangle = \lambda_1 \cdots \lambda_n \langle x_1, \ldots, x_n \rangle 
         =~& \lambda_1 \cdots \lambda_n ~ \! ( \cdots ( ( x_1 \triangleleft x_2) \triangleleft x_3 ) \cdots ) \triangleleft x_n \\
         =~&  ( \cdots ( ( \lambda_1 x_1 \triangleleft \lambda_2 x_2) \triangleleft \lambda_3 x_3 ) \cdots ) \triangleleft \lambda_n x_n \\
         =~& \langle \lambda_1 x_1, \ldots, \lambda_n x_n \rangle_\triangleleft,
    \end{align*}
    which shows the desired result.
\end{proof}

On the other hand, suppose we start with an $n$-rack $(X, \langle -, \ldots, - \rangle)$. Then by Proposition \ref{prop-nrack-rack}, one may construct the rack $(X^{\times (n-1)} , \triangleleft)$ and hence a linear rack $({\bf k}[X^{\times (n-1)} ], \triangleleft)$ by Example \ref{exam-lrack}. On the other hand, from the given $n$-rack $(X, \langle -, \ldots , - \rangle)$, by Example \ref{exam-20}, one can construct the linear $n$-rack $({\bf k} [X], \langle -, \ldots, - \rangle)$ and hence a linear rack structure $( {\bf k} [X]^{\otimes (n-1)}, \triangleleft_{\langle -, \ldots, - \rangle})$ by Theorem \ref{thm-linn-lin}. The following result demonstrates the existence of a linear rack homomorphism $\psi$ between the above two linear racks, making the following diagram commutative:
\begin{align}
    \xymatrix{
     & \substack{   \text{rack} \\ (X^{\times (n-1)}, \triangleleft)   }\ar[rr]^{\mathrm{Exam~}\ref{exam-lrack}} &  & \substack{ \text{linear rack}\\ ({\bf k}[X^{\times (n-1)}], \triangleleft ) }\ar@{-->}[dd]^\psi\\
     \substack{ n\text{-rack} \\ (X, \langle - , \ldots, - \rangle)  } \ar@/^0.8pc/[ru]^{\mathrm{Proposition~} \ref{prop-nrack-rack}} \ar@/_0.8pc/[rd]_{\mathrm{Example ~} \ref{exam-20}} &  &  & \\
      & \substack{\text{linear } n\text{-rack}\\ ({\bf k}[X], \langle - , \ldots, -\rangle)  } \ar[rr]_{\mathrm{Theorem~}  \ref{thm-linn-lin}} & & \substack{ \text{linear rack }\\ ({\bf k}[X]^{\otimes (n-1)}, \triangleleft_{\langle -, \ldots, - \rangle}).    }
    }
\end{align}

\begin{thm}
    Let $(X, \langle -, \ldots , - \rangle)$ be an $n$-rack. Then the linear map $\psi : {\bf k}[X^{\times (n-1)} ] \rightarrow {\bf k} [X]^{\otimes (n-1)}$ defined by
    \begin{align*}
        \psi (\lambda (x_1, \ldots, x_{n-1})) = \lambda ~\! x_1 \otimes \cdots \otimes x_{n-1}, \text{ for } \lambda \in {\bf k} \text{ and } (x_1, \ldots, x_{n-1}) \in X^{\times (n-1)}
    \end{align*}
    is a linear rack homomorphism from $({\bf k}[X^{\times (n-1)} ], \triangleleft)$ to $( {\bf k} [X]^{\otimes (n-1)}, \triangleleft_{\langle - , \ldots, - \rangle})$.
\end{thm}

\begin{proof}
    For any $\lambda (x_1, \ldots, x_{n-1}) \in {\bf k} [ X^{\times (n-1)} ]$, we observe that
    \begin{align*}
        (\Delta_{{\bf k} [X]^{\otimes (n-1)} } \circ \psi) \big( \lambda (x_1, \ldots, x_{n-1}) \big) =~& \Delta_{{\bf k} [X]^{\otimes (n-1)} } ( \lambda x_1 \otimes \cdots \otimes x_{n-1}) \\
        =~& \lambda (x_1 \otimes \cdots \otimes x_{n-1}) \otimes (x_1 \otimes \cdots \otimes x_{n-1}) \\
        =~& (\psi \otimes \psi) \big(  \lambda (x_1, \ldots, x_{n-1}) \otimes   (x_1, \ldots, x_{n-1})  \big) \\
        =~& \big( (\psi \otimes \psi) \circ \Delta_{{\bf k} [X^{\times (n-1)}]}    \big)
    \end{align*}
    and 
    \begin{align*}
         ( \varepsilon_{{\bf k} [X]^{\otimes (n-1)} } \circ \psi) \big( \lambda (x_1, \ldots, x_{n-1}) \big) = \varepsilon_{{\bf k} [X]^{\otimes (n-1)} } (\lambda x_1 \otimes \cdots \otimes x_{n-1}) = \lambda = \varepsilon_{{\bf k} [X^{\times (n-1)}]} \big( \lambda (x_1, \ldots, x_{n-1}) \big).
    \end{align*}
    Thus $\psi : ( {\bf k}[X^{\times (n-1)} ], \Delta_{{\bf k} [X^{\times (n-1)}] }, \varepsilon_{ {\bf k} [X^{\times (n-1)}]}  ) \longrightarrow ( {\bf k} [X]^{\otimes (n-1)}, \Delta_{{\bf k} [X]^{\otimes (n-1)} } , \varepsilon_{{\bf k} [X]^{\otimes (n-1)} }  )$ is a coalgebra homomorphism. Finally, for any $\lambda (x_1, \ldots, x_{n-1}), ~\! \mu (y_1, \ldots, y_{n-1}) \in {\bf k} [X^{\times (n-1)}]$, we also have
    \begin{align*}
        \psi \big( \lambda (x_1, \ldots, x_{n-1}) \triangleleft  \mu (y_1, \ldots, y_{n-1})   \big) 
        =~& \lambda \mu ~\! \psi \big( (x_1, \ldots, x_{n-1}) \triangleleft (y_1, \ldots, y_{n-1})   \big) \\
        =~& \lambda \mu ~\! \psi \big( \langle x_1, y_1, \ldots, y_{n-1} \rangle, \ldots,  \langle x_{n-1}, y_1, \ldots, y_{n-1} \rangle    \big) \\
        =~& \lambda \mu ~\! \langle x_1, y_1, \ldots, y_{n-1} \rangle \otimes \cdots \otimes  \langle x_{n-1}, y_1, \ldots, y_{n-1} \rangle  \\
        =~& (\lambda x_1 \otimes \cdots \otimes x_{n-1}) \triangleleft_{ \langle - , \ldots , - \rangle} ( \mu y_1 \otimes \cdots \otimes y_{n-1} ) \\
        =~& \psi (  \lambda (x_1, \ldots, x_{n-1})) \triangleleft_{ \langle -, \ldots, - \rangle }  \psi ( \mu (y_1, \ldots, y_{n-1}) ).
    \end{align*}
    This shows that $\psi$ is a linear rack homomorphism.
\end{proof}

\medskip

In Example \ref{last-examm} we have seen that every Leibniz algebra gives rise to a linear rack. This construction has been generalized recently to the context of $3$-Leibniz algebras. More precisely, it has been shown by Xu and Sheng \cite{xu-sheng} that $3$-Leibniz algebras yield trilinear racks. In the following, we shall consider its $n$-ary generalization.

Let $\mathcal{L}$ be any vector space equipped with an $n$-linear operation $[-, \ldots, -] : \mathcal{L} \times \cdots \times \mathcal{L} \rightarrow \mathcal{L}$. For $\overline{\mathcal{L}} = {\bf k} \oplus \mathcal{L}$, we define a linear map $\langle -, \ldots, - \rangle : \overline{ \mathcal{L}}^{\otimes n} \rightarrow \overline{ \mathcal{L}}$ by
\begin{align*}
    \langle (\lambda_1, x_1), (\lambda_2, x_2), \ldots,(\lambda_n, x_n) \rangle = \big( \lambda_1 \cdots \lambda_n ~ \! , ~ \! \lambda_2 \cdots \lambda_n x_1 + [x_1, \ldots, x_n ] \big),
\end{align*}
for $(\lambda_1, x_1),  \ldots,(\lambda_n, x_n) \in \overline{\mathcal{L}}.$ Then we have the following result.

\begin{thm}\label{thm-nleib-lin}
    Let $(\mathcal{L}, [-, \ldots, -])$ be an $n$-Leibniz algebra. Then the pair $(\overline{\mathcal{L}}, \langle -, \ldots, - \rangle)$ is a linear $n$-rack, where the coassociative counital coalgebra structure on $\overline{\mathcal{L}}$ is defined in Example \ref{last-examm}. Moreover, if $\varphi : \mathcal{L} \rightarrow \mathcal{L}'$ is a homomorphism of $n$-Leibniz algebras then the map $\overline{\varphi} : \overline{\mathcal{L}} \rightarrow \overline{\mathcal{L}'}$, $\overline{\varphi} (\lambda, x) = (\lambda, \varphi(x))$, for $(\lambda, x) \in \overline{\mathcal{L}}$, is a homomorphism between the corresponding linear $n$-racks.
\end{thm}

\begin{proof}
     {\sf Step 1.} {\em ${\langle -, \ldots, - \rangle } : \overline{\mathcal{L}}^{\otimes n} \rightarrow \overline{\mathcal{L}}$ is a coalgebra homomorphism.} For $(\lambda_1, x_1), \ldots, (\lambda_n, x_n) \in \overline{\mathcal{L}}$, we have
     \begin{align*}
         &\big( ( \langle -, \ldots, - \rangle \otimes \langle -, \ldots, - \rangle ) \circ \Delta_{  \overline{\mathcal{L}}^{\otimes n}   } \big) ( (\lambda_1, x_1) \otimes \cdots \otimes (\lambda_n, x_n)) \\
         &= \langle   (\lambda_1, x_1)^{(1)}, \ldots, (\lambda_n, x_n)^{(1)}  \rangle \otimes \langle (\lambda_1, x_1)^{(2)}, \ldots, (\lambda_n, x_n)^{(2)}    \rangle \\
         &=  \langle (\lambda_1, x_1), \ldots , (\lambda_n, x_n ) \rangle \otimes \langle (1, 0), \ldots, (1, 0) \rangle  ~+ ~\langle (1, 0), (\lambda_2, x_2) , \ldots, (\lambda_n, x_n) \rangle \otimes \langle (0, x_1) , (1, 0), \ldots , (1, 0) \rangle \\
         & \qquad + \langle (1, 0), \ldots , (1, 0) \rangle \otimes \langle (0, x_1), \ldots, (0, x_n ) \rangle \\
         &= (\lambda_1 \cdots \lambda_n ~ \! , ~ \! \lambda_2 \cdots \lambda_n x_1 + [x_1, \ldots, x_n]) \otimes (1 , 0) ~ +~ (\lambda_2 \cdots \lambda_n, 0) \otimes (0, x_1) ~+~ (1, 0) \otimes (0, [x_1, \ldots, x_n])\\
         &= (\lambda_1 \cdots \lambda_n ~ \! , ~ \! \lambda_2 \cdots \lambda_n x_1 + [x_1, \ldots, x_n]) \otimes (1 , 0) ~ +~ (1, 0) \otimes (0, \lambda_2 \cdots \lambda_n x_1 + [x_1, \ldots, x_n]) \\
         &= \Delta (\lambda_1 \cdots \lambda_n ~ \!, ~ \! \lambda_2 \cdots \lambda_n x_1 + [x_1, \ldots, x_n ]) \\
         &= \Delta \langle (\lambda_1, x_1), \ldots, (\lambda_n, x_n) \rangle
     \end{align*}
     and 
     \begin{align*}
         \varepsilon \langle (\lambda_1, x_1), \ldots, (\lambda_n, x_n) \rangle =~& \varepsilon ( (\lambda_1 \cdots \lambda_n ~ \!, ~ \! \lambda_2 \cdots \lambda_n x_1 + [x_1, \ldots, x_n ])) \\
         =~& \lambda_1 \cdots \lambda_n = \varepsilon_{ \overline{\mathcal{L}}^{\otimes n} } ( (\lambda_1, x_1) \otimes \cdots \otimes (\lambda_n, x_n)    ).
     \end{align*}

     \medskip

  \noindent   {\sf Step 2.} {\em  $( \overline{\mathcal{L}}, {\langle -, \ldots, - \rangle } )$ is a linear $n$-shelf.}  First, note that the map $\Delta^{n-1} : \overline{\mathcal{L}} \rightarrow \overline{\mathcal{L}}^{\otimes n}$ is given by
     \begin{align*}
         \Delta^{n-1} (\lambda, x) = (\lambda , x) \otimes (1, 0) \otimes \cdots \otimes (1, 0) ~ + ~ \sum_{i=1}^{n-1} (1, 0)^{\otimes i} \otimes (0, x) \otimes (1, 0)^{\otimes (n-1-i)}.
     \end{align*}
     Hence, for any $(\lambda_1, x_1), \ldots , (\lambda_n, x_n), (\mu_1, y_1), \ldots, (\mu_{n-1}, y_{n-1}) \in \overline{\mathcal{L}}$,
     \begin{align*}
         &\big\langle \langle (\lambda_1, x_1), (\mu_1, y_1)^{(1)}, \ldots, (\mu_{n-1}, y_{n-1})^{(1)} \rangle, \ldots, \langle (\lambda_n, x_n), (\mu_1, y_1)^{(n)}, \ldots, (\mu_{n-1}, y_{n-1})^{(n)} \rangle \big\rangle \\
         &=  \big\langle \langle (\lambda_1, x_1), (\mu_1, y_1), \ldots , (\mu_{n-1}, y_{n-1}) \rangle ~ \! , ~ \! \langle (\lambda_2, x_2), (1, 0), \ldots, (1, 0) \rangle ~ \! , \ldots, ~ \! \langle (\lambda_n, x_n), (1, 0), \ldots, (1, 0) \rangle   \rangle \\
        & \quad + \big\langle \langle (\lambda_1, x_1), (1, 0), \ldots, (1, 0) \rangle ~ \!, ~ \! \langle (\lambda_2, x_2), (0, y_1), \ldots, (0, y_{n-1}) \rangle ~ \!, \ldots, \langle (\lambda_n, x_n ) , (1, 0), \ldots, (1, 0) \rangle \big\rangle  \\ 
        & \qquad \! \! \vdots \\
        & \quad + \big\langle \langle (\lambda_1, x_1), (1, 0), \ldots, (1, 0) \rangle ~ \! , ~ \! \langle (\lambda_2, x_2), (1, 0), \ldots, (1, 0)  \rangle ~ \! , \ldots, \langle (\lambda_n, x_n) , (0, y_1), \ldots, (0, y_{n-1})   \rangle \big\rangle \\
         &= \langle (\lambda_1 \mu_1 \cdots \mu_{n-1} ~ \!, ~ \! \mu_1 \cdots \mu_{n-1} x_1 + [x_1, y_1, \ldots, y_{n-1}]), (\lambda_2, x_2) , \ldots , (\lambda_n, x_n) \rangle \\
         & \quad  + \langle (\lambda_1, x_1), (0, [x_2, y_1, \ldots, y_{n-1}]), \ldots, (\lambda_n, x_n)  \rangle + \cdots  \\
         & \quad + \langle  (\lambda_1, x_1), \ldots, (\lambda_{n-1}, x_{n-1}), (0, [x_n, y_1, \ldots, y_{n-1}] ) \rangle 
\end{align*}
         \begin{align*}
         &= \big( \lambda_1 \cdots \lambda_n \mu_1 \cdots \mu_{n-1} ~ \! , ~ \! \lambda_2 \cdots \lambda_n \mu_1 \cdots \mu_{n-1} x_1 + \lambda_2 \cdots \lambda_n [x_1, y_1, \ldots, y_{n-1}] + \mu_1 \cdots \mu_{n-1} [x_1, \ldots, x_n] \\
         & \quad + [[x_1 , y_1, \ldots, y_{n-1}], x_2, \ldots, x_n] + [ x_1, [x_2, y_1, \ldots, y_{n-1}], \ldots, x_n] + \cdots + [x_1, \ldots, x_{n-1}, [x_n, y_1, \ldots, y_{n-1}]] \big) \\
         &= \big( \lambda_1 \cdots \lambda_n \mu_1 \cdots \mu_{n-1} ~ \! , ~ \! \lambda_2 \cdots \lambda_n \mu_1 \cdots \mu_{n-1} x_1 + + \lambda_2 \cdots \lambda_n [x_1, y_1, \ldots, y_{n-1}] + \mu_1 \cdots \mu_{n-1} [x_1, \ldots, x_n] \\ & \quad + [[ x_1, \ldots, x_n ], y_1, \ldots, y_{n-1} ] \big) \\
         &= \langle    (\lambda_1 \cdots \lambda_n ~ \! , \lambda_2 \cdots \lambda_n x_1 + [x_1, \ldots, x_n ]) ,  (\mu_1, y_1), \ldots, (\mu_{n-1}, y_{n-1}) \rangle \\
         &= \langle \langle (\lambda_1, x_1), \ldots , (\lambda_n , x_n) \rangle , (\mu_1, y_1), \ldots, (\mu_{n-1}, y_{n-1}) \rangle.
     \end{align*}
     %This proves the self-distributivity condition (\textcolor{red}{left}).

     \medskip

  \noindent   {\sf Step 3.} {\em $( \overline{\mathcal{L}}, {\langle -, \ldots, - \rangle } )$ is a linear $n$-rack.} We define a linear map $\ll - , \ldots, - \gg : \overline{\mathcal{L}}^{\otimes n} \rightarrow \overline{ \mathcal{L}}$ by 
     \begin{align*}
         \ll (\lambda_1, x_1), \ldots, (\lambda_n, x_n) \gg := (\lambda_1 \cdots \lambda_n ~ \! , \lambda_2 \cdots \lambda_n x_1 - [x_1, x_n, \ldots, x_2]),
     \end{align*}
     for $(\lambda_1, x_1), \ldots, (\lambda_n, x_n) \in \overline{\mathcal{L}}$. Similar to Step 1 and Step 2, one can show that $(\overline{\mathcal{L}}, \ll -, \ldots, - \gg)$ is a linear $n$-shelf. Finally, we have
     \begin{align*}
         &\ll \langle  (\lambda, x), (\mu_1, y_1)^{(2)}, \ldots, (\mu_{n-1}, y_{n-1})^{(2)} \rangle, (\mu_{n-1}, y_{n-1})^{(1)}, \ldots, (\mu_1, y_1)^{(1)} \gg \\
         &= \ll \langle (\lambda, x) , (1, 0), \ldots, (1, 0) \rangle, {\scriptstyle  (\mu_{n-1}, y_{n-1}), \ldots , (\mu_1, y_1)} \gg + \ll \langle (\lambda, x), (0, y_1), \ldots, (0, y_{n-1}) \rangle, {\scriptstyle 
 (1, 0), \ldots, (1, 0) } \gg \\
         &= \ll (\lambda, x) ,  (\mu_{n-1}, y_{n-1}), \ldots , (\mu_1, y_1) \gg + \ll (0, [x, y_1, \ldots, y_{n-1}]) , (1, 0), \ldots, (1, 0) \gg \\
         &= (\lambda \mu_1 \cdots \mu_{n-1} ~ \! , \mu_1 \cdots \mu_{n-1} x - [x, y_1, \ldots, y_{n-1}]) + (0, [x, y_1, \ldots, y_{n-1}]) \\
         &= ( \lambda \mu_1 \cdots \mu_{n-1} ~ \! , ~ \! \mu_1 \cdots \mu_{n-1} x) \\
         &= \mu_1 \cdots \mu_{n-1} (\lambda , x) = \varepsilon (\mu_1, y_1) \cdots \varepsilon (\mu_{n-1}, y_{n-1}) (\lambda , x)
     \end{align*}
     and similarly, 
     \begin{align*}
         \langle \ll (\lambda , x), (\mu_1, y_1)^{(2)}, \ldots, (\mu_{n-1}, y_{n-1})^{(2)} \gg , {\scriptstyle (\mu_{n-1}, y_{n-1})^{(1)}, \ldots, (\mu_1, y_1)^{(1)} } \rangle = \varepsilon (\mu_1, y_1) \cdots \varepsilon (\mu_{n-1}, y_{n-1}) (\lambda , x).
     \end{align*}
     This shows that $(\overline{\mathcal{L}}, \langle -, \ldots, - \rangle )$ is a linear $n$-rack. 

     For the second part, it is easy to see that $\Delta_{ \overline{ \mathcal{L}'}} \circ \overline{\varphi} = (\overline{\varphi} \otimes \overline{\varphi}) \circ \Delta_{ \overline{ \mathcal{L}}}$ and $\varepsilon_{ \overline{ \mathcal{L}'}} \circ \overline{\varphi} = \varepsilon_{\overline{ \mathcal{L}}}$, which shows that $\overline{ \varphi}: \overline{ \mathcal{L}} \rightarrow \overline{ \mathcal{L}'}$ is a coalgebra homomorphism. For any $(\lambda_1, x_1), \ldots, (\lambda_n, x_n) \in \overline{ \mathcal{L}}$, we also have
     \begin{align*}
         \overline{\varphi} \big(  \langle (\lambda_1, x_1), \ldots, (\lambda_n, x_n) \rangle_{\overline{ \mathcal{L}}}   \big) =~& \overline{\varphi} \big( \lambda_1 \cdots \lambda_n ~\!, ~\! \lambda_2 \cdots \lambda_n x_1 + [x_1, \ldots, x_n]    \big) \\
         =~&  \big( \lambda_1 \cdots \lambda_n ~\!, ~\! \lambda_2 \cdots \lambda_n \varphi (x_1) + [\varphi(x_1), \ldots, \varphi( x_n)]'   \big) \\
         =~& \langle  \overline{\varphi} (\lambda_1, x_1) , \ldots,   \overline{\varphi} (\lambda_n, x_n) \rangle_{\overline{ \mathcal{L}'}}.
     \end{align*}
     This concludes the proof.
\end{proof}

By combining Theorem \ref{thm-lin-n-soln} and Theorem  \ref{thm-nleib-lin}, we obtain the following Yang-Baxter operator from an $n$-Leibniz algebra.

\begin{thm}
    Let $(\mathcal{L}, [-, \ldots, -])$ be an $n$-Leibniz algebra. Then we define a linear map 
    \begin{align*}
    R : \overline{\mathcal{L}}^{\otimes (n-1)} \otimes \overline{\mathcal{L}}^{\otimes (n-1)} \rightarrow \overline{\mathcal{L}}^{\otimes (n-1)} \otimes \overline{\mathcal{L}}^{\otimes (n-1)}  \text{ by }
    \end{align*}
    \begin{align*}
        &R  \big(  ( (\lambda_1, x_1) \otimes \cdots \otimes (\lambda_{n-1}, x_{n-1}) )  \otimes   ( (\mu_1, y_1) \otimes \cdots \otimes (\mu_{n-1}, y_{n-1}) )  \big) \\
        & \qquad =  ( (\mu_1, y_1) \otimes \cdots \otimes (\mu_{n-1}, y_{n-1}) ) \otimes  ( (\lambda_1, x_1) \otimes \cdots \otimes (\lambda_{n-1}, x_{n-1}) ) \\
        & \quad \qquad + \sum_{i=1}^{n-1} \big( \underbrace{ (1, 0) \otimes \cdots \otimes (1, 0)}_{n-1 \mathrm{~ times}} \big) \otimes \big(  (\lambda_1, x_1) \otimes \cdots \otimes \underbrace{ (0, [x_i, y_1, \ldots, y_{n-1}])}_{i\text{-}\mathrm{th~ place}} \otimes \cdots \otimes (\lambda_{n-1}, x_{n-1}) \big).
    \end{align*}
    Then $R$ is a Yang-Baxter operator on the vector space $\overline{\mathcal{L}}^{\otimes (n-1)}$.
\end{thm}

Note that the Yang-Baxter operator $R : \overline{\mathcal{L}}^{\otimes (n-1)} \otimes \overline{\mathcal{L}}^{\otimes (n-1)} \rightarrow \overline{\mathcal{L}}^{\otimes (n-1)} \otimes \overline{\mathcal{L}}^{\otimes (n-1)}$ constructed above coincides with the one given in Proposition \ref{nl-cnl} (which was obtained through the central $n$-Leibniz algebra $\overline{\mathcal{L}}$).

Let $(\mathcal{L}, [-, \ldots, - ])$ be an $n$-Leibniz algebra. Then by Proposition \ref{funda-leibniz}, one may construct the Leibniz algebra $(\mathcal{L}^{\otimes (n-1)}, \{ ~, ~ \})$ and hence a linear rack $(\overline{ \mathcal{L}^{\otimes (n-1)}}, \triangleleft  )$ by Example \ref{last-examm}. On the other hand, from the given $n$-Leibniz algebra $(\mathcal{L}, [-, \ldots, -])$, by Theorem \ref{thm-nleib-lin}, one obtains the cocommutative linear $n$-rack $(\overline{\mathcal{L}}, \langle - , \ldots, - \rangle)$ and hence a linear rack structure $(\overline{\mathcal{L}}^{\otimes (n-1)}, \triangleleft_{ \langle -, \ldots, - \rangle})$, by Theorem \ref{thm-linn-lin}. It is important to note that the map $\eta : \overline{ \mathcal{L}^{\otimes (n-1)}} \rightarrow \overline{\mathcal{L}}^{\otimes (n-1)}$ defined in (\ref{eta-map}) is not a linear rack homomorphism. This is not even a coalgebra homomorphism.

%In the following, we shall show that $(\overline{ \mathcal{L}^{\otimes (n-1)}}, \triangleleft)$ is a linear subrack of  $(\overline{\mathcal{L}}^{\otimes (n-1)}, \triangleleft_{ \langle -, \ldots, - \rangle})$.

%\begin{prop}
 %   Let $(\mathcal{L}, [-, \ldots, - ])$ be an $n$-Leibniz algebra. We define an injective map
 %   \begin{align*}
  %     \eta: \overline{ \mathcal{L}^{\otimes (n-1)}} \rightarrow \overline{\mathcal{L}}^{\otimes (n-1)} ~~ \text{ by } ~~ \eta ( (\lambda, x_1 \otimes \cdots \otimes x_{n-1})) = \lambda \underbrace{(1, 0) \otimes \cdots \otimes (1, 0)}_{(n-1) \mathrm{~ times}} +~\! (0, x_1) \otimes \cdots \otimes (0, x_{n-1}),
  %  \end{align*}
  %  for $ (\lambda, x_1 \otimes \cdots \otimes x_{n-1}) \in {\bf k} \oplus \mathcal{L}^{\otimes (n-1)} = \overline{ \mathcal{L}^{\otimes (n-1)}} $. Then $\eta$ is a homomorphism of linear racks from $(\overline{ \mathcal{L}^{\otimes (n-1)}}, \triangleleft )$ to $(\overline{\mathcal{L}}^{\otimes (n-1)}, \triangleleft_{ \langle -, \ldots, - \rangle})$.
%\end{prop}

%\begin{proof}
%    \textcolor{red}{proof}
%\end{proof}

%\begin{align}
 %   \xymatrix{
%     & \text{Leibniz algebra} \ar[rr] &  & \text{linear rack} \ar@{-->}[dd]^\eta \\
 %  \substack{ n\text{-Leibniz ~algebra} \\ (\mathcal{L}, [-, \ldots, -]) } \ar[ru] \ar[rd] &  &  & \\
 %     &  \text{linear } n\text{-rack} \ar[rr] & & \text{linear rack }
 %   }
%\end{align}

\section{A higher analogue to the Yang-Baxter equation}\label{sec5}
In this section, we first introduce the $n$-Yang-Baxter equation (whose invertible solutions are called $n$-Yang-Baxter operators) as the $n$-ary generalization of the Yang-Baxter equation. We show that (central) $n$-Leibniz algebras and cocommutative linear $n$-racks provide examples of $n$-Yang-Baxter operators. We observe that a Yang-Baxter operator on a vector space naturally gives rise to an $n$-Yang-Baxter operator on the same vector space, and conversely, any $n$-Yang-Baxter operator induces a Yang-Baxter operator (however, not on the same vector space). Subsequently, we also consider the set-theoretical version of the $n$-Yang-Baxter equation, whose bijective solutions are called set-theoretical $n$-solutions. Any $n$-rack gives rise to a set-theoretical $n$-solution.

\begin{defn}
Let $V$ be a vector space. For any $n \geq 2$, a linear map $S : V^{\otimes n} \rightarrow V^{\otimes n}$ is said to be a {\bf pre-$n$-Yang-Baxter operator} (or a {\bf pre-$n$-braided operator}) on the vector space $V$ if
\begin{align}\label{nybe}
    (S \otimes \mathrm{Id}^{\otimes ( n-1)})& \cdots (\mathrm{Id}^{\otimes (n-2)} \otimes S \otimes \mathrm{Id})  (\mathrm{Id}^{\otimes (n-1)} \otimes S) (S \otimes \mathrm{Id}^{\otimes (n-1)}) \\
    &= (\mathrm{Id}^{\otimes (n-1)} \otimes S)  (S \otimes \mathrm{Id}^{\otimes (n-1)}) \cdots (\mathrm{Id}^{\otimes (n-2)} \otimes S \otimes \mathrm{Id})  (\mathrm{Id}^{\otimes (n-1)} \otimes S) \nonumber
\end{align}
as a linear map from $V^{\otimes (2n-1)}$ to itself. Additionally, if $S$ is invertible, then it is called an {\bf $n$-Yang-Baxter operator} (or {\bf $n$-braided operator}) on the vector space $V$. The equation (\ref{nybe}) is referred to as the {\bf $n$-Yang-Baxter equation}.
\end{defn}

\begin{remark} (When $n=2$)
    It follows from the above definition that a $2$-Yang-Baxter operator is nothing but a Yang-Baxter operator. Thus, $n$-Yang-Baxter operators naturally generalize Yang-Baxter operators.
\end{remark}

Suppose a linear map $S : V^{\otimes 3} \rightarrow V^{\otimes 3}$ is represented by the diagram ~~~ \!  \!  \begin{tikzpicture}[scale=0.4]
    \draw[thick, blue]
  (8,2) .. controls (8,0.5) and (10,1) .. (10,0)
  (8,0) .. controls (8,1.25)and  (9,0.75).. (9,2)
  (9,0) .. controls (8.75,0.8) and (10.1,1.20) .. (10,2);
    \end{tikzpicture}, then the equation (\ref{nybe}) for $n=3$ can be understood by the equality of the following diagrams:

    \medskip

    \medskip

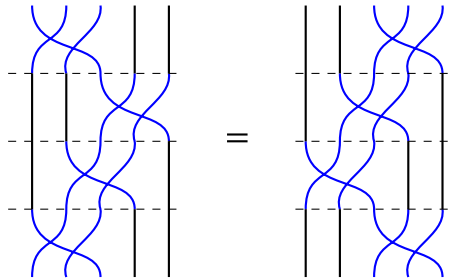
\begin{center}
\begin{tikzpicture}[scale=0.45]
\draw[thick, blue]
  (9,2) .. controls (9,0.5) and (11,1) .. (11,0)
  (9,0) .. controls (9,1.25)and  (10,0.75).. (10,2)
  (10,0) .. controls (9.75,0.8) and (11.1,1.20) .. (11,2)
  (8,0) .. controls (8,-1.5) and (10,-1) .. (10,-2)
  (8,-2) .. controls (8,-0.75)and  (9,-1.25).. (9,0)
  (9,-2) .. controls (8.75,-1.20) and (10.1,-0.80) .. (10,0)
  (7,-2) .. controls (7,-3.5) and (9,-3) .. (9,-4)
  (7,-4) .. controls (7,-2.75)and  (8,-3.25).. (8,-2)
  (8,-4) .. controls (7.75,-3.20) and (9.1,-2.80) .. (9,-2)
  (9,-4) .. controls (9,-5.5) and (11,-5) .. (11,-6)
  (9,-6) .. controls (9,-4.75)and  (10,-5.25).. (10,-4)
  (10,-6) .. controls (9.75,-5.20) and (11.1,-4.80) .. (11,-4)
  
  (-1,2) .. controls (-1,0.5) and (1,1) .. (1,0)
  (-1,0) .. controls (-1,1.25)and  (0,0.75).. (0,2)
  (0,0) .. controls (-0.25,0.8) and (1.1,1.20) .. (1,2)
  (1,0) .. controls (1,-1.5) and (3,-1) .. (3,-2)
  (1,-2) .. controls (1,-0.75)and  (2,-1.25).. (2,0)
  (2,-2) .. controls (1.75,-1.20) and (3.1,-0.80) .. (3,0)
  (0,-2) .. controls (0,-3.5) and (2,-3) .. (2,-4)
  (0,-4) .. controls (0,-2.75)and  (1,-3.25).. (1,-2)
  (1,-4) .. controls (0.75,-3.20) and (2.1,-2.80) .. (2,-2)
  (-1,-4) .. controls (-1,-5.5) and (1,-5) .. (1,-6)
  (-1,-6) .. controls (-1,-4.75)and  (0,-5.25).. (0,-4)
  (0,-6) .. controls (-0.25,-5.20) and (1.1,-4.80) .. (1,-4);
  \draw[thick]
  (7,0) -- (7,2)
  (8,0) -- (8,2)
  (7,0) -- (7,-2)
  (11,0) -- (11,-2)
  (10,-2) -- (10,-4)
  (11,-2) -- (11,-4)
  (7,-4) -- (7,-6)
  (8,-4) -- (8,-6)
  
  (2,2) -- (2,0)
  (3,2) -- (3,0)
(-1,0) -- (-1,-2)
(0,0) -- (0,-2)
(-1,-2)--(-1,-4)
(3,-2) -- (3,-4)
(2,-4)--(2,-6)
(3,-4) -- (3,-6)
(4.7,-2) -- (5.3,-2)
(4.7,-1.8) -- (5.3,-1.8);
\draw[thin, dashed]
(6.7,0) -- (11.3,0)
(6.7,-2) -- (11.3,-2)
(6.7,-4) -- (11.3,-4)
(-1.7,0) -- (3.3,0)
(-1.7,-2) -- (3.3,-2)
(-1.7,-4) -- (3.3,-4);
%\filldraw[red] (1.3,0.82) circle[radius=0.07];
%\node at (5,-2) {$\xlongequal{\hspace{8pt}}$};
%\node at (5,-7) {Yang-Baxter equation for n=3 Leibniz algebra};
\end{tikzpicture}
\end{center}
where all the diagrams (here and in the rest of the paper) can be read from top to bottom.

Similarly, if a linear map $S : V^{\otimes 4} \rightarrow V^{\otimes 4}$ is represented by the diagram ~~~ \!  \!
 \begin{tikzpicture}[scale=0.4]
    \draw[thick, blue]
  (0,0) .. controls (0,-2) and (3,-1) .. (3,-2.5)
  (1,0) .. controls (1,-1.5) and (0,-1) .. (0,-2.5)
  (2,0) .. controls (2,-1.5) and (1,-1) .. (1,-2.5)
  (3,0) .. controls (3,-1.5) and (2,-1) .. (2,-2.5);
  %\node at (1,-5) {Right};
    \end{tikzpicture}, then the equation (\ref{nybe}) for $n = 4$ can be understood by the following equality:

\medskip

\medskip
    
    
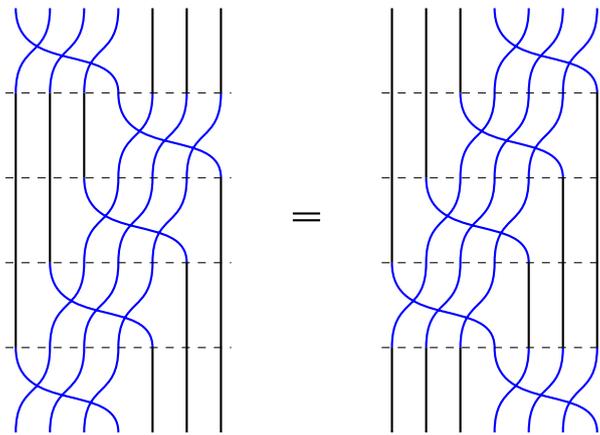
\begin{center}
    \begin{tikzpicture}[scale=0.45]

  \draw[thick, blue]
  (-1,0) .. controls (-1,-2) and (2,-1) .. (2,-2.5)
  (0,0) .. controls (0,-1.5) and (-1,-1) .. (-1,-2.5)
  (1,0) .. controls (1,-1.5) and (0,-1) .. (0,-2.5)
  (2,0) .. controls (2,-1.5) and (1,-1) .. (1,-2.5)
  
  (2,-2.5) .. controls (2,-4.5) and (5,-3.5) .. (5,-5)
  (3,-2.5) .. controls (3,-4) and (2,-3.5) .. (2,-5)
  (4,-2.5) .. controls (4,-4) and (3,-3.5) .. (3,-5)
  (5,-2.5) .. controls (5,-4) and (4,-3.5) .. (4,-5)
  
  (1,-5) .. controls (1,-7) and (4,-6) .. (4,-7.5)
  (2,-5) .. controls (2,-6.5) and (1,-6) .. (1,-7.5)
  (3,-5) .. controls (3,-6.5) and (2,-6) .. (2,-7.5)
  (4,-5) .. controls (4,-6.5) and (3,-6) .. (3,-7.5)
  
  (0,-7.5) .. controls (0,-9.5) and (3,-8.5) .. (3,-10)
  (1,-7.5) .. controls (1,-9) and (0,-8.5) .. (0,-10)
  (2,-7.5) .. controls (2,-9) and (1,-8.5) .. (1,-10)
  (3,-7.5) .. controls (3,-9) and (2,-8.5) .. (2,-10)
  
  (-1,-10) .. controls (-1,-12) and (2,-11) .. (2,-12.5)
  (0,-10) .. controls (0,-11.5) and (-1,-11) .. (-1,-12.5)
  (1,-10) .. controls (1,-11.5) and (0,-11) .. (0,-12.5)
  (2,-10) .. controls (2,-11.5) and (1,-11) .. (1,-12.5)
  
  (13,0) .. controls (13,-2) and (16,-1) .. (16,-2.5)
  (14,0) .. controls (14,-1.5) and (13,-1) .. (13,-2.5)
  (15,0) .. controls (15,-1.5) and (14,-1) .. (14,-2.5)
  (16,0) .. controls (16,-1.5) and (15,-1) .. (15,-2.5)
  
  (12,-2.5) .. controls (12,-4.5) and (15,-3.5) .. (15,-5)
  (13,-2.5) .. controls (13,-4) and (12,-3.5) .. (12,-5)
  (14,-2.5) .. controls (14,-4) and (13,-3.5) .. (13,-5)
  (15,-2.5) .. controls (15,-4) and (14,-3.5) .. (14,-5)
  
  (11,-5) .. controls (11,-7) and (14,-6) .. (14,-7.5)
  (12,-5) .. controls (12,-6.5) and (11,-6) .. (11,-7.5)
  (13,-5) .. controls (13,-6.5) and (12,-6) .. (12,-7.5)
  (14,-5) .. controls (14,-6.5) and (13,-6) .. (13,-7.5)
  
  (10,-7.5) .. controls (10,-9.5) and (13,-8.5) .. (13,-10)
  (11,-7.5) .. controls (11,-9) and (10,-8.5) .. (10,-10)
  (12,-7.5) .. controls (12,-9) and (11,-8.5) .. (11,-10)
  (13,-7.5) .. controls (13,-9) and (12,-8.5) .. (12,-10)
  
  (13,-10) .. controls (13,-12) and (16,-11) .. (16,-12.5)
  (14,-10) .. controls (14,-11.5) and (13,-11) .. (13,-12.5)
  (15,-10) .. controls (15,-11.5) and (14,-11) .. (14,-12.5)
  (16,-10) .. controls (16,-11.5) and (15,-11) .. (15,-12.5);
  
  \draw[thick]
   (3,0) -- (3,-2.5)
  (4,0) -- (4,-2.5)
  (5,0) -- (5,-2.5)
  (0,-2.5) -- (0,-5)
  (0,-5) -- (0,-7.5)
  (1,-2.5) -- (1,-5)
  (-1,-2.5) -- (-1,-5)
  (-1,-5) -- (-1,-7.5)
  (5,-5) -- (5,-7.5)
  (-1,-7.5) -- (-1,-10)
  (4,-7.5) -- (4,-10)
  (5,-7.5) -- (5,-10)
(3,-10) -- (3,-12.5)
 (4,-10) -- (4,-12.5)
(5,-10) -- (5,-12.5)

(12,0) -- (12,-2.5)
(10,0) -- (10,-2.5)
(11,0) -- (11,-2.5)
(11,-2.5) -- (11,-5)
(10,-2.5) -- (10,-5)
(16,-2.5) -- (16,-5)
(10,-5) -- (10,-7.5)
(16,-5) -- (16,-7.5)
(15,-5) -- (15,-7.5)
(16,-7.5) -- (16,-10)
(14,-7.5) -- (14,-10)
(15,-7.5) -- (15,-10)
(12,-10) -- (12,-12.5)
(10,-10) -- (10,-12.5)
(11,-10) -- (11,-12.5)
(7.1, -6.25) -- (7.9, - 6.25) 
(7.1, -6.05) -- (7.9, - 6.05);
\draw[thin, dashed]
(-1.3,-2.5) -- (5.3,-2.5)
(-1.3,-5) -- (5.3,-5)
(-1.3,-7.5) -- (5.3,-7.5)
(-1.3,-10) -- (5.3,-10)

(9.7,-2.5) -- (16.3,-2.5)
(9.7,-5) -- (16.3,-5)
(9.7,-7.5) -- (16.3,-7.5)
(9.7,-10) -- (16.3,-10)
;
%\filldraw[red] (1.3,0.82) circle[radius=0.07];
%\node at (7.5,-6.25) {$\xlongequal{\hspace{8pt}}$};
\end{tikzpicture}
\end{center}

\begin{exam}
    (i) As mentioned earlier, a Yang-Baxter operator is a $2$-Yang-Baxter operator.

    (ii) For any $n \geq 2$, the identity map $\mathrm{Id}_{V^{\otimes n}} : V^{\otimes n} \rightarrow V^{\otimes n}$ and the map $\mathcal{F} :  V^{\otimes n} \rightarrow V^{\otimes n}$ defined by
    \begin{align*}
        \mathcal{F} (x_1 \otimes x_2 \otimes \cdots \otimes x_n) = x_2 \otimes \cdots \otimes x_n \otimes x_1
    \end{align*}
    both are $n$-Yang-Baxter operators on $V$.
\end{exam}

\begin{exam}
    Let $A$ be an unital associative algebra with the unit $1$. Then the map $S : A^{\otimes 3} \rightarrow A^{\otimes 3}$ defined by $S (a \otimes b \otimes c) = 1 \otimes 1 \otimes abc$ is a pre-$3$-Yang-Baxter operator on the vector space $A$. However, $S$ is not invertible in general.
\end{exam}

\begin{exam}
    Let $S : V^{\otimes n} \rightarrow V^{\otimes n}$ be an $n$-Yang-Baxter operator. Then for any linear automorphism $\varphi \in \mathrm{Aut} (V)$, the map $S_\varphi : V^{\otimes n} \rightarrow V^{\otimes n}$ (called the conjugation of $S$ by the automorphism $\varphi$) defined by $S_\varphi = (\varphi^{-1})^{\otimes n} \circ S \circ \varphi^{\otimes n}$ is also an $n$-Yang-Baxter operator. This shows that the set of all $n$-Yang-Baxter operators on $V$ is closed under conjugations by $\mathrm{Aut}(V)$.
\end{exam}

\begin{exam}
    Let $H$ be a Hopf algebra. In \cite{zappala} Zappala showed that the map $R_H : H \otimes H \rightarrow H \otimes H$ defined by $R_H (x \otimes y) = y^{(1)} \otimes y^{(2)} x S(y^{(3)})$, is a Yang-Baxter operator on $H$. Here, $S$ is the antipode of $H$, and the juxtaposition on the right-hand side indicates the multiplication in $H$. Generalizing this, we define a map $S_H : H^{\otimes n} \rightarrow H^{\otimes n}$ by
    \begin{align*}
        S_H (x_1 \otimes \cdots \otimes x_n) = x_2^{(1)} \otimes \cdots \otimes x_n^{(1)} \otimes x_n^{(2)} \cdots x_2^{(2)} x_1 S(x_2^{(3)}) \cdots S(x_n^{(3)}),
    \end{align*}
    for $x_1, \ldots, x_n \in H$. Then the properties of the antipode $S$ yield that $S_H$ is an $n$-Yang-Baxter operator on $H$.
\end{exam}

More examples of $n$-Yang-Baxter operators coming from (central) $n$-Leibniz algebras and cocommutative linear $n$-racks are given in the following results. The next one generalizes the result of Lebed \cite{lebed2} about the construction of a Yang-Baxter operator from a central Leibniz algebra.

\begin{thm}\label{central-n-leibniz-nyb}
    Let $(\mathcal{L}, [-, \ldots, -], {\bf 1})$ be a central $n$-Leibniz algebra. Define a linear map $S : \mathcal{L}^{\otimes n} \rightarrow \mathcal{L}^{\otimes n}$ by
    \begin{align*}
        S (x_1 \otimes \cdots \otimes x_n) = x_2 \otimes \cdots \otimes x_n \otimes x_1 ~+~ \underbrace{{\bf 1} \otimes \cdots \otimes {\bf 1} }_{n-1 \mathrm{~ times}}\otimes ~ \! [x_1, \ldots, x_n],
    \end{align*}
    for $x_1, \ldots, x_n \in \mathcal{L}$. Then $S$ is an $n$-Yang-Baxter operator on the vector space $\mathcal{L}$.
\end{thm}

\begin{proof}
Let $x_1, \ldots, x_n, y_1, \ldots, y_{n-1} \in \mathcal{L}$. By a direct calculation, we have
    \begin{align*}
        &(S \otimes \mathrm{Id}^{\otimes (n-1)}) \cdots (\mathrm{Id}^{\otimes (n-2)} \otimes S \otimes \mathrm{Id})  (\mathrm{Id}^{\otimes (n-1)} \otimes S) (S \otimes \mathrm{Id}^{\otimes (n-1)}) (x_1 \otimes \cdots \otimes x_n \otimes y_1 \otimes \cdots \otimes y_{n-1}) \\
        &= y_1 \otimes \cdots \otimes y_{n-1} \otimes x_2 \otimes \cdots \otimes x_n \otimes x_1 + \underbrace{{\bf 1} \otimes \cdots \otimes {\bf 1}}_{n-1 \text{ times}}  \otimes ~ { x_2} \otimes \cdots \otimes x_n \otimes [x_1, y_1, \ldots, y_{n-1}] \\
        & \quad + \sum_{i=2}^n ~\underbrace{{\bf 1} \otimes \cdots \otimes {\bf 1}}_{n-1 \text{ times}} \otimes ~ x_2 \otimes \cdots \otimes [x_i, y_1, \ldots, y_{n-1}] \otimes \cdots \otimes x_n \otimes x_1 \\
        & \quad + y_1 \otimes \cdots \otimes y_{n-1} \otimes \underbrace{{\bf 1} \otimes \cdots \otimes {\bf 1}}_{n-1 \text{ times}} \otimes ~ [x_1, \ldots, x_n ]  ~+~ \underbrace{{\bf 1} \otimes \cdots \otimes {\bf 1}}_{2n-2 \text{ times}} \otimes ~[[x_1, \ldots, x_n], y_1, \ldots, y_{n-1}].
    \end{align*}
    On the other hand, 
    \begin{align*}
       &(\mathrm{Id}^{\otimes (n-1)} \otimes S)  (S \otimes \mathrm{Id}^{\otimes (n-1)}) \cdots (\mathrm{Id}^{\otimes (n-2)} \otimes S \otimes \mathrm{Id})  (\mathrm{Id}^{\otimes (n-1)} \otimes S)  (x_1 \otimes \cdots \otimes x_n \otimes y_1 \otimes \cdots \otimes y_{n-1}) \\
       &= y_1 \otimes \cdots \otimes y_{n-1} \otimes x_2 \otimes \cdots \otimes x_n \otimes x_1   + \sum_{i=2}^n ~\underbrace{{\bf 1} \otimes \cdots \otimes {\bf 1}}_{n-1 \text{ times}} \otimes ~ x_2 \otimes \cdots \otimes [x_i, y_1, \ldots, y_{n-1}] \otimes \cdots \otimes x_n \otimes x_1 \\
      & \quad + \underbrace{{\bf 1} \otimes \cdots \otimes {\bf 1}}_{n-1 \text{ times}}  \otimes ~ { x_2} \otimes \cdots \otimes x_n \otimes [x_1, y_1, \ldots, y_{n-1}]  ~+~ y_1 \otimes \cdots \otimes y_{n-1} \otimes \underbrace{{\bf 1} \otimes \cdots \otimes {\bf 1}}_{n-1 \text{ times}} \otimes ~ [x_1, \ldots, x_n ]  \\
      & \quad + \sum_{i=1}^n ~ \underbrace{{\bf 1} \otimes \cdots \otimes {\bf 1}}_{2n-2 \text{ times}} \otimes ~[x_1, \ldots, x_{i-1}, [x_i, y_1, \ldots, y_{n-1}], \ldots, x_n].
    \end{align*}
    Hence, the equation (\ref{nybe}) follows from the fundamental identity of the bracket $[-, \ldots, - ]$.

    Note that the map $S$ is surjective as for any $x_1 , \ldots, x_n \in \mathcal{L}$, 
    \begin{align*}
        S \big( x_n \otimes x_1 \otimes \cdots \otimes x_{n-1} - [x_n, x_1, \ldots, x_{n-1}] \otimes {\bf 1} \otimes \cdots \otimes {\bf 1} \big)
        = x_1 \otimes \cdots \otimes x_n.
        %+ {\bf 1} \otimes \cdots \otimes {\bf 1} \otimes [x_n, x_1, \ldots, x_{n-1}] -  {\bf 1} \otimes \cdots \otimes {\bf 1} \otimes [x_n, x_1, \ldots, x_{n-1}]
    \end{align*}
    Next, for $\sum_{i} x_{1i} \otimes \cdots \otimes x_{ni} \in \mathrm{ker}(S)$, we have
    \begin{align}\label{kernel-eq}
        \sum_i \big(  x_{2i} \otimes \cdots \otimes x_{ni} \otimes x_{1i} ~+~ {\bf 1} \otimes \cdots \otimes {\bf 1} \otimes [x_{1i}, \ldots,  x_{ni}]  \big) = 0.
    \end{align}
    By applying the bijection map $\tau : V^{\otimes n} \rightarrow V^{\otimes n}, ~\tau (x_1 \otimes \cdots \otimes x_n) = x_n \otimes x_1 \otimes \cdots \otimes x_{n-1}$ to the both sides of (\ref{kernel-eq}), and then using the $n$-Leibniz bracket, we obtain
    \begin{align*}
         \sum_i [x_{1i} , \ldots, x_{ni}] = 0.
    \end{align*}
    This in turn implies that (follows from (\ref{kernel-eq}) and by applying $\tau$) $\sum_i x_{1i} \otimes \cdots \otimes x_{ni} = 0$. Hence, $S$ is also injective and hence bijective. This proves the result.
\end{proof}

\begin{thm}
    Let $\mathcal{L}$ be a vector space equipped with an $n$-linear operation $[-, \ldots, - ] : \mathcal{L} \times \cdots \times \mathcal{L} \rightarrow \mathcal{L}$. We define a linear map $S : \overline{\mathcal{L}}^{ \otimes n} \rightarrow \overline{\mathcal{L}}^{ \otimes n}$ by 
    \begin{align*}
        &S \big( (\lambda_1, x_1) \otimes (\lambda_2, x_2) \otimes \cdots \otimes (\lambda_n, x_n) \big) \\
        & \qquad = (\lambda_2, x_2) \otimes \cdots \otimes (\lambda_n, x_n) \otimes (\lambda_1, x_1) ~+~ \underbrace{(1, 0) \otimes \cdots \otimes (1, 0)}_{n-1 \mathrm{~times}} \otimes ~ \! (0, [x_1, \ldots, x_n]),
    \end{align*}
    for $(\lambda_1, x_1), \ldots , (\lambda_n, x_n) \in \overline{\mathcal{L}}$.
    Then $(\mathcal{L}, [-, \ldots, -])$ is an $n$-Leibniz algebra if and only if $S$ is an $n$-Yang-Baxter operator (on the vector space $\overline{\mathcal{L}} = {\bf k} \oplus \mathcal{L}$).
\end{thm}

\begin{proof}
    We consider the $n$-linear operation $[ \! \! [ -, \ldots, - ] \! \! ] : \overline{\mathcal{L}} \times \cdots \times \overline{\mathcal{L}} \rightarrow \overline{\mathcal{L}}$ given by 
    \begin{align*}
      [ \! \! [ (\lambda_1, x_1), \ldots, (\lambda_n, x_n) ] \! \! ] = (0, [x_1, \ldots, x_n]),  
    \end{align*}
    for $(\lambda_1, x_1), \ldots, (\lambda_n, x_n) \in \overline{\mathcal{L}}$. Then $(\mathcal{L}, [-, \ldots, - ])$ is an $n$-Leibniz algebra if and only if the triple $(\overline{\mathcal{L}}, [ \! \! [ -, \ldots, - ] \! \! ], {\bf 1} = (1, 0))$ is a central $n$-Leibniz algebra. Hence, the result follows as a consequence (and calculations) of Theorem \ref{central-n-leibniz-nyb}.
\end{proof}

\begin{thm}
    Let $(C, \langle -, \ldots, - \rangle)$ be a linear $n$-rack where $C$ is cocommutative. Then the map 
    \begin{align*}
    S (u_1 \otimes \cdots \otimes u_n) = u_2^{(1)} \otimes \cdots \otimes u_n^{(1)} \otimes \langle u_1, u_2^{(2)}, \ldots, u_n^{(2)} \rangle, \text{ for } u_1 \otimes \cdots \otimes u_n \in C^{\otimes n}
    \end{align*}
    is an $n$-Yang-Baxter operator (on the vector space $C$) with the inverse
    \begin{align}\label{inv-nybe}
       S^{-1} (u_1 \otimes \cdots \otimes u_n) = \ll u_n, u_{n-1}^{(2)}, \ldots, u_1^{(2)}  \gg \otimes ~ \! u_1^{(1)} \otimes \cdots \otimes u_{n-1}^{(1)}.
    \end{align}
\end{thm}

\begin{proof}
    For any $u_1, \ldots, u_n, v_1, \ldots, v_{n-1} \in C$, we observe that
    \begin{align}\label{ncomm1}
       & (S \otimes \mathrm{Id}^{\otimes (n-1)}) (\mathrm{Id} \otimes S \otimes \mathrm{Id}^{\otimes (n-2)}) \cdots (\mathrm{Id}^{\otimes (n-1)} \otimes S ) (S \otimes \mathrm{Id}^{\otimes (n-1)}) (u_1 \otimes \cdots \otimes u_n \otimes v_1 \otimes \cdots \otimes v_{n-1}) \nonumber \\
       & = (S \otimes \mathrm{Id}^{\otimes (n-1)}) (\mathrm{Id} \otimes S \otimes \mathrm{Id}^{\otimes (n-2)}) \cdots (\mathrm{Id}^{\otimes (n-1)} \otimes S ) \nonumber \\
       & \qquad \qquad \big(   u_2^{(1)} \otimes \cdots \otimes u_n^{(1)} \otimes \langle  u_1, u_2^{(2)}, \ldots, u_n^{(2)}  \rangle \otimes v_1 \otimes \cdots \otimes v_{n-1}  \big) \nonumber \\
        &= (S \otimes \mathrm{Id}^{\otimes (n-1)}) (\mathrm{Id} \otimes S \otimes \mathrm{Id}^{\otimes (n-2)}) \cdots (\mathrm{Id}^{\otimes (n-2)} \otimes S \otimes \mathrm{Id}) \nonumber \\
        & \qquad \qquad \qquad \big(  u_2^{(1)} \otimes \cdots \otimes u_n^{(1)} \otimes v_1^{(1)} \otimes \cdots \otimes v_{n-1}^{(1)} \otimes \langle \langle u_1, u_2^{(2)}, \ldots, u_n^{(2)}    \rangle, v_1^{(2)}, \ldots, v_{n-1}^{(2)}   \big) \nonumber \\
       & = v_1^{(1)(1) \cdots (1)} \otimes \cdots v_{n-1}^{(1)(1) \cdots (1)} \otimes \langle u_2^{(1)},v_1^{(1) \cdots (1) (2)}, \ldots, v_{n-1}^{(1) \cdots (1) (2)}  \rangle \otimes \cdots \otimes \langle u_{n-1}^{(1)}, v_1^{(1) (1) (2)}, \ldots,  v_{n-1}^{(1) (1) (2)} \rangle \nonumber \\
       & \qquad \qquad \qquad \otimes \langle u_n^{(1)}, v_1^{(1) (2)}, \ldots, v_{n-1}^{(1) (2)} \rangle 
        \otimes \big\langle \langle u_1, u_2^{(2)}, \ldots, u_n^{(2)} \rangle , v_1^{(2)} , \ldots, v_{n-1}^{(2)} \big\rangle.
    \end{align}
    On the other hand, 
    \begin{align}\label{ncomm2}
        &(\mathrm{Id}^{\otimes (n-1)} \otimes S) (S \otimes \mathrm{Id}^{\otimes (n-1)}) \cdots (\mathrm{Id}^{\otimes (n-2)} \otimes S \otimes \mathrm{Id}) (\mathrm{Id}^{\otimes (n-1)} \otimes S) (u_1 \otimes \cdots \otimes u_n \otimes v_1 \otimes \cdots \otimes v_{n-1})  \nonumber \\
        &= (\mathrm{Id}^{\otimes (n-1)} \otimes S) (S \otimes \mathrm{Id}^{\otimes (n-1)}) \cdots (\mathrm{Id}^{\otimes (n-2)} \otimes S \otimes \mathrm{Id})  \nonumber \\
        & \qquad \qquad \qquad \big(    u_1 \otimes \cdots \otimes u_{n-1} \otimes v_1^{(1)} \otimes \cdots \otimes v_{n-1}^{(1)} \otimes \langle u_n, v_1^{(2)}, \ldots, v_{n-1}^{(2)} \rangle \big)  \nonumber  \\
        &= (\mathrm{Id}^{\otimes (n-1)} \otimes S) \big(   v_1^{(1) (1) \cdots (1)} \otimes \cdots \otimes v_{n-1}^{(1) (1) \cdots (1)} \otimes \langle u_1, v_1^{(1) \cdots (1) (2)}, \ldots, v_{n-1}^{(1) \cdots (1) (2)}   \rangle \otimes \cdots  \nonumber \\
        & \qquad \qquad \qquad  \cdots \otimes \langle u_{n-1}, v_1^{(1) (2)}, \ldots, v_{n-1}^{(1) (2)} \rangle \otimes \langle u_n , v_1^{(2)}, \ldots, v_{n-1}^{(2)} \rangle   \big)  \nonumber \\
        &= v_1^{(1) (1) \cdots (1)} \otimes \cdots \otimes v_{n-1}^{(1) (1) \cdots (1)} \otimes \langle u_2, v_1^{(1) \cdots (1) (2)} , \ldots, v_{n-1}^{(1) \cdots (1) (2)}   \rangle^{(1)} \otimes \cdots \otimes \langle  u_n , v_1^{(2)}, \ldots, v_{n-1}^{(2)} \rangle^{(1)} \otimes \nonumber \\
        & \qquad \big\langle  \langle u_1, v_1^{(1) \cdots (1) (2)}, \ldots , v_{n-1}^{(1) \cdots (1) (2)} \rangle, \langle u_2, v_1^{(1) \cdots (1) (2)}, \ldots , v_{n-1}^{(1) \cdots (1) (2)} \rangle^{(2)}, \ldots,  \langle u_n, v_1^{(2)}, \ldots, v_{n-1}^{(2)} \rangle^{(2)} \big\rangle.
    \end{align}
    Since $C$ is cocommutative and $\langle -, \ldots, - \rangle$ satisfies the self-distributivity (\ref{self-dis}), the expressions in (\ref{ncomm1}) and (\ref{ncomm2}) are the same. This shows that $S$ is a pre-$n$-Yang-Baxter operator on the vector space $C$. Finally, $S$ is invertible with the inverse given in (\ref{inv-nybe}). Hence, the result follows.  
\end{proof}

In the next result, we show that a Yang-Baxter operator naturally gives rise to an $n$-Yang-Baxter operator on the same vector space. More precisely, we have the following.

\begin{thm}\label{them-soln-nsoln}
     Let $V$ be a vector space and $R: V \otimes V \rightarrow V \otimes V$ be a Yang-Baxter operator on $V$. For any $n \geq 3$, we define a map $S_n : V^{\otimes n} \rightarrow V^{\otimes n}$ by
     \begin{align*}
         S_n = (\mathrm{Id}^{\otimes (n-2)} \otimes R) \cdots (\mathrm{Id} \otimes R \otimes \mathrm{Id}^{\otimes (n-3)}) (R \otimes \mathrm{Id}^{\otimes (n-2)}).
     \end{align*}
     Then $S_n$ is an $n$-Yang-Baxter operator on $V$.
\end{thm}

\begin{proof}
   We will prove the result by using mathematical induction. For this, we first assume that $S_{n-1}$ (for some $n\ge3$ with the convention that $S_2=R$) is an $(n-1)$-Yang-Baxter operator on $V$. We note that
\begin{align*}
      S_{n}=(\mathrm{Id}\otimes S_{n-1})(R\otimes \mathrm{Id}^{\otimes(n-2)})=(\mathrm{Id}^{\otimes(n-2)}\otimes R)(S_{n-1}\otimes \mathrm{Id}).
\end{align*}Hence 
\begin{align*}
    &(S_{n}\otimes \mathrm{Id}^{\otimes (n-1)})\cdots(\mathrm{Id}^{\otimes(n-1)}\otimes S_n)(S_n\otimes \mathrm{Id}^{\otimes(n-1)})\\
    &=(\mathrm{Id}^{\otimes(n-2)}\otimes R\otimes \mathrm{Id}^{\otimes(n-1)})(S_{n-1}\otimes \mathrm{Id}^{\otimes n})\cdots (\mathrm{Id}^{\otimes(2n-3)}\otimes R)(\mathrm{Id}^{\otimes(n-1)}\otimes S_{n-1}\otimes \mathrm{Id})\\&\qquad\qquad\qquad\qquad\qquad\qquad\qquad\qquad\qquad\qquad\qquad\qquad\qquad\qquad(\mathrm{Id}\otimes S_{n-1}\otimes \mathrm{Id}^{\otimes(n-1)})(R\otimes \mathrm{Id}^{\otimes(2n-3)}).
\end{align*}
This gives a composition of $(2n+2)$ maps, among which $(n+1)$ maps are of the form $(\mathrm{Id}^{\otimes i}\otimes S_{n-1}\otimes \mathrm{Id}^{\otimes j})$ with $i+j=n$ and remaining $(n+1)$ maps are of the form $(\mathrm{Id}^{\otimes p}\otimes R\otimes \mathrm{Id}^{\otimes q})$ with $p+q=2n-3$. In the above expression, we move all the maps of the form $(\mathrm{Id}^{\otimes p}\otimes R\otimes \mathrm{Id}^{\otimes q})$ to the left-hand side, preserving their original order, except for the right-most term $(R\otimes \mathrm{Id}^{\otimes (2n-3)})$ which we leave in its place. After this rearrangement, the expression becomes
\begin{align*}
    &  {\scriptstyle (\mathrm{Id}^{\otimes(n-2)}\otimes R\otimes \mathrm{Id}^{\otimes(n-1)})(\mathrm{Id}^{\otimes(n-1)}\otimes R\otimes \mathrm{Id}^{\otimes(n-2)})\cdots(\mathrm{Id}^{\otimes (2n-3)}\otimes R)(S_{n-1}\otimes \mathrm{Id}^{\otimes n}) } \\
    &\qquad {\scriptstyle  \underbrace{ {\scriptstyle (\mathrm{Id}\otimes S_{n-1}\otimes \mathrm{Id}^{\otimes(n-1)}) } }_{\mbox{\tiny (n+2)\text{\tiny-th position}}}(\mathrm{Id}^{\otimes 2}\otimes S_{n-1}\otimes \mathrm{Id}^{\otimes(n-2)} ) \cdots(\mathrm{Id}^{\otimes(n-1)}\otimes S_{n-1}\otimes \mathrm{Id})\underbrace{   {\scriptstyle (\mathrm{Id}\otimes S_{n-1}\otimes \mathrm{Id}^{\otimes(n-1)}) }  }_{\mbox{\tiny (2n+1)\text{\tiny-th position}}} (R\otimes \mathrm{Id}^{\otimes(2n-3)}) }.
\end{align*}
At this point, we may apply the $(n-1)$-Yang-Baxter equation to the $n$ maps located from $(n+2)$-th to $(2n+1)$-th positions. Subsequently, we get the expression 
\begin{align*}
    & {\scriptstyle (\mathrm{Id}^{\otimes(n-2)}\otimes R\otimes \mathrm{Id}^{\otimes(n-1)})(\mathrm{Id}^{\otimes(n-1)}\otimes R\otimes \mathrm{Id}^{\otimes(n-2)})\cdots(\mathrm{Id}^{\otimes (2n-3)}\otimes R)\underbrace{ {\scriptstyle (S_{n-1}\otimes \mathrm{Id}^{\otimes n}) } }_{\mbox{\tiny(n+1)-\text{\tiny th position}}} (\mathrm{Id}^{\otimes(n-1)}\otimes S_{n-1}\otimes \mathrm{Id}) } \\
    & \qquad {\scriptstyle (\mathrm{Id}\otimes S_{n-1}\otimes \mathrm{Id}^{\otimes(n-1)})\underbrace{   {\scriptstyle (\mathrm{Id}^{\otimes 2}\otimes S_{n-1}\otimes \mathrm{Id}^{\otimes(n-2)}) } }_{\mbox{\tiny(n+4)-\text{\tiny th position}}}\cdots(\mathrm{Id}^{\otimes(n-1)}\otimes S_{n-1}\otimes\mathrm{Id}) (R\otimes \mathrm{Id}^{\otimes(2n-3)}).}
\end{align*}
Note that the maps present in the $(n+1)$-th position and in the $(n+2)$-th position are interchangeable. Furthermore, we observe that the map $(R\otimes \mathrm{Id}^{\otimes(2n-3)})$, which appears in the last position, can be moved to the $(n+4)$-th position. This yields the expression 
\begin{align*}
    & {\scriptstyle  (\mathrm{Id}^{\otimes(n-2)}\otimes R\otimes \mathrm{Id}^{\otimes(n-1)})(\mathrm{Id}^{\otimes(n-1)}\otimes R\otimes \mathrm{Id}^{\otimes(n-2)})\cdots(\mathrm{Id}^{\otimes (2n-3)}\otimes R)\underbrace{ {\scriptstyle  (\mathrm{Id}^{\otimes(n-1)}\otimes S_{n-1}\otimes \mathrm{Id}) }  }_{\mbox{\tiny(n+1)-\text{\tiny th position}}}  }   \\
    & \qquad {\scriptstyle (S_{n-1}\otimes \mathrm{Id}^{\otimes n})(\mathrm{Id}\otimes S_{n-1}\otimes \mathrm{Id}^{\otimes(n-1)})(R\otimes \mathrm{Id}^{\otimes(2n-3)})(\mathrm{Id}^{\otimes 2}\otimes S_{n-1}\otimes \mathrm{Id}^{\otimes(n-2)})\cdots(\mathrm{Id}^{\otimes(n-1)}\otimes S_{n-1}\otimes \mathrm{Id})}.
\end{align*}
At this stage, we may successively apply the Yang-Baxter equation $(n-2)$ times to the maps present from $2$nd position to the $(n+1)$-th position, yielding the expression
\begin{align*}
    & {\scriptstyle (\mathrm{Id}^{\otimes(n-2)}\otimes R\otimes \mathrm{Id}^{\otimes(n-1)})(\mathrm{Id}^{\otimes n}\otimes S_{n-1})(\mathrm{Id}^{\otimes(n-1)}\otimes R\otimes \mathrm{Id}^{\otimes(n-2)})\cdots(\mathrm{Id}^{\otimes (2n-3)}\otimes R)  } \\
    &\qquad\qquad  \underbrace{ {\scriptstyle  (S_{n-1}\otimes \mathrm{Id}^{\otimes n}) } }_{\mbox{\tiny(n+2)\text{\tiny-th position}}}  \underbrace{ {\scriptstyle (\mathrm{Id}\otimes S_{n-1}\otimes \mathrm{Id}^{\otimes(n-1)})   } }_{\mbox{\tiny(n+3)\text{\tiny-th position}}} \underbrace{ {\scriptstyle (R\otimes \mathrm{Id}^{\otimes(2n-3)})  } }_{\mbox{\tiny(n+4)\text{\tiny -th position}}}  {\scriptstyle (\mathrm{Id}^{\otimes 2}\otimes S_{n-1}\otimes \mathrm{Id}^{\otimes(n-2)})\cdots(\mathrm{Id}^{\otimes(n-1)}\otimes S_{n-1}\otimes \mathrm{Id}) }.
\end{align*}
We can now interchange the maps in the 1st and 2nd positions. At the same time, we may successively apply the Yang-Baxter equation $(n-2)$ times to the three maps present in $(n+2)$-th, $(n+3)$-th and $(n+4)$-th positions. This gives the expression as 
\begin{align*}
   & {\scriptstyle (\mathrm{Id}^{\otimes n}\otimes S_{n-1}) (\mathrm{Id}^{\otimes(n-2)}\otimes R\otimes \mathrm{Id}^{\otimes(n-1)})(\mathrm{Id}^{\otimes(n-1)}\otimes R\otimes \mathrm{Id}^{\otimes(n-2)})(\mathrm{Id}^{\otimes n}\otimes R\otimes \mathrm{Id}^{\otimes(n-3)})\cdots  (\mathrm{Id}^{\otimes (2n-3)}\otimes R) } \\
   & \qquad {\scriptstyle (\mathrm{Id}^{\otimes(n-2)}\otimes R\otimes \mathrm{Id}^{\otimes(n-1)})(S_{n-1}\otimes \mathrm{Id}^{\otimes n})(\mathrm{Id}\otimes S_{n-1}\otimes \mathrm{Id}^{\otimes(n-1)})(\mathrm{Id}^{\otimes 2}\otimes S_{n-1}\otimes \mathrm{Id}^{\otimes(n-2)})  \cdots(\mathrm{Id}^{\otimes(n-1)}\otimes S_{n-1}\otimes \mathrm{Id})  }.
\end{align*}
At this place, for any $0 \leq p \leq n-3$, the map $(\mathrm{Id}^{\otimes(n+p)}\otimes R\otimes\mathrm{Id}^{\otimes (n-p-3)})$ can be placed from its original $(p+4)$-th position to the $(n+p+5)$-th position. Such a map $(\mathrm{Id}^{\otimes(n+p)}\otimes R\otimes\mathrm{Id}^{\otimes (n-p-3)})$ is placed between 
$(\mathrm{Id}^{\otimes(p+1)}\otimes S_{n-1}\otimes \mathrm{Id}^{\otimes(n-p-1)})$ and $(\mathrm{Id}^{\otimes(p+2)}\otimes S_{n-1}\otimes \mathrm{Id}^{\otimes(n-p-2)})$. After this transformation, the expression gives rise to  
\begin{align*}
    &(\mathrm{Id}^{\otimes n}\otimes S_{n-1})(\mathrm{Id}^{\otimes (n-2)}\otimes R\otimes \mathrm{Id}^{\otimes(n-1)})(\mathrm{Id}^{\otimes (n-1)}\otimes R\otimes \mathrm{Id}^{\otimes(n-2)})(\mathrm{Id}^{\otimes (n-2)}\otimes R\otimes \mathrm{Id}^{\otimes(n-1)})\\
    &\qquad\qquad\qquad(S_{n-1}\otimes \mathrm{Id}^{\otimes n})(\mathrm{Id}\otimes S_{n-1}\otimes \mathrm{Id}^{\otimes(n-1)})(\mathrm{Id}^{\otimes n}\otimes R\otimes \mathrm{Id}^{\otimes(n-3)})(\mathrm{Id}^{\otimes 2}\otimes S_{n-1}\otimes \mathrm{Id}^{\otimes(n-2)})\cdots\\
    &\qquad(\mathrm{Id}^{\otimes(p+1)}\otimes S_{n-1}\otimes \mathrm{Id}^{\otimes(n-p-1)})(\mathrm{Id}^{\otimes(n+p)}\otimes R\otimes\mathrm{Id}^{\otimes (n-p-3)})(\mathrm{Id}^{\otimes(p+2)}\otimes S_{n-1}\otimes \mathrm{Id}^{\otimes(n-p-2)})\cdots\\
&\qquad\qquad\qquad\qquad\qquad\qquad\qquad\qquad(\mathrm{Id}^{\otimes(n-2)}\otimes S_{n-1}\otimes \mathrm{Id}^{\otimes 2})(\mathrm{Id}^{\otimes (2n-3)}\otimes R)(\mathrm{Id}^{\otimes(n-1)}\otimes S_{n-1}\otimes \mathrm{Id})\\
    &=(\mathrm{Id}^{\otimes n}\otimes S_{n-1})\underbrace{(\mathrm{Id}^{\otimes (n-2)}\otimes R\otimes \mathrm{Id}^{\otimes(n-1)})(\mathrm{Id}^{\otimes (n-1)}\otimes R\otimes \mathrm{Id}^{\otimes(n-2)})(\mathrm{Id}^{\otimes (n-2)}\otimes R\otimes \mathrm{Id}^{\otimes(n-1)})}\\
    &\qquad\qquad\qquad\qquad\qquad(S_{n-1}\otimes \mathrm{Id}^{\otimes n})(\mathrm{Id}\otimes S_{n-1}\otimes \mathrm{Id}^{\otimes(n-1)})(\mathrm{Id}^{\otimes 2}\otimes S_{n}\otimes \mathrm{Id}^{\otimes(n-3)})\cdots (\mathrm{Id}^{\otimes (n-1)}\otimes S_{n})
    \end{align*}
    \begin{align*}
    &=(\mathrm{Id}^{\otimes n}\otimes S_{n-1})(\mathrm{Id}^{\otimes (n-1)}\otimes R\otimes \mathrm{Id}^{\otimes(n-2)})(\mathrm{Id}^{\otimes (n-2)}\otimes R\otimes \mathrm{Id}^{\otimes(n-1)})(\mathrm{Id}^{\otimes (n-1)}\otimes R\otimes \mathrm{Id}^{\otimes(n-2)})\\
    &\qquad\qquad\qquad\qquad\qquad(S_{n-1}\otimes \mathrm{Id}^{\otimes n})(\mathrm{Id}\otimes S_{n-1}\otimes \mathrm{Id}^{\otimes(n-1)})(\mathrm{Id}^{\otimes 2}\otimes S_{n}\otimes \mathrm{Id}^{\otimes(n-3)})\cdots (\mathrm{Id}^{\otimes (n-1)}\otimes S_{n})\\
    &=(\mathrm{Id}^{\otimes (n-1)}\otimes S_{n})(S_{n}\otimes \mathrm{Id}^{\otimes (n-1)})(\mathrm{Id}\otimes S_{n}\otimes \mathrm{Id}^{\otimes(n-2)}) \cdots (\mathrm{Id}^{\otimes (n-2)} \otimes S_n \otimes \mathrm{Id}) (\mathrm{Id}^{\otimes (n-1)}\otimes S_{n}).
\end{align*}In the underlined terms, we have used that $R$ is a Yang-Baxter operator on $V$. This shows that $S_n$ is a pre-$n$-Yang-Baxter operator on $V$. Finally, $S_n$ is invertible as it is a composition of invertible maps.
\end{proof}

\begin{exam}
    Let $(\mathcal{L} , \{ -, - \} , {\bf 1})$ be a central Leibniz algebra. Then the map $R : \mathcal{L} \otimes \mathcal{L} \rightarrow \mathcal{L} \otimes \mathcal{L}$ defined in (\ref{lebed-exp}) is a Yang-Baxter operator on $\mathcal{L}$. Hence by Theorem \ref{them-soln-nsoln}, one can formulate an $n$-Yang-Baxter operator $S_n : \mathcal{L}^{\otimes n} \rightarrow \mathcal{L}^{\otimes n}$  on the vector space $\mathcal{L}$. In particular, the map $S_3 : \mathcal{L}^{\otimes 3} \rightarrow \mathcal{L}^{\otimes 3}$ given by
\begin{align*}
    S_3 (x \otimes y \otimes z) = y \otimes z \otimes x ~ \! +~ \! y \otimes {\bf 1} \otimes \{ x, z \}  ~ \! +~ \!   {\bf 1} \otimes z \otimes \{ x, y \}  ~ \! +~ \!  {\bf 1} \otimes {\bf 1} \otimes \{ \{ x, y \}, z\}, \text{ for } x, y, z \in \mathcal{L},
\end{align*}
is a $3$-Yang-Baxter operator on $\mathcal{L}.$ Note that this $3$-Yang-Baxter operator on $\mathcal{L}$ is different than the one obtained from the central $3$-Leibniz algebra $(\mathcal{L}, [-, -, - ], {\bf 1})$, where $[x, y, z ] = \{ x, \{ y, z \} \}$, for $x, y, z \in \mathcal{L}$.
\end{exam}

In the following, we show that an $n$-Yang-Baxter operator naturally induces a Yang-Baxter operator. To understand this result for general values of $n$, we shall first briefly discuss the cases for $n=3, 4$ in the next two results. Adapting the ideas from such small values of $n$, we shall intuitively describe the general case.

\begin{prop}
    Let $V$ be a vector space and $S : V^{\otimes 3} \rightarrow V^{\otimes 3}$ be a $3$-Yang-Baxter operator on $V$. Then the map $\widetilde{S} : V^{\otimes 2} \otimes V^{\otimes 2} \rightarrow V^{\otimes 2} \otimes V^{\otimes 2}$ defined by $\widetilde{ S} = (S \otimes \mathrm{Id}) (\mathrm{Id} \otimes S)$ is a Yang-Baxter operator on the vector space $V^{\otimes 2}$.
\end{prop}

\begin{proof}
    We observe that
    \begin{align*}
        &(\widetilde{S} \otimes \mathrm{Id}^{\otimes 2}) (\mathrm{Id}^{\otimes 2} \otimes \widetilde{S}) (\widetilde{S} \otimes \mathrm{Id}^{\otimes 2}) \\
        &= (S \otimes \mathrm{Id}^{\otimes 3}) (\mathrm{Id} \otimes S \otimes \mathrm{Id}^{\otimes 2}) (\mathrm{Id}^{\otimes 2} \otimes S \otimes \mathrm{Id}) (\mathrm{Id}^{\otimes 3} \otimes S)  (S \otimes \mathrm{Id}^{\otimes 3}) (\mathrm{Id} \otimes S \otimes \mathrm{Id}^{\otimes 2}) \\
        &=  \underbrace{(S \otimes \mathrm{Id}^{\otimes 3})  (\mathrm{Id} \otimes S \otimes \mathrm{Id}^{\otimes 2}) (\mathrm{Id}^{\otimes 2} \otimes S \otimes \mathrm{Id}) (S \otimes \mathrm{Id}^{\otimes 3})}  (\mathrm{Id}^{\otimes 3} \otimes S) (\mathrm{Id} \otimes S \otimes \mathrm{Id}^{\otimes 2}) \\ 
        &= (\mathrm{Id}^{\otimes 2} \otimes S \otimes \mathrm{Id}) (S \otimes \mathrm{Id}^{\otimes 3}) \underbrace{ (\mathrm{Id} \otimes S \otimes \mathrm{Id}^{\otimes 2}) (\mathrm{Id}^{\otimes 2} \otimes S \otimes \mathrm{Id}) (\mathrm{Id}^{\otimes 3} \otimes S) (\mathrm{Id} \otimes S \otimes \mathrm{Id}^{\otimes 2})} \\ 
        &= (\mathrm{Id}^{\otimes 2} \otimes S \otimes \mathrm{Id}) (S \otimes \mathrm{Id}^{\otimes 3}) (\mathrm{Id}^{\otimes 3} \otimes S)  (\mathrm{Id} \otimes S \otimes \mathrm{Id}^{\otimes 2}) (\mathrm{Id}^{\otimes 2} \otimes S \otimes \mathrm{Id}) (\mathrm{Id}^{\otimes 3} \otimes S) \\
        &= (\mathrm{Id}^{\otimes 2} \otimes S \otimes \mathrm{Id}) (\mathrm{Id}^{\otimes 3} \otimes S)  (S \otimes \mathrm{Id}^{\otimes 3})  (\mathrm{Id} \otimes S \otimes \mathrm{Id}^{\otimes 2})  (\mathrm{Id}^{\otimes 2} \otimes S \otimes \mathrm{Id}) (\mathrm{Id}^{\otimes 3} \otimes S) \\
        &= (\mathrm{Id}^{\otimes 2} \otimes \widetilde{S})   (\widetilde{S} \otimes \mathrm{Id}^{\otimes 2}) (\mathrm{Id}^{\otimes 2} \otimes \widetilde{S}).
    \end{align*}
    Regarding the underlined terms in the above calculation, we used the fact that $S$ is a $3$-Yang-Baxter operator on $V$. Finally, the map $\widetilde{S}$ is invertible as it is a composition of invertible maps.
\end{proof}

\begin{prop}
    Let $S : V^{\otimes 4} \rightarrow V^{\otimes 4}$ be a $4$-Yang-Baxter operator on a vector space $V$. Then the map $\widetilde{S} : V^{\otimes 3} \otimes  V^{\otimes 3}  \rightarrow V^{\otimes 3} \otimes  V^{\otimes 3}$ defined by $\widetilde{S} = (S \otimes \mathrm{Id}^{\otimes 2}) (\mathrm{Id} \otimes S \otimes \mathrm{Id}) (\mathrm{Id}^{\otimes 2} \otimes S)$ is a Yang-Baxter operator on the vector space $V^{\otimes 3}$.
\end{prop}

\begin{proof}
By a direct calculation, we see that
    \begin{align*}
        &(\widetilde{S} \otimes \mathrm{Id}^{\otimes 3})  (\mathrm{Id}^{\otimes 3} \otimes \widetilde{S}) (\widetilde{S} \otimes \mathrm{Id}^{\otimes 3}) \\
       & = ( S \otimes \mathrm{Id}^{\otimes 5}) (\mathrm{Id} \otimes S \otimes \mathrm{Id}^{\otimes 4}) (\mathrm{Id}^{\otimes 2} \otimes S \otimes \mathrm{Id}^{\otimes 3})     
(\mathrm{Id}^{\otimes 3} \otimes S \otimes \mathrm{Id}^{\otimes 2}) (\mathrm{Id}^{\otimes 4} \otimes S \otimes \mathrm{Id}) (\mathrm{Id}^{\otimes 5} \otimes S) \\ 
& \qquad \qquad \qquad \qquad \qquad \qquad \qquad \qquad \qquad ( S \otimes \mathrm{Id}^{\otimes 5}) (\mathrm{Id} \otimes S \otimes \mathrm{Id}^{\otimes 4}) (\mathrm{Id}^{\otimes 2} \otimes S \otimes \mathrm{Id}^{\otimes 3})  \\
&= \underbrace{( S \otimes \mathrm{Id}^{\otimes 5}) (\mathrm{Id} \otimes S \otimes \mathrm{Id}^{\otimes 4}) (\mathrm{Id}^{\otimes 2} \otimes S \otimes \mathrm{Id}^{\otimes 3})     
(\mathrm{Id}^{\otimes 3} \otimes S \otimes \mathrm{Id}^{\otimes 2})  ( S \otimes \mathrm{Id}^{\otimes 5}) } (\mathrm{Id}^{\otimes 4} \otimes S \otimes \mathrm{Id}) \\
& \qquad \qquad \qquad \qquad \qquad \qquad \qquad \qquad \qquad  (\mathrm{Id} \otimes S \otimes \mathrm{Id}^{\otimes 4}) (\mathrm{Id}^{\otimes 5} \otimes S) (\mathrm{Id}^{\otimes 2} \otimes S \otimes \mathrm{Id}^{\otimes 3}) 
\end{align*}
\begin{align*}
&= (\mathrm{Id}^{\otimes 3} \otimes S \otimes \mathrm{Id}^{\otimes 2}) ( S \otimes \mathrm{Id}^{\otimes 5}) \underbrace{(\mathrm{Id} \otimes S \otimes \mathrm{Id}^{\otimes 4}) (\mathrm{Id}^{\otimes 2} \otimes S \otimes \mathrm{Id}^{\otimes 3}) (\mathrm{Id}^{\otimes 3} \otimes S \otimes \mathrm{Id}^{\otimes 2}) (\mathrm{Id}^{\otimes 4} \otimes S \otimes \mathrm{Id}) }_{\text{underline continued}} \\
& \qquad \qquad \qquad \qquad \qquad \qquad \qquad \qquad \underbrace{(\mathrm{Id} \otimes S \otimes \mathrm{Id}^{\otimes 4}) } (\mathrm{Id}^{\otimes 5} \otimes S) (\mathrm{Id}^{\otimes 2} \otimes S \otimes \mathrm{Id}^{\otimes 3}) \\
&= (\mathrm{Id}^{\otimes 3} \otimes S \otimes \mathrm{Id}^{\otimes 2}) ( S \otimes \mathrm{Id}^{\otimes 5}) (\mathrm{Id}^{\otimes 4} \otimes S \otimes \mathrm{Id}) (\mathrm{Id} \otimes S \otimes \mathrm{Id}^{\otimes 4}) \\
& \qquad \qquad \qquad \qquad \underbrace{ (\mathrm{Id}^{\otimes 2} \otimes S \otimes \mathrm{Id}^{\otimes 3}) (\mathrm{Id}^{\otimes 3} \otimes S \otimes \mathrm{Id}^{\otimes 2})  (\mathrm{Id}^{\otimes 4} \otimes S \otimes \mathrm{Id}) (\mathrm{Id}^{\otimes 5} \otimes S) (\mathrm{Id}^{\otimes 2} \otimes S \otimes \mathrm{Id}^{\otimes 3})} \\
&= (\mathrm{Id}^{\otimes 3} \otimes S \otimes \mathrm{Id}^{\otimes 2}) ( S \otimes \mathrm{Id}^{\otimes 5}) (\mathrm{Id}^{\otimes 4} \otimes S \otimes \mathrm{Id}) (\mathrm{Id} \otimes S \otimes \mathrm{Id}^{\otimes 4}) (\mathrm{Id}^{\otimes 5} \otimes S) (\mathrm{Id}^{\otimes 2} \otimes S \otimes \mathrm{Id}^{\otimes 3}) \\
& \qquad \qquad \qquad \qquad \qquad \qquad \qquad \qquad (\mathrm{Id}^{\otimes 3} \otimes S \otimes \mathrm{Id}^{\otimes 2}) (\mathrm{Id}^{\otimes 4} \otimes S \otimes \mathrm{Id}) (\mathrm{Id}^{\otimes 5} \otimes S) \\
&= (\mathrm{Id}^{\otimes 3} \otimes S \otimes \mathrm{Id}^{\otimes 2})  (\mathrm{Id}^{\otimes 4} \otimes S \otimes \mathrm{Id})  (\mathrm{Id}^{\otimes 5} \otimes S)  ( S \otimes \mathrm{Id}^{\otimes 5})  (\mathrm{Id} \otimes S \otimes \mathrm{Id}^{\otimes 4}) (\mathrm{Id}^{\otimes 2} \otimes S \otimes \mathrm{Id}^{\otimes 3}) \\ & \qquad \qquad \qquad \qquad \qquad \qquad \qquad \qquad 
  (\mathrm{Id}^{\otimes 3} \otimes S \otimes \mathrm{Id}^{\otimes 2}) (\mathrm{Id}^{\otimes 4} \otimes S \otimes \mathrm{Id})  (\mathrm{Id}^{\otimes 5} \otimes S) \\
  &= (\mathrm{Id}^{\otimes 3} \otimes \widetilde{S})   (\widetilde{S} \otimes \mathrm{Id}^{\otimes 3})  (\mathrm{Id}^{\otimes 3} \otimes \widetilde{S}).
    \end{align*}
    In the underlined terms, we have used that $S$ is a $4$-Yang-Baxter operator on $V$. The map $\widetilde{S}$ is invertible as it is a composition of invertible maps.
\end{proof}

As mentioned earlier, we will consider a generalization of the two propositions mentioned above. More precisely, we have the following result.

\begin{thm}\label{theorem-19}
    Let $V$ be a vector space and $S : V^{\otimes n} \rightarrow V^{\otimes n}$ be an $n$-Yang-Baxter operator on $V$. We define a map $\widetilde{S} : V^{\otimes (n-1)} \otimes V^{\otimes (n-1)}  \rightarrow V^{\otimes (n-1)}  \otimes V^{\otimes (n-1)} $ by 
    \begin{align*}
        \widetilde{S} =   (S \otimes \mathrm{Id}^{\otimes (n-2)}) (\mathrm{Id} \otimes S \otimes \mathrm{Id}^{\otimes (n-3)}) \cdots (\mathrm{Id}^{\otimes (n-3)} \otimes S \otimes \mathrm{Id}) (\mathrm{Id}^{\otimes (n-2)} \otimes S).
    \end{align*}
    Then $\widetilde{S} $ is a Yang-Baxter operator on the vector space $V^{\otimes (n-1)}$.
\end{thm}

\begin{proof}
If we write the expression $(\widetilde{S} \otimes \mathrm{Id}^{\otimes (n-1)}) ( \mathrm{Id}^{\otimes (n-1)} \otimes \widetilde{S}) (\widetilde{S} \otimes \mathrm{Id}^{\otimes (n-1)})$ in terms of $S$, we obtain a composition of $3 (n-1)$ maps, each of which is of the form $(\mathrm{Id}^{\otimes i} \otimes S \otimes \mathrm{Id}^{\otimes j})$ with $i+j = 2n-3$. By keeping the first $n-1$ factors (from the left-hand side) unchanged, we apply the following to the last $2(n-1)$ factors.

    First, observe that the $(2n-2)$-th factor $(\mathrm{Id}^{\otimes (2n-3)} \otimes S)$ and the $(2n-1)$-th factor $(S \otimes \mathrm{Id}^{\otimes (2n-3)})$ are interchangable. After this interchange, we move the new $(2n-2)$-th factor $(S \otimes \mathrm{Id}^{\otimes (2n-3)})$ to the left and put it in the $(n+1)$-th position. Similarly, we move the $(2n-1)$-th factor $(\mathrm{Id}^{\otimes (2n-3)} \otimes S)$ to the right and put it in the $(3n-4)$-th position. Pictorially, this process can be understood by the diagram:
    
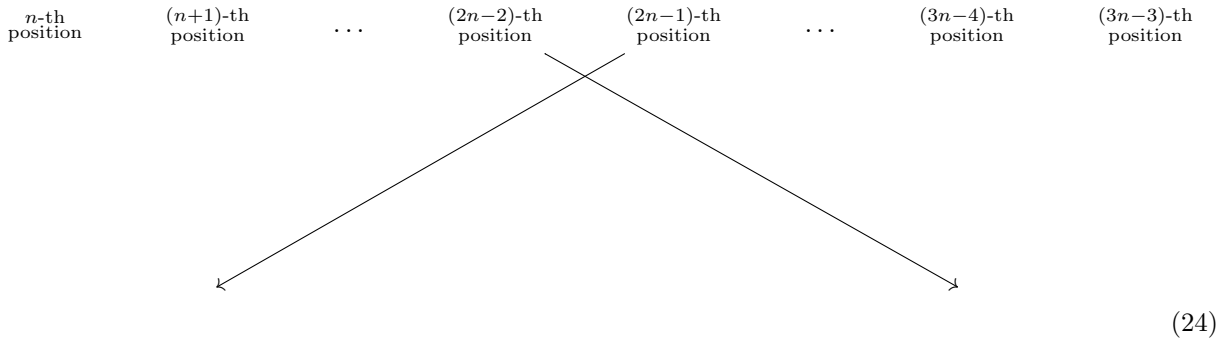
\begin{align}\label{inter-diagram}
 \xymatrix{
   \substack{n\text{-th} \\ \text{position}}  & \substack{(n+1)\text{-th} \\ \text{position}} & \ldots & \substack{(2n-2)\text{-th} \\ \text{position}} \ar[rrrddd]  & \substack{(2n-1)\text{-th} \\ \text{position}} \ar[dddlll] & \ldots &  \substack{(3n-4)\text{-th} \\ \text{position}}& \substack{(3n-3)\text{-th} \\ \text{position}} \\
     & & & & & & & \\
     & & & & & & & \\
     & & & & & & & 
    }
\end{align}
After this transformation, on the first $(n+1)$ factors, one may apply the $n$-Yang-Baxter equation. Soon after, from the $3$-rd to $(n+3)$-th factor, one may again apply the $n$-Yang-Baxter equation. Immediately after this, one may apply again the $n$-Yang-Baxter equation to the $5$-th to $(n+5)$-th factor. This process continues until we apply the $n$-Yang-Baxter equation to the $(2n-3)$-th factor to $3 (n-1)$-th factor. After completing this, we apply the inverse process of (\ref{inter-diagram}) to the first $2(n-1)$ factors, which can be described by the following:

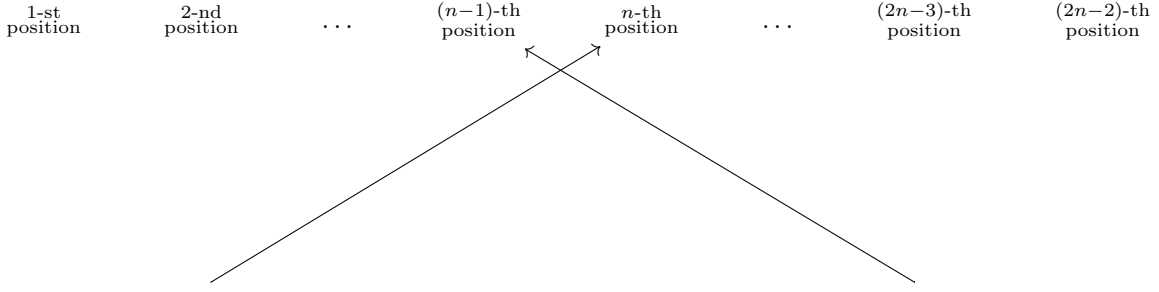
\begin{align*}
 \xymatrix{
   \substack{1\text{-st} \\ \text{position}}  & \substack{2\text{-nd} \\ \text{position}} & \ldots & \substack{(n-1)\text{-th} \\ \text{position}}   & \substack{n\text{-th} \\ \text{position}} & \ldots &  \substack{(2n-3)\text{-th} \\ \text{position}}& \substack{(2n-2)\text{-th} \\ \text{position}} \\
     & & & & & & & \\
     & & & & & & & \\
     & \ar[rrruuu] & & & & & \ar[llluuu] & \\
    }
\end{align*}
One can easily see that the final expression coincides with $(\mathrm{Id}^{\otimes (n-1)} \otimes \widetilde{S}) (\widetilde{S} \otimes \mathrm{Id}^{\otimes (n-1)}) (\mathrm{Id}^{\otimes (n-1)} \otimes \widetilde{S})$. Hence, the result follows.
\end{proof}

Let $(\mathcal{L}, [-, \ldots, - ], {\bf 1})$ be a central $n$-Leibniz algebra. By Theorem \ref{central-n-leibniz-nyb}, we obtain an $n$-Yang-Baxter operator $S : \mathcal{L}^{\otimes n} \rightarrow \mathcal{L}^{\otimes n}$, and hence a Yang-Baxter operator $\widetilde{S} : \mathcal{L}^{\otimes (n-1)} \otimes \mathcal{L}^{\otimes (n-1)} \rightarrow \mathcal{L}^{\otimes (n-1)} \otimes \mathcal{L}^{\otimes (n-1)}$ on the vector space $\mathcal{L}^{\otimes (n-1)}$. The explicit form of $\widetilde{S}$ is given by
\begin{align*}
    &\widetilde{S} (x_1 \otimes \cdots \otimes x_{n-1} \otimes y_1 \otimes \cdots \otimes y_{n-1}) \\
    &= (S \otimes \mathrm{Id}^{\otimes (n-2)}) \cdots ( \mathrm{Id}^{\otimes (n-3)} \otimes S \otimes \mathrm{Id}) (\mathrm{Id}^{\otimes (n-2)} \otimes S) (x_1 \otimes \cdots \otimes x_{n-1} \otimes y_1 \otimes \cdots \otimes y_{n-1}) \\
    &= (S \otimes \mathrm{Id}^{\otimes (n-2)}) \cdots ( \mathrm{Id}^{\otimes (n-3)} \otimes S \otimes \mathrm{Id}) \big( x_1 \otimes \cdots \otimes x_{n-2} \otimes y_1 \otimes \cdots \otimes y_{n-1} \otimes x_{n-1} \\
    & \qquad \qquad \qquad \qquad \qquad \qquad \quad  \qquad \qquad \qquad + x_1 \otimes \cdots \otimes x_{n-2} \otimes \underbrace{ {\bf 1} \otimes \cdots \otimes {\bf 1}}_{n-2 \text{~times}} \otimes [x_{n-1}, y_1, \ldots, y_{n-1}] \big) \\
    &= y_1 \otimes \cdots \otimes y_{n-1} \otimes x_1 \otimes \cdots \otimes x_{n-1} ~+~ \underbrace{ {\bf 1} \otimes \cdots \otimes {\bf 1}}_{n-1 \text{~times}} \otimes \big(    \sum_{i=1}^{n-1} x_1 \otimes \cdots \otimes [x_i, y_1, \ldots, y_{n-1}] \otimes \cdots \otimes x_{n-1}  \big).
\end{align*}
The final expression shows that the Yang-Baxter operator $\widetilde{S}$ is simply the one obtained from the central Leibniz algebra $(\mathcal{L}^{\otimes (n-1)},  \{ -, - \}, {\bf 1}^{\otimes (n-1)})$. Thus, the following diagram also commutes
\begin{align*}
 \xymatrix{
     & \substack{n\text{-Yang-Baxter operator} \\ S : \mathcal{L}^{\otimes n} \rightarrow \mathcal{L}^{\otimes n} } \ar[rr]^{\text{Theorem~} \ref{theorem-19}} &  & \text{Yang-Baxter operator } \ar@{=}[dd]\\
     \substack{\text{central }n\text{-Leibniz} \\ (\mathcal{L}, [-, \ldots, -], {\bf 1}) } \ar@/^0.8pc/[ru]^{\text{Thm } \ref{central-n-leibniz-nyb}}   \ar@/_0.8pc/[rd]_{\text{Prop } \ref{funda-leibniz}} &  &  & \\
      & \substack{\text{central Leibniz}\\ (   \mathcal{L}^{\otimes n-1 }, \{ -, - \}  , {\bf 1}^{\otimes n-1}) } \ar[rr]_{\mathrm{ Eq. ~} (\ref{lebed-exp})}   & & \text{Yang-Baxter operator.}
    }
\end{align*}

\medskip

In the following, we shall consider the set-theoretical $n$-Yang-Baxter equation. Let $X$ be a nonempty set. For any $n \geq 2$, a set-map $s : X^{\times n} \rightarrow X^{\times n}$ is said to be a {\bf set-theoretical solution of the $n$-Yang-Baxter equation} (also called a {\bf set-theoretical $n$-solution} in short) if $s$ is bijective and satisfies
\begin{align}\label{set-nybe}
     (s \times \mathrm{Id}^{\times (n-1)} )& \cdots (  \mathrm{Id}^{\times (n-2)} \times s \times \mathrm{Id} ) (\mathrm{Id}^{\times (n-1)} \times s) (s \times \mathrm{Id}^{\times (n-1)}) \\
     &= (\mathrm{Id}^{\times (n-1)} \times s) (s \times \mathrm{Id}^{\times (n-1)}) \cdots (\mathrm{Id}^{\times (n-2)} \times s \times \mathrm{Id}) ( \mathrm{Id}^{\times (n-1)} \times s). \nonumber
\end{align}

\begin{exam}
    Let $X$ be a nonempty set. Then the identity map $\mathrm{Id}_{X^{\times n}} : X^{\times n} \rightarrow X^{\times n}$ is a set-theoretical $n$-solution. The map
    \begin{align*}
        s : X^{\times n} \rightarrow X^{\times n} ~~~ \text{ given by }  ~~~ s (x_1, \ldots, x_n) = (x_2, \ldots, x_n, x_1)
    \end{align*}
    is also a set-theoretical $n$-solution, called the {\em flip}. More generally, let $f_1, \ldots, f_n : X \rightarrow X$ be any arbitrary functions on $X$. Then the set-map $s : X^{\times n} \rightarrow X^{\times n}$ defined by $s(x_1, \ldots, x_n) = ( f_2 (x_2), \ldots, f_n (x_n), f_1 (x_1))$ is a set-theoretical $n$-solution if and only if $f_1$ commues with each $f_i$, that is, $f_1 f_i = f_i f_1$, for all $i =1, \ldots , n$.
\end{exam}

\begin{exam}
    Let $G$ be a group with the neutral element $1$. Then the maps $s_1, s_2 : G^{\times 3} \rightarrow G^{\times 3}$ defined by $s_1 (g, h, k)= (1, 1, ghk)$ and $s_2 (g, h, k) = (1,1, gh^{-1} k)$ both are set-theoretical $3$-solutions.
\end{exam}

\begin{thm}\label{diag-thm}
    Let $X$ be a nonempty set and $\langle -, \ldots, - \rangle : X \times \cdots \times X \rightarrow X$ be an $n$-ary operation on $X$. Define a map $s : X^{\times n} \rightarrow X^{\times n}$ by 
    \begin{align*}
        s (x_1, x_2, \ldots, x_n ) =   (x_2, \ldots, x_n, \langle x_1, x_2, \ldots, x_n \rangle), \text{ for } (x_1, \ldots, x_n ) \in X^{\times n}.
    \end{align*}
    Then $(X, \langle -, \ldots, - \rangle)$ is an $n$-rack if and only if $s$ is a set-theoretical $n$-solution.
\end{thm}

\begin{proof}
    For any $x_1, \ldots, x_n, y_1, \ldots, y_{n-1} \in X$, we observe that 
     \begin{align*}
        & (s \times \mathrm{Id}^{\times (n-1)} ) \cdots (\mathrm{Id}^{\times (n-2)} \times s \times \mathrm{Id})(\mathrm{Id}^{\times (n-1)} \times s) (s \times \mathrm{Id}^{\times (n-1)}) (x_1, \ldots, x_n, y_1, \ldots, y_{n-1})\\
        &= (s \times \mathrm{Id}^{\times (n-1)} ) \cdots (\mathrm{Id}^{\times (n-2)} \times s \times \mathrm{Id})(\mathrm{Id}^{\times (n-1)} \times s) (x_2, \ldots, x_n, \langle x_1, \ldots, x_n \rangle, y_1, \ldots, y_{n-1} ) \\
        &= (s \times \mathrm{Id}^{\times (n-1)} ) \big(  x_2, y_1, \ldots, y_{n-1} , \langle x_3, y_1, \ldots, y_{n-1} \rangle, \ldots, \langle x_n, y_1, \ldots, y_{n-1} \rangle,  \langle \langle x_1, \ldots, x_n \rangle, y_1, \ldots, y_{n-1} \rangle  \big)\\
        &= \big( y_1, \ldots, y_{n-1}, \langle x_2, y_1, \ldots, y_{n-1} \rangle ,  \ldots, \langle x_n, y_1, \ldots, y_{n-1} \rangle, \langle \langle x_1, \ldots, x_n \rangle, y_1, \ldots, y_{n-1} \rangle    \big).
    \end{align*}
    On the other hand,
    \begin{align*}
       & (\mathrm{Id}^{\times (n-1)} \times s) (s \times \mathrm{Id}^{\times (n-1)}) \cdots (\mathrm{Id}^{\times (n-2)} \times s \times \mathrm{Id}) ( \mathrm{Id}^{\times (n-1)} \times s ) (x_1, \ldots, x_n, y_1, \ldots, y_{n-1}) \\
        &= (\mathrm{Id}^{\times (n-1)} \times s) (y_1, \ldots, y_{n-1}, \langle x_1, y_1, \ldots, y_{n-1} \rangle, \ldots, \langle x_n, y_1, \ldots, y_{n-1} \rangle ) \\
        &= \big( y_1, \ldots, y_{n-1}, \langle x_2, y_1, \ldots, y_{n-1} \rangle, \ldots, \langle x_n, y_1, \ldots, y_{n-1} \rangle, \langle \langle x_1, y_1, \ldots, y_{n-1}\rangle, \ldots,  \langle x_n, y_1, \ldots, y_{n-1}\rangle \rangle \big).
    \end{align*}
    This shows that $s$ satisfies the equation (\ref{set-nybe}) if and only if the $n$-ary operation $\langle -, \ldots, - \rangle$ satisfies the self-distributivity (\ref{n-shelf}). Finally, the map $s$ is bijective if and only if for all $y_1, \ldots, y_{n-1} \in X$, the translation map $\langle - , y_1, \ldots, y_{n-1} \rangle : X \rightarrow X$ is bijective. This concludes the proof.
\end{proof}

\begin{remark}\label{remark-yb-nyb-vv}
    It is important to note that the set-theoretical versions of Theorem \ref{them-soln-nsoln} and Theorem \ref{theorem-19} are also valid. More precisely, let $X$ be a nonempty set and $r : X \times X \rightarrow X \times X$ be a set-theoretical solution on $X$. Then 
    \begin{align*}
        s_n :=    (\mathrm{Id}^{\times (n-2)} \times r) \cdots (\mathrm{Id} \times r \times \mathrm{Id}^{\times (n-3)}) (r \times \mathrm{Id}^{\times (n-2)})   : X^{\times n} \rightarrow X^{\times n}
    \end{align*}
    is a set-theoretical $n$-solution on $X$. On the other hand, if $s : X^{\times n} \rightarrow X^{\times n}$ is any set-theoretical $n$-solution then the map
    \begin{align*}
        \widetilde{s} =  (s \times \mathrm{Id}^{\times (n-2)}) \cdots  ( \mathrm{Id}^{\times (n-3)} \times s \times \mathrm{Id}) (\mathrm{Id}^{\times (n-2)} \times s)   : X^{\times (n-1)} \times  X^{\times (n-1)} \rightarrow X^{\times (n-1)} \times  X^{\times (n-1)}
    \end{align*}
    is a set-theoretical solution on the set $X^{\times (n-1)}$.
\end{remark}

%\begin{thm}
 %   Let $s: X \times X \rightarrow X \times X$ be a set-theoretical solution of the Yang-Baxter equation. Define a map $\widetilde{s} : X^{\times n} \rightarrow  X^{\times n} $ by
%    \begin{align*}
 %       \widetilde{s} = (\mathrm{Id}^{\times (n-2)} \times s) \cdots (\mathrm{Id} \times s \times \mathrm{Id}^{\times (n-3)}) (s \times \mathrm{Id}^{\times (n-2)}).
 %   \end{align*}
 %   Then $\widetilde{s}$ is a set-theoretical solution of the $n$-Yang-Baxter equation.
%\end{thm}

Let $(X, \triangleleft)$ be a rack. Then it is known that the map $r: X \times X \rightarrow X \times X$ given by $r (x, y) = (y, x \triangleleft y)$ is a set-theoretical solution. Therefore, by Remark \ref{remark-yb-nyb-vv}, one may construct a set-theoretical $n$-solution $s_n: X^{\times n} \rightarrow X^{\times n}$, explicitly given by
\begin{align*}
    s_n (x_1, \ldots, x_n) =~& (\mathrm{Id}^{\times (n-2)} \times r) \cdots (\mathrm{Id} \times r \times \mathrm{Id}^{\times (n-3)}) (r \times \mathrm{Id}^{\times (n-2)}) (x_1, \ldots, x_n) \\
    =~& (\mathrm{Id}^{\times (n-2)} \times r) \cdots (\mathrm{Id} \times r \times \mathrm{Id}^{\times (n-3)}) (x_2, x_1 \triangleleft x_2, x_3, \ldots, x_n) \\
    =~& \big(  x_2, \ldots, x_n ,~\!  ( \cdots   ((x_1 \triangleleft x_2)    \triangleleft x_3)  \cdots ) \triangleleft x_n   \big) 
    = \big(  x_2, \ldots, x_n , \langle x_1 , \ldots, x_n \rangle  \big) \quad (\text{cf. Exam } \ref{exam-rack-to-nrack}).
\end{align*}
Hence this set-theoretical $n$-solution is same with the one obtained from the induced $n$-rack $(X, \langle -, \ldots , - \rangle)$, see Example \ref{exam-rack-to-nrack}. Thus, we obtain the following commutative diagram

\begin{align}
    \xymatrix{
     & \substack{\text{solution} \\ r : X \times X \rightarrow X \times X } \ar[rr]^{\text{Remark~} \ref{remark-yb-nyb-vv}} &  & n\text{-solution } \ar@{=}[dd]\\
    \substack{ \text{rack} \\ (X, \triangleleft)  }  \ar@/^0.8pc/[ru] \ar@/_0.8pc/[rd]_{\text{Example } \ref{exam-rack-to-nrack}} &  &  & \\
      & \substack{n\text{-rack}\\ (X, \langle -, \ldots , - \rangle) } \ar[rr]_{\text{Theorem~}\ref{diag-thm}}   & & n\text{-solution.}
    }
\end{align}

\medskip

\medskip

On the other hand, suppose we start with an $n$-rack $(X, \langle -, \ldots, - \rangle)$. Then by Theorem \ref{diag-thm}, one obtains a set-theoretical $n$-solution $s : X^{\times n} \rightarrow X^{\times n}$ on the set $X$. This yields a set-theoretical solution $\widetilde{s} : X^{\times (n-1)} \times X^{\times (n-1)} \rightarrow X^{\times (n-1)} \times X^{\times (n-1)}$ on the set $X^{\times (n-1)}$, given by (see Remark \ref{remark-yb-nyb-vv})
\begin{align*}
   & \widetilde{s} (x_1, \ldots, x_{n-1}, y_1, \ldots, y_{n-1}) \\
    &= (s \times \mathrm{Id}^{\times (n-2)}) \cdots (\mathrm{Id}^{\times (n-3)} \times s \times \mathrm{Id}) (\mathrm{Id}^{\times (n-2)} \times s) (x_1, \ldots, x_{n-1}, y_1, \ldots, y_{n-1}) \\
    &=  (s \times \mathrm{Id}^{\times (n-2)}) \cdots (\mathrm{Id}^{\times (n-3)} \times s \times \mathrm{Id})  (x_1, \ldots, x_{n-2}, y_1, \ldots, y_{n-1}, \langle x_{n-1} , y_1, \ldots, y_{n-1} \rangle ) \\
    &= \big(  y_1, \ldots, y_{n-1} , \langle x_1, y_1, \ldots, y_{n-1} \rangle, \ldots, \langle x_{n-1}, y_1, \ldots, y_{n-1} \rangle  \big) \\
    &= \big( (y_1, \ldots, y_{n-1}) ~ \! , ~ \! (x_1, \ldots, x_{n-1}) \triangleleft_{\langle - , \ldots, - \rangle} (y_1, \ldots, y_{n-1}) \big) \quad (\text{following Example }\ref{prop-nrack-rack}).
\end{align*}
It follows that this set-theoretical solution on $X^{\times (n-1)}$ is precisely the one obtained from the induced rack $(X^{\times (n-1)}, \triangleleft_{\langle - , \ldots, - \rangle})$ given in Proposition \ref{prop-nrack-rack}. As a result, the following diagram also commutes:
\begin{align*}
 \xymatrix{
     & \substack{n\text{-solution} \\ s : X^{\times n} \rightarrow X^{\times n} } \ar[rr]^{\text{Remark~} \ref{remark-yb-nyb-vv}} &  & \text{solution } \ar@{=}[dd]\\
     \substack{n\text{-rack~} \\ (X, \langle - , \ldots, - \rangle) } \ar@/^0.8pc/[ru] \ar@/_0.8pc/[rd]_{\text{Example } \ref{exam-rack-to-nrack}} &  &  & \\
      & \substack{\text{rack}\\ (X^{\times (n-1)}, \triangleleft_{\langle -, \ldots , - \rangle } ) } \ar[rr]   & & \text{solution.}
    }
\end{align*}

\begin{remark} It has been noted that the $n$-Yang-Baxter equation (\ref{nybe}) is closely related to (right) $n$-Leibniz algebras and (right) linear $n$-racks. In the same way, the set-theoretical $n$-Yang-Baxter equation (\ref{set-nybe}) is associated with (right) $n$-racks. Therefore, the equations (\ref{nybe}) and (\ref{set-nybe}) can be regarded as the right $n$-Yang-Baxter equation and the set-theoretical right $n$-Yang-Baxter equation, respectively. However, to study left $n$-Leibniz algebras and left (linear) $n$-racks, we need to consider the following variants of the equations (\ref{nybe}) and (\ref{set-nybe}). For a given vector space $V$, a linear map $S : V^{\otimes n} \rightarrow V^{\otimes n}$ is a left pre-$n$-Yang-Baxter operator if it satisfies
\begin{align}\label{nybe-left}
    (S \otimes \mathrm{Id}^{\otimes (n-1)}) & ( \mathrm{Id}^{\otimes (n-1)} \otimes S) \cdots (\mathrm{Id} \otimes S \otimes \mathrm{Id}^{\otimes (n-2)})  (S \otimes \mathrm{Id}^{\otimes (n-1)})  \\
     &= ( \mathrm{Id}^{\otimes (n-1)} \otimes S) \cdots  ( \mathrm{Id} \otimes S \otimes \mathrm{Id}^{\otimes (n-2)}) (S \otimes \mathrm{Id}^{\otimes (n-1)}) ( \mathrm{Id}^{\otimes (n-1)} \otimes S). \nonumber
\end{align}
As before, we skip the adjective `pre' when $S$ is additionally invertible. For example, if a linear map $S : V^{\otimes 3} \rightarrow V^{\otimes 3}$ is represented by the diagram ~~~ \!  \!
\begin{tikzpicture}[scale=0.45]
  \draw[thick, purple]
  (0,3) .. controls (0,1.7) and (1,2.3) .. (1,1)
  (1,3) .. controls (1,1.7) and (2,2.3) .. (2,1)
  (0,1) .. controls (0,2.2) and (2.1,1.8) .. (2,3);
 % \node at (1,-1) {Left};
    \end{tikzpicture}, then the equation (\ref{nybe-left}) for $n=3$ can be described by the equality

    \medskip

    
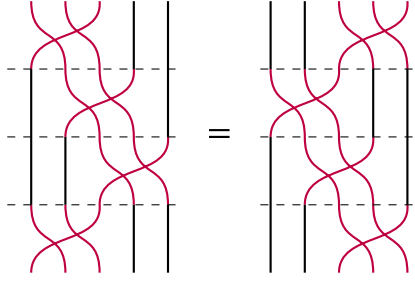
\begin{center}
\begin{tikzpicture}[scale=0.45]
    \draw[thick, purple]
  (0,3) .. controls (0,1.7) and (1,2.3) .. (1,1)
  (1,3) .. controls (1,1.7) and (2,2.3) .. (2,1)
  (0,1) .. controls (0,2.2) and (2.1,1.8) .. (2,3)
  (1,1) .. controls (1,-.3) and (2,0.3) .. (2,-1)
  (2,1) .. controls (2,-.3) and (3,.3) .. (3,-1)
  (1,-1) .. controls (1,.2) and (3.1,-.2) .. (3,1)
  (2,-1) .. controls (2,-2.3) and (3,-1.7) .. (3,-3)
  (3,-1) .. controls (3,-2.3) and (4,-1.7) .. (4,-3)
  (2,-3) .. controls (2,-1.8) and (4.1,-2.2) .. (4,-1)
  (0,-3) .. controls (0,-4.3) and (1,-3.7) .. (1,-5)
  (1,-3) .. controls (1,-4.3) and (2,-3.7) .. (2,-5)
  (0,-5) .. controls (0,-3.8) and (2.1,-4.2) .. (2,-3)%%
  (9,3) .. controls (9,1.7) and (10,2.3) .. (10,1)
  (10,3) .. controls (10,1.7) and (11,2.3) .. (11,1)
  (9,1) .. controls (9,2.2) and (11.1,1.8) .. (11,3)
  (7,1) .. controls (7,-0.3) and (8,.3) .. (8,-1)
  (8,1) .. controls (8,-0.3) and (9,.3) .. (9,-1)
  (7,-1) .. controls (7,.2) and (9.1,-0.2) .. (9,1)
  (8,-1) .. controls (8,-2.3) and (9,-1.7) .. (9,-3)
  (9,-1) .. controls (9,-2.3) and (10,-1.7) .. (10,-3)
  (8,-3) .. controls (8,-1.8) and (10.1,-2.2) .. (10,-1)
  (9,-3) .. controls (9,-4.3) and (10,-3.7) .. (10,-5)
  (10,-3) .. controls (10,-4.3) and (11,-3.7) .. (11,-5)
  (9,-5) .. controls (9,-3.8) and (11.1,-4.2) .. (11,-3);
  \draw[thick]
  (3,3) -- (3,1)
  (4,3) -- (4,1)
  (4,1) -- (4,-1)
  (0,1) -- (0,-1)
  (0,-1) -- (0,-3)
  (1,-1) -- (1,-3)
  (3,-3) -- (3,-5)
  (4,-3) -- (4,-5)
  (7,3) -- (7,1)
  (8,3) -- (8,1)
  (10,1) -- (10,-1)
  (11,1) -- (11,-1)
  (7,-1) -- (7,-3)
   (11,-1) -- (11,-3)
   (7,-3) -- (7,-5)
   (8,-3) -- (8,-5)
   (5.2, -1) -- (5.8,-1)
   (5.2, -0.8) -- (5.8,-0.8);
  \draw[thin, dashed]
  (-0.7,1) -- (4.3,1)
  (-0.7,-1) -- (4.3,-1)
   (-0.7,-3) -- (4.3,-3)
   
   (6.7,1) -- (11.3,1)
   (6.7,-1) -- (11.3,-1)
   (6.7,-3) -- (11.3,-3);
%   \node at (5.5,-1) {$\xlongequal{\hspace{8pt}}$};
  % \node at (4,-7) {Left};
\end{tikzpicture}
\end{center}
It is not hard to see that a linear map $S: V^{\otimes n} \rightarrow V^{\otimes n}$ is a right (pre-)$n$-Yang-Baxter operator if and only if $\tau \circ S \circ \tau$ is a left (pre-)$n$-Yang-Baxter operator, where $\tau: V^{\otimes n} \rightarrow V^{\otimes n}$ is the reverse map $\tau (v_1 \otimes \cdots \otimes v_n) = v_n \otimes \cdots \otimes v_1$. It turns out that for a central left $n$-Leibniz algebra $(\mathcal{L}, [-, \ldots, - ], {\bf 1})$, the map $S_l : \mathcal{L}^{\otimes n} \rightarrow \mathcal{L}^{\otimes n}$ given by
\begin{align*}
    S_l (x_1 \otimes \cdots \otimes x_n) = x_n \otimes x_1 \otimes \cdots \otimes x_{n-1}  ~+~ [x_1, \ldots, x_n] \otimes \underbrace{{\bf 1} \otimes \cdots \otimes {\bf 1}}_{n-1 \text{~times}}
\end{align*}
is a left $n$-Yang-Baxter operator on the vector space $\mathcal{L}$.

\medskip

Similarly, in the set-theoretical context, a map $s: X^{\times n} \rightarrow X^{\times n}$ is a set-theoretical solution of the left $n$-Yang-Baxter equation if $s$ is bijective and satisfies
\begin{align*}
     (s \times \mathrm{Id}^{\times (n-1)}) & ( \mathrm{Id}^{\times (n-1)} \times s) \cdots (\mathrm{Id} \times s \times \mathrm{Id}^{\times (n-2)})  (s \times \mathrm{Id}^{\times (n-1)})  \\
     &= ( \mathrm{Id}^{\times (n-1)} \times s)  \cdots  (\mathrm{Id} \times s \times \mathrm{Id}^{\times (n-2)}) (s \times \mathrm{Id}^{\times (n-1)}) ( \mathrm{Id}^{\times (n-1)} \times s).
\end{align*}
For a left $n$-rack $(X, \langle - , \ldots , - \rangle)$, the map $s_l : X^{\times n} \rightarrow X^{\times n}$ defined by
\begin{align*}
    s_l (x_1, \ldots, x_n) = (\langle x_1, \ldots, x_n \rangle, x_1, \ldots, x_{n-1})
\end{align*}
is a set-theoretical solution of the left $n$-Yang-Baxter equation.
\end{remark}

%\begin{pspicture}(7,7)(-1,-1)
%\begin{center}
  
 %   \end{center}
%\end{pspicture}

\noindent {\bf Further discussions.} 

In this paper, we developed a higher analogue of the Yang-Baxter operator in connection with $n$-ary structures (such as $n$-Leibniz algebras, linear $n$-racks and $n$-racks). It is well-known that Yang-Baxter operators are equivalent to quantum Yang-Baxter operators (also called $R$-matrices) that are useful in the study of quantum groups, Hopf algebras, quantum field theory, etc. Thus, it remains to find the quantum analogue of the $n$-Yang-Baxter equation and its importance in mathematics and mathematical physics. We will come back with this question soon. On the other hand, the notion of the tetrahedron equation was introduced by Zamolodchikov \cite{zamo1,zamo2} as a three-dimensional version of the Yang-Baxter equation. In his equation, a three-dimensional tetrahedron plays a crucial role, and this equation has found important applications in scattering theory. It could be interesting to find at least some connections between the tetrahedron equation and the $3$-Yang-Baxter equation introduced in the present paper.

\medskip

It has been mentioned in the introduction that racks, cycle sets and (skew) braces are algebraic structures closely related to set-theoretical solutions. For example, nondegenerate cycle sets are bijectively correspond to nondegenerate set-theoretical solutions. On the other hand, any brace yields a nondegenerate involutive set-theoretical solution and a skew brace (noncommutative analogue of a brace) provides a non-degenerate (not necessarily involutive) set-theoretical solution. Let $s : X^{\times 3} \rightarrow X^{\times 3}$ be a set-theoretical $3$-solution that we write as
\begin{align*}
    s (x, y, z)= \big(   \sigma_{x, z} (y) ~ \! , ~ \! \tau_{x, y} (z)~\! , ~ \! \eta_{y, z} (x)  \big), \text{ for any } x, y, z \in X.
\end{align*}
We say that $s$ is left (resp. right, middle) nondegenerate if the map $\tau_{x, y}: X \rightarrow X$ (resp. $\eta_{x, y}$, $\sigma_{x, y}$) is bijective, for all $x, y \in X$. It is said to be nondegenerate if it is left, right and middle nondegenerate. The map $s$ is said to be $3$-involutive if $s^3 = \mathrm{Id}_{X^{\times 3}}$. It follows that the set-theoretical $3$-solution obtained from a $3$-rack $(X, \langle -, -, - \rangle)$ is nondegenerate, which need not be $3$-involutive in general. However, the flip $3$-solution $s : X^{\otimes 3} \rightarrow X^{\times 3}$ given by $s (x, y, z) = (y, z, x)$ is both nondegenerate and $3$-involutive. In a subsequent work, we plan to study nondegenerate and $3$-involutive set-theoretical $3$-solutions, and related algebraic structures. In this direction, we will explore the possible notion of $3$-cycle sets and their relations with cycle sets.

\medskip
    
\noindent  {\bf Acknowledgements.} The authors would like to thank the Department of Mathematics, IIT Kharagpur, for providing the beautiful academic atmosphere where the research has been conducted. S. Majhi thanks UGC for the PhD research fellowship.

\medskip

\noindent{\bf Funding.} There is no funding.

\medskip

\noindent {\bf Data Availability Statement.} Data sharing does not apply to this article as no new data were created or analyzed in this study.

\end{document}